\documentclass[10pt]{article}

\usepackage[active]{srcltx}

 \usepackage{graphics,color}

\usepackage{amsmath}
\usepackage{amsfonts}
\usepackage{amssymb}
\usepackage{amscd}
\newcounter{theorem}
\renewcommand{\thetheorem}{\arabic{section}.\arabic{theorem}}

\newenvironment{thm}[1]{\par
\begin{sloppypar}\refstepcounter{theorem}%
\noindent{\bf #1 \thetheorem.}\it{}}{\end{sloppypar}}

\newenvironment{theorem}{\begin{thm}{Theorem}}{\end{thm}}
\newenvironment{proposition}{\begin{thm}{Proposition}}{\end{thm}}
\newenvironment{corollary}{\begin{thm}{Corollary}}{\end{thm}}
\newenvironment{lemma}{\begin{thm}{Lemma}}{\end{thm}}

\newenvironment{defi}[1]{\par
\begin{sloppypar}\refstepcounter{theorem}%
\noindent{\bf #1 \thetheorem.}\rm{}}{\end{sloppypar}}
\newenvironment{definition}{\begin{defi}{Definition}}{\end{defi}}

\newenvironment{remark}{\begin{defi}{Remark}}{\end{defi}}
\newenvironment{hypothesis}{\begin{defi}{Hypothesis}}{\end{defi}}

\newcommand{\eh}{\hfill}\newlength{\sperr}
\newenvironment{proof}{{\settowidth{\sperr}{\rm Proof}
\par\addvspace{0.3cm}\noindent\parbox[t]{1.3\sperr}{\rm P\eh r\eh o\eh o\eh
f\eh.}}}{\nopagebreak\mbox{}\hfill $\blacksquare $\par\addvspace{0.25cm}}

\newenvironment{proofP32}{{\settowidth{\sperr}{\rm Proof of Proposition a.b}
\par\addvspace{0.3cm}\noindent\parbox[t]{1.3\sperr}{\rm P\eh r\eh o\eh o\eh
f\eh\eh o\eh f\eh\eh P\eh r\eh o\eh p\eh o\eh s\eh i\eh t\eh i\eh o\eh n\eh\eh\ref{P3.2}.}}}{\nopagebreak\mbox{}\hfill $\blacksquare $\par\addvspace{0.25cm}}

\numberwithin{equation}{section}
\numberwithin{theorem}{section}

\def\R{{\rm I\kern-.2em R}}

\def\X{\mathcal X}
\def\Rd{\mathbb{R}^d}
\def\Rn{\mathbb{R}^n}
\def\Rm{\mathbb{R}^m}
\def\d{{\rm d}}
\def\dbar{\;\;\bar{}\!\!\!d}

\def\Tr{\text{\sf T\!r}}

\begin{document}

\title{Magnetic Fourier Integral Operators\\
\textit{extended version}}

\date{\today}

\author{Viorel Iftimie\footnote{Institute
of Mathematics Simion Stoilow of the Romanian Academy, P.O.  Box
1-764, Bucharest, Romania.} \hspace{2pt}and Radu
Purice\footnote{Institute
of Mathematics Simion Stoilow of the Romanian Academy, P.O.  Box
1-764, Bucharest, Romania.}
\footnote{Laboratoire Europ\'een Associ\'e CNRS Franco-Roumain {\it Math-Mode}}}

\maketitle

\begin{abstract}
In some previous papers we have defined and studied a 'magnetic' pseudodifferential calculus as a gauge covariant generalization of the Weyl calculus when a magnetic field is present. In this paper we extend the standard Fourier Integral Operators Theory to the case with a magnetic field, proving composition theorems, continuity theorems in 'magnetic' Sobolev spaces and Egorov type theorems. The main application is the representation of the evolution group generated by a 1-st order 'magnetic' pseudodifferential operator (in particular the relativistic Schr\"{o}dinger operator with magnetic field) as such a 'magnetic' Fourier Integral Operator. As a consequence of this representation we obtain some estimations for the distribution kernel of this evolution group and a result on the propagation of singularities.
\end{abstract}

%...........................................................................................
\section{Introduction}

In a series of papers \cite{KO1,KO2,Mu,MP} a gauge covariant formalism has been proposed for associating to any \textit{classical observable} $a$ (a certain function on the {\it phase space} $\mathbb{R}^{2d}$, usually a standard symbol from $S^m(\mathbb{R}^d)$) and any magnetic field $B=dA$ (with $A$ a 1-form on $\mathbb{R}^d$ with $C^\infty(\mathbb{R}^d)$ coefficients) a \textit{quantum observable} $\mathfrak{Op}^A(a)$, defined as a 'magnetic' pseudodifferential operator acting on $\mathcal{S}(\mathbb{R}^d)$ in the following way
\begin{equation}\label{0.1}
 \left[\mathfrak{Op}^A(a)u\right](x):=\int_{\mathbb{R}^{2d}}e^{i<x-y,\eta>}e^{-i\int_{[x,y]}A}a\left(\frac{x+y}{2},\eta\right)u(y)\,dy\,\dbar\eta,\quad\forall u\in\mathcal{S}(\mathbb{R}^d),\,\forall x\in\mathbb{R}^d,
\end{equation}
with $\dbar\eta:=(2\pi)^{-d}d\eta$.
The properties of operators of the form \eqref{0.1} have been studied in \cite{IMP1}, where we also gave some applications. In particular we have proved that for any real elliptic symbol $a\in S^1(\mathbb{R}^d)$, the operator $\mathfrak{Op}^A(a)$ has a self-adjoint extension $P$ on $L^2(\mathbb{R}^d)$. In particular, the operator $\mathfrak{Op}^A(<\xi>)$ with $<\xi>:=\sqrt{1+|\xi|^2}$ for any $\xi\in\mathbb{R}^d$ can be considered as the relativistic Schr\"{o}dinger operator with magnetic field. Let us remark that $\mathfrak{Op}^A(<\xi>^2)-1$ is the usual non-relativistic Schr\"{o}dinger operator with magnetic field. Let us also remark that the usual choice for the relativistic Schr\"{o}dinger operator with magnetic field $\sqrt{\mathfrak{Op}^A(<\xi>^2)}$, although gauge covariant, has not a reasonable classical counterpart, its symbol being rather complicated.

An important problem is the study of the unitary group $\{e^{-itP}\}_{t\in\mathbb{R}}$ generated by the self-adjoint operator $P$, and for such a study an integral representation may be very useful. Having in mind other versions of pseudodifferential operators we expect that at least for $|t|$ small the operator $e^{-itP}$ should be a 'Fourier Integral Operator', where the definition of such an object should depend on the magnetic field $B$ and should be gauge covariant. In fact such a representation has been obtained in a former paper \cite{IMP2}, where we have proved that  $e^{-itP}$ verifies the hypothesis of a very implicit definition of Fourier Integral Operators based on commutation properties (in the spirit of Bony \cite{Bo}). For applications, an explicit integral representation is needed and this is the object of the present paper.

Let us notice first that Proposition \ref{P1.14} allows us to replace $a\left(\frac{x+y}{2},\eta\right)$ in \eqref{0.1} with $b(x,\eta)$ with $b$ another symbol of the same order as $a$. Taking this fact into account and the standard definition of Fourier Integral Operators, it is natural that to a triple $(a,B,\Phi)$ with $a\in S^m(\mathbb{R}^d)$, $B=dA$ a magnetic field and $\Phi:\mathbb{T}^*\mathbb{R}^d\rightarrow \mathbb{T}^*\mathbb{R}^d$ a symplectomorphism with generating function $U$, we shall associate a linear operator
\begin{equation}\label{0.2}
\left[\mathfrak{Op}^A_\Phi(a)u\right](x):=\int_{\mathbb{R}^d}e^{i(U(x,\eta)-<y,\eta>}e^{-i\int_{[x,y]}A}\ a\left(x,\eta\right)u(y)\,dy\,\dbar\eta,\quad\forall u\in\mathcal{S}(\mathbb{R}^d),\,\forall x\in\mathbb{R}^d,
\end{equation}
that we shall call a magnetic Fourier Integral Operator.

Section 2 brings only a few new results, as for example Proposition \ref{P1.14} mentioned above. Our aim here is to recall some notations, definitions and results concerning the symbol spaces that we use, the magnetic field and the 'magnetic' pseudodifferential operators ($\Psi$DO), facts that have been discussed in \cite{IMP1,IMP2}. Important for what follows are the canonical systems and the Hamilton-Jacobi equation that we briefly recall here following \cite{Ro,DG}. Some consequences of this theory will be essential in the last two sections of the paper. Section 3 contains the precise definition of 'magnetic' Fourier Integral Operators (FIO) and some of their elementary properties. Sections 4-7 are essential from the technical point of view and contain a number of results concerning the composition of 'magnetic' FIO and $\Psi$DO. To some extent these results should be compared with \cite{Ku}, but the technical details due to the presence of the magnetic field make them rather different from those. In Section 4 we prove that by composing a $\Psi$DO with a FIO one obtains a FIO for which the first two terms of the asymptotic development of the symbol are computed. In Section 5 we prove that by composing a FIO with a $\Psi$DO one obtains a FIO and we compute its principal symbol. In sections 6 and 7 we compose a FIO with an adjoint of FIO and respectively in the reversed order, prove that they are $\Psi$DO and compute their principal symbol. Section 8 contains the main result of the paper and is based on sections 2-5. We prove that for $a\in S^1(\mathbb{R}^d)$ a real elliptic symbol with principal part $a_0$ positive homogeneous of degree 1, at least for $|t|$ small (but sometimes for $t\in\mathbb{R}$), if $P$ is the self-adjoint operator in $L^2(\mathbb{R}^d)$ associated to $\mathfrak{Op}^A(a)$, then $e^{-itP}$ is a 'magnetic' FIO associated to the symplectomorphism $\Phi_t$ defined as the Hamiltonian flow of $a_0$ (truncated near the zero section of $\mathbb{T}^*\mathbb{R}^d$); its symbol belongs to $S^0(\mathbb{R}^d)$ and is the asymptotic sum of a series of symbols that verify a system of transport equations. The last section is devoted to a number of applications of the results obtained in the previous sections. As consequence of the composition theorems proved in sections 4,5 and 7 we obtain a continuity property for FIO acting between the 'magnetic' Sobolev spaces defined in \cite{IMP1}. The Egorov type Theorem in Section 9.b is based on the results of sections 5,6 and 8. Under the assumptions made in section 8 we obtain some estimations for the distribution kernel of $e^{-itP}$, on the domain on which this kernel is $C^\infty$. Finally in section 9.d we prove a theorem on the propagation of singularities saying that 
$$
WF(e^{-itP}u)=\Phi^0_t(WFu),\quad\forall u\in\mathcal{S}^\prime(\mathbb{R}^d),\,\forall t\in\mathbb{R},
$$
where $\Phi^0_t:\mathbb{T}^*\mathbb{R}^d\setminus0\rightarrow\mathbb{T}^*\mathbb{R}^d\setminus0$ is the homogeneous canonical transformation defined by the Hamiltonian system associated to the homogeneous (non-truncated) principal symbol $a_0$.

%............................................................................................
\section{Preliminaries}
%...........................................................................................
\subsection{Notations}
The configuration space is $\X=\Rd$ with generic points $x=\left(\begin{array}{c}x_1\\\ldots\\\ldots\\\\x_d\end{array}\right)$. 

\noindent The scalar product on $\X$ is denoted by $<\cdot,\cdot>$.

\noindent For any $f:\Rd\rightarrow\Rn$ of class $C^1$ we write
$$
f(x)=\left(\begin{array}{c}f_1(x)\\\ldots\\\ldots\\\\f_n(x)\end{array}\right);
$$
$$
\partial f\equiv\d_xf:=\left(\partial_1f,\ldots,\partial_df\right)\equiv\left(\begin{array}{ccc}
\partial_1f_1&\ldots&\partial_df_1\\
\ldots&\ldots&\ldots\\
\partial_1f_n&\ldots&\partial_df_n
\end{array}
\right);
$$
$$
\nabla f\equiv\nabla_xf:=\left(\begin{array}{c}\nabla_1f\\\ldots\\\nabla_df\end{array}\right)\equiv\left(\begin{array}{ccc}
\partial_1f_1&\ldots&\partial_1f_n\\
\ldots&\ldots&\ldots\\
\partial_df_1&\ldots&\partial_df_n
\end{array}
\right).
$$
For $f:\Rd\rightarrow\Rn$ and $g:\Rm\rightarrow\Rd$, with $h:=f\circ g:\Rm\rightarrow\Rn$ we have
\begin{equation}\label{1.1}
\partial(f\circ g)(x)=\left(\partial f\right)\big(g(x)\big)\cdot\left(\partial g\right)(x),\quad\forall x\in\Rm.
\end{equation}
For $f:\Rd\rightarrow\Rd$ of class $C^1$ invertible, if we denote $g=f^{-1}$ we have thus
\begin{equation}\label{1.2}
\big(\partial g\big)(x)=\left[\big(\partial f\big)\big(g(x)\big)\right]^{-1},\quad\forall x\in\Rd.
\end{equation}
If $S:\Rd\rightarrow\mathbb{R}$ is of class $C^2$, then
$$
S_{x,x}^{\prime\prime}\equiv\nabla_{x,x}^2S:=\nabla_x\big(\nabla_xS\big)=\left(\begin{array}{ccc}
\partial_1^2S&\ldots&\partial_1\partial_dS\\
\ldots&\ldots&\ldots\\
\partial_d\partial_1S&\ldots&\partial_d^2S
\end{array}
\right).
$$
If $S:\Rm\times\Rn\rightarrow\mathbb{R}$ is of class $C^2$, then
$$
S_{x,\xi}^{\prime\prime}\equiv\nabla_{x,\xi}^2S:=\nabla_x\big(\nabla_\xi S\big)=\left(\begin{array}{ccc}
\partial_{x_1}\partial_{\xi_1}S&\ldots&\partial_{x_1}\partial_{\xi_n}S\\
\ldots&\ldots&\ldots\\
\partial_{x_m}\partial_{\xi_1}S&\ldots&\partial_{x_m}\partial_{\xi_n}S
\end{array}
\right).
$$

\subsection{Symbols}
\begin{definition}\label{D1.1}
 Let $n\in\mathbb{N}$, $d_j\in\mathbb{N}^*$ for $j\in\{0,1,\ldots,n\}$ and $m_k\in\mathbb{R}$ for $k\in\{1,\ldots,n\}$. For a function $f\in C^\infty(\mathbb{R}^{d_0}\times\mathbb{R}^{d_1}\times\ldots\times\mathbb{R}^{d_n})$ we define
$$
\rvert f\rvert_{\alpha_0,\ldots,\alpha_n}:=\underset{(u,v_1,\ldots v_n)\in\mathbb{R}^{d_0}\times\mathbb{R}^{d_1}\times\ldots\times\mathbb{R}^{d_n}}{\sup}<v_1>^{|\alpha_1|-m_1}\ldots<v_n>^{|\alpha_n|-m_n}\left|\partial_u^{\alpha_0}\partial_{v_1}^{\alpha_1}\ldots\partial_{v_n}^{\alpha_n}f(u,v_1,\ldots v_n)\right|,\ \forall\alpha_j\in\mathbb{N}^{d_j}
$$
and
$$
S^{m_1,\ldots,m_n}(\mathbb{R}^{d_0}\times\mathbb{R}^{d_1}\times\ldots\times\mathbb{R}^{d_n}):=\left\{f\in C^\infty(\mathbb{R}^{d_0}\times\mathbb{R}^{d_1}\times\ldots\times\mathbb{R}^{d_n})\mid\ \rvert f\rvert_{\alpha_0,\ldots,\alpha_n}<\infty,\quad\forall\alpha_j\in\mathbb{N}^{d_j}, 0\leq j\leq n\right\}.
$$
\end{definition}
Then $\rvert f\rvert_{\alpha_0,\ldots,\alpha_n}$ are semi-norms on $S^{m_1,\ldots,m_n}(\mathbb{R}^{d_0}\times\mathbb{R}^{d_1}\times\ldots\times\mathbb{R}^{d_n})$ and they define a Frechet structure on it.

We shall use the notations:
$$
S^m\equiv S^m(\Rd):=S^m(\Rd\times\Rd),\quad S^{-\infty}:=\underset{m\in\mathbb{R}}{\bigcap}S^m,
$$
$$
\mathfrak{p}:\Rd\times\Rd\rightarrow\mathbb{R},\ \mathfrak{p}(x,\xi):=<\xi>\equiv\sqrt{1+|\xi|^2},
$$
$$
S^m_0:=\left\{f\in C^\infty(\Rd\times\Rd)\,\mid\,\mathfrak{p}^{-m}f\in BC^\infty(\Rd\times\Rd)\right\},\quad S^{-\infty}_0:=\underset{m\in\mathbb{R}}{\bigcap}S^m_0,
$$
$$
S^+:=\left\{f\in C^\infty(\Rd\times\Rd)\,\mid\,\partial_{x_j}f\in S^1,\ \partial_{\xi_j}f\in S^0,\ 1\leq j\leq d\right\}.
$$
\begin{remark}\label{R1.2}
 All these are Fr\'{e}chet spaces for their natural topology. Evidently: $S^{-\infty}=S^{-\infty}_0$.
\end{remark}
We recall Proposition 18.1.3 from \cite{Ho}
\begin{lemma}\label{L1.3}
 Suppose given $a_j\in S^{m_j}$ for $j\in\mathbb{N}$ such that $m_j\searrow-\infty$ for $j\rightarrow\infty$. Then there exists a symbol $a\in S^{m_0}$, unique modulo $S^{-\infty}$, such that for any $k\geq1$
$$
a-\sum_{j=0}^{k-1}a_j\in S^{m_k}.
$$
\end{lemma}
Such a symbol $a$ as in the above statement is called {\it the asymptotic sum of the series} $\sum_ja_j$ and is usually denoted by $a\sim\sum_{j=0}^\infty a_j$.

Let us also recall a usual way to construct such an {\it asymptotic sum}. We choose a function $\chi\in C^\infty(\Rd)$ verifying $\chi(\xi)=0$ for $|\xi|\leq 1$ and $\chi(\xi)=1$ for $|\xi|\geq 2$ and a sequence $\{t_j\}_{j\in\mathbb{N}}$ verifying
$$
1\leq t_0<t_1<\ldots<tk<t_{k+1},\ t_j\rightarrow\infty\ \text{for}\ j\rightarrow\infty\ \text{rapidly enough}.
$$
Then one can define
\begin{equation}\label{1.4}
a(x,\xi):=\sum_{j=0}^\infty\chi(\xi/t_j)a_j(x,\xi),\quad(x,\xi)\in\Rd\times\Rd.
\end{equation}
\begin{definition}\label{D1.4}
 A symbol $a\in S^m$ is said to be {\it elliptic} when there exist two positive constants $R$ and $C$ such that
$$
\left|a(x,\xi)\right|\geq C<\xi>^m,\quad\forall(x,\xi)\in\Rd\times\Rd,\text{ with }|\xi|\geq R.
$$
\end{definition}

Here are two technical results needed in the forthcoming arguments.
\begin{lemma}\label{L1.5}
 Let $a\in C^\infty(\mathbb{R}^{2d}\times\Rd)$ be a real function such that $\partial_{y_j}a\in S^1(\mathbb{R}^{2d}\times\Rd)$. Suppose also that there exists $\delta\in[0,1)$ such that
\begin{equation}\label{1.5}
 \left\|\left(\nabla^2_{\xi,y}a\right)(x,y,\xi)\right\|\leq\delta,\quad\forall(x,y)\in\mathbb{R}^{2d}\ \forall\xi\in\Rd.
\end{equation}
If we denote by $\Lambda(x,y,\xi):=\xi+\left(\nabla_{y}a\right)(x,y,\xi)$, then we have:
\begin{enumerate}
 \item $\Lambda\in S^1(\mathbb{R}^{2d}\times\Rd)$ and for any $(x,y)\in\mathbb{R}^{2d}$ the map
$
\Rd\ni\xi\mapsto\Lambda(x,y,\xi)\in\Rd
$
is a $C^\infty$-diffeomorphism. We denote by $\Rd\ni\eta\mapsto\lambda(x,y,\eta)\in\Rd$ its inverse.
\item There exists a strictly positive constant $C$ such that
\begin{equation}\label{1.6}
 C^{-1}<\xi>\leq\left<\Lambda(x,y,\xi)\right>\leq C<\xi>,\quad C^{-1}<\eta>\leq\left<\lambda(x,y,\eta)\right>\leq C<\eta>,\quad\forall x,y,\xi,\eta\in\mathbb{R}^{d}.
\end{equation}
\item The map $\lambda:\Rd\rightarrow\Rd$ has components of class $S^1(\mathbb{R}^{2d}\times\Rd)$.
\item For any $(x,y)\in\mathbb{R}^{2d}$ let us denote by $D(x,y,\eta)$ the functional determinant of the map $\Rd\ni\eta\mapsto\lambda(x,y,\eta)\in\Rd$. Then $D\in S^0(\mathbb{R}^{2d}\times\Rd)$.
\end{enumerate}
\end{lemma}
\begin{proof}

 \noindent
{\it 1.} As $\nabla_\xi\Lambda=1_d+\nabla^2_{\xi,y}a$ we deduce from the hypothesis of the Lemma that
$$
\left\|\left(\nabla_\xi\Lambda\right)(x,y,\xi)\right\|\geq1-\delta.
$$
As it is evident by definition that $\Lambda\in S^1(\mathbb{R}^{2d}\times\Rd)$, the statement results from the Schwartz global inversion theorem.

 \noindent
{\it 2.} We evidently can write
$$
\Lambda(x,y,\xi)=\xi+\left(\nabla_{y}a\right)(x,y,0)+\left[\int_0^1\left(\nabla^2_{y,\xi}a\right)(x,y,t\xi)dt\right]\cdot\xi.
$$
But $\left\|\nabla^2_{\xi,y}a(x,y,\xi)\right\|=\left\|\nabla^2_{y,\xi}a(x,y,\xi)\right\|$ and thus the first two inequalities in statement (\ref{1.6}) follow. For the other two inequalities we just take $\xi=\lambda(x,y,\eta)$ in the first two and obtain them.

 \noindent
{\it 3.} From \eqref{1.6} we deduce that $|\lambda(x,y,\eta)|\leq C<\eta>$. By differentiating the equality $\eta=\Lambda(x,y,\lambda(x,y,\eta))$ and using the fact that
$$
\left\|\big(\d_\xi\Lambda\big)(x,y,\xi)\right\|=\left\|\big(\nabla_\xi\Lambda\big)(x,y,\xi)\right\|\geq 1-\delta,
$$
we obtain (as matrices $d\times d$)
$$
\begin{array}{rcl}
\big(\d_\eta\lambda\big)(x,y,\eta)&=&\left[\big(\d_\xi\Lambda\big)(x,y,\lambda(x,y,\eta))\right]^{-1}\\
\big(\d_x\lambda\big)(x,y,\eta)&=&-\left[\big(\d_\xi\Lambda\big)(x,y,\lambda(x,y,\eta))\right]^{-1}\cdot\big(\d_x\Lambda\big)(x,y,\lambda(x,y,\eta))\\
\big(\d_y\lambda\big)(x,y,\eta)&=&-\left[\big(\d_\xi\Lambda\big)(x,y,\lambda(x,y,\eta))\right]^{-1}\cdot\big(\d_y\Lambda\big)(x,y,\lambda(x,y,\eta)).
\end{array}
$$
An induction on the differentiation order, the fact that $\Lambda\in S^1(\mathbb{R}^{2d}\times\Rd)$ and (\ref{1.6}) imply that $\lambda\in S^1(\mathbb{R}^{2d}\times\Rd)$.

 \noindent
{\it 4.} Results from point 3.
\end{proof}

\begin{lemma}\label{L1.6}
 Let $a:\mathbb{R}^{4d}\rightarrow\Rd$ be a map with components of class $BC^\infty(\mathbb{R}^{4d})$. Suppose that there exists a constant $\delta\in[0,1)$ such that
\begin{equation}\label{1.7}
 \left\|\big(\nabla_ya\big)(x,y,\xi,\eta)\right\|\leq\delta,\quad\forall(x,y,\xi,\eta)\in\mathbb{R}^{4d}.
\end{equation}
Then:
\begin{enumerate}
 \item For any $(x,\xi,\eta)\in\mathbb{R}^{3d}$, the map $\Rd\ni y\mapsto y+a(x,y,\xi,\eta)\in\Rd$ is a $C^\infty$-diffeomorphism.
\item Let us denote by $g(x,z,\xi,\eta)$ the inverse of the diffeomorphism of point (1) and by $D(x,z,\xi,\eta)$ the functional determinant of the map $\Rd\ni z\mapsto g(x,z,\xi,\eta)\in\Rd$. Then $g(x,z,\xi,\eta)-z$ and $D(x,z,\xi,\eta)$ are of class $BC^\infty(\mathbb{R}^{4d})$.
\item Suppose also that there exists some $\epsilon>0$ such that on the set
$$
\mathcal{M}_\epsilon:=\left\{(x,y,\xi,\eta)\in\mathbb{R}^{4d}\,\mid\,|\eta|\leq\epsilon|\xi|\right\}
$$
the map $a$ verifies the inequalities:
\begin{equation}\label{1.8}
 \left|\left(\partial^\alpha_{(x,y)}\partial^\beta_{(\xi,\eta)}a\right)(x,y,\xi,\eta)\right|\leq C_{\alpha\beta}<\xi>^{-|\beta|},\quad\forall(\alpha,\beta)\in\mathbb{N}^{2d}\times\mathbb{N}^{2d},
\end{equation}
where $C_{\alpha\beta}$ are positive constants. Then similar inequalities are verified by $g(x,z,\xi,\eta)-z$ and $D(x,z,\xi,\eta)$ for $(x,z,\xi,\eta)\in\mathcal{M}_\epsilon$.
\end{enumerate}
\end{lemma}
\begin{proof}
 
 \noindent
{\it 1.} This results in a similar way with the proof of (1) of Lemma \ref{L1.5} by using the global inversion Theorem by Schwartz and the fact that
\begin{equation}\label{1.9}
 \left\|\boldsymbol{1}_d+\left(\nabla_ya\right)(x,y,\xi,\eta)\right\|\geq1-\delta.
\end{equation}

 \noindent
{\it 2.} Let us observe that
\begin{equation}\label{1.10}
 g(x,z,\xi,\eta)+a(x,g(x,z,\xi,\eta),\xi,\eta)=z.
\end{equation}
This equality implies the boundedness of $g(x,z,\xi,\eta)-z$ and by an induction and successive differentiation of \eqref{1.10} (using also \eqref{1.9}) the boundedness of the derivatives also follows.

 \noindent
{\it 3.} This follows in a similar way as point (2) by using the inequalities \eqref{1.8}.
\end{proof}
\begin{definition}\label{D1.7}
 We say that a symbol $a\in S^m(\Rd)$ has {\it a homogeneous principal part} $a_0\in C^\infty\big(\Rd\times(\Rd\setminus\{0\})\big)$ when the following facts hold:
\begin{itemize}
 \item $a_0(x,\lambda\xi)=\lambda^ma_0(x,\xi)$ for any $(x,\xi)\in\Rd\times(\Rd\setminus\{0\})$ and any $\lambda>0$;
\item $\partial^\alpha_x\partial^\beta_\xi a_0$ are bounded on the set $\Rd\times\{\xi\in\Rd\,|\,|\xi|=1\}$ for any $(\alpha,\beta)\in\mathbb{N}^{2d}$;
\item for any cut-off function $\chi\in C^\infty(\Rd)$ with $\chi(\xi)=0$ for $|\xi|\leq1$ and $\chi(\xi)=1$ for $|\xi|\geq2$, we have $$a-\chi a_0\in S^{m-1}(\Rd).$$
\end{itemize}
\end{definition}

\noindent
{\bf Example} The homogeneous principal part of the symbol $a(\xi)=<\xi>$ (of class $S^1(\Rd)$) is $a_0(\xi)=|\xi|$.

\subsection{The Magnetic Fields}

\begin{definition}\label{D1.9}
 A magnetic field on $\mathbb{R}^d$ is a closed 2-form 
$$
B:=\frac{1}{2}\sum_{j,k=1}^d B_{jk}dx_j\wedge dx_k
$$
where $B_{jk}=-B_{kj}\in BC^\infty(\mathbb{R}^d)$ (thus we have  $dB=0$).
\end{definition}
\begin{remark}\label{R1.10}
 Using the transverse gauge we can define a vector potential $A:=\sum_{j=1}^d A_jdx_j$ with $A_j\in C^\infty_{\text{\sf pol}}(\mathbb{R}^d)$ real functions such that $B=dA$. In fact we can define:
$$
A_j(x):=-\sum_{k=1}^d\int_0^1ds\,B_{jk}(sx)sx_k,\quad\forall x\in\mathbb{R}^d.
$$
Any other such vector potential is of the form $A^\prime=A+\d\varphi$ with $\varphi\in C^\infty_{\text{\sf pol}}(\mathbb{R}^d)$ a real function.
\end{remark}

We shall use the following notations:
\begin{equation}\label{GammaA}
 \Gamma^A(x,y):=\int_0^1ds\,A\big((1-s)x+sy\big);\qquad\int_{[x,y]}A:=-\left<x-y,\Gamma^A(x,y)\right>,
\end{equation}
\begin{equation}\label{omegaA}
 \omega^A(x,y):=\exp\left\{-i\int_{[x,y]}A\right\}
\end{equation}
\begin{equation}\label{FB}
 F^B(x,y,z)\equiv\Gamma^B(<x,y,z>):=\int_{<x,y,z>}B,
\end{equation}
\begin{equation}\label{OmegaB}
 \Omega^B(x,y,z):=\exp\left\{-i\int_{<x,y,z>}B\right\}.
\end{equation}
The next Lemma summarizes some elementary properties of these functions (see \cite{IMP1}).
\begin{lemma}\label{L1.11}
 
\begin{enumerate}
 \item $\omega^A(x,y)\omega^A(y,x)=1.$
\item $\omega^A(x,y)\omega^A(y,z)\omega^A(z,x)=\Omega^B(x,y,z).$
\item If we denote by
\begin{equation}\label{1.12}
 c_{jk}(x,y,z):=\int_0^1ds\int_0^sdt\,B_{jk}\big(tx+(s-t)y+(1-s)z\big),\quad\forall(x,y,z)\in\mathbb{R}^{3d},\,1\leq j,k\leq d
\end{equation}
$$
C(x,y,z):=\big(c_{jk}(x,y,z)\big)_{1\leq j,k\leq d},
$$
then we have that $c_{jk}\in BC^\infty(\mathbb{R}^{3d})$, for each $(x,y,z)\in\mathbb{R}^{3d}$ the matrix $C(x,y,z)$ is antisymetric and
\begin{equation}\label{1.13}
 F(x,y,z)=-\left<C(x,y,z)(x-y),(x-z)\right>.
\end{equation}
\item All the vectors of the form $\nabla_xF(x,y,z)$, $\nabla_yF(x,y,z)$ and $\nabla_zF(x,y,z)$ are of the form
$$
D(x,y,z)(x-y)\,+\,E(x,y,z)(x-z)
$$
with $D$ and $E$ $d\times d$ matrices with entries from $BC^\infty(\mathbb{R}^{3d})$.
\end{enumerate}

\end{lemma}

\subsection{The Magnetic Pseudo-differential Operators}

\begin{definition}\label{D1.12} (\cite{Mu,MP,KO1,KO2}).
 We define the \textit{magnetic pseudodifferential operator} associated to the symbol $a\in S^m(\mathbb{R}^d)$ and to the magnetic field $B=dA$ as the linear continuous map $\mathfrak{Op}^A(a):\mathcal{S}(\mathbb{R}^d)\rightarrow\mathcal{S}(\mathbb{R}^d)$ given by the following oscillating integral:
\begin{equation}\label{1.14}
\left[\mathfrak{Op}^A(a)u\right](x):=\int_{\mathbb{R}^d}\int_{\mathbb{R}^d}e^{i<x-y,\eta>}\omega^A(x,y)a\left(\frac{x+y}{2},\eta\right)u(y)\,dy\,\dbar\eta.
\end{equation}
\end{definition}

Evidently there exists a unique linear and continuous extension $\mathfrak{Op}^A(a):\mathcal{S}^\prime(\mathbb{R}^d)\rightarrow\mathcal{S}^\prime(\mathbb{R}^d)$.

\begin{remark}\label{R1.13}

 \begin{enumerate}
  \item This operator is in fact determined by $B$ due to the gauge covariance of the formula in its definition: if $A^\prime=A+\d\varphi$ with $\varphi\in C^\infty_{\text{\sf pol}}(\mathbb{R}^d)$ a real function, then $\mathfrak{Op}^{A^\prime}(a)=e^{i\varphi}\mathfrak{Op}^A(a)e^{-i\varphi}$.
\item One can define pseudodifferential operators associated to some other symbol classes, as for example $S^m_0(\mathbb{R}^d)$ (see \cite{IMP1}).
 \end{enumerate}

The above {\it quantization procedure} defined in Definition \ref{D1.12} is a 'magnetic' analog of the Weyl quantization. We shall also need the 'magnetic' equivalent of the Kohn-Nirenberg quantization (or the 'left quantization'):
\begin{equation}\label{1.15}
 \left[b^A(x,D)u\right](x):=\int_{\mathbb{R}^d}\int_{\mathbb{R}^d}e^{i<x-y,\eta>}\omega^A(x,y)b(x,\eta)u(y)\,dy\,\dbar\eta,\quad\forall x\in\mathbb{R}^d,\,\forall u\in\mathcal{S}(\mathbb{R}^d),
\end{equation}
for any symbol $b\in S^m(\mathbb{R}^d)$. 
\end{remark}

The equivalence of these two 'quantizations' results from the following Proposition.
\begin{proposition}\label{P1.14}
 
\begin{enumerate}
 \item Given any $a\in S^m(\mathbb{R}^d)$ there exists a symbol $b\in S^m(\mathbb{R}^d)$ such that $\mathfrak{Op}^A(a)=b^A(x,D)$ and moreover we have that
\begin{equation}\label{1.16}
 b(x,\eta)\,\sim\,\sum_{\alpha\in\mathbb{N}^d}\frac{1}{\alpha!}\left(\frac{1}{2}\right)^{|\alpha|}\left(D^\alpha_x\partial^\alpha_\eta a\right)(x,\eta).
\end{equation}
\item Reciprocally, for any $b\in S^m(\mathbb{R}^d)$ there exists a symbol $a\in S^m(\mathbb{R}^d)$ such that $b^A(x,D)=\mathfrak{Op}^A(a)$ and moreover we have that
\begin{equation}\label{1.17}
 a(x,\eta)\,\sim\,\sum_{\alpha\in\mathbb{N}^d}\frac{1}{\alpha!}\left(-\frac{1}{2}\right)^{|\alpha|}\left(D^\alpha_x\partial^\alpha_\eta b\right)(x,\eta).
\end{equation}
\end{enumerate}
\end{proposition}
\begin{proof}
 The distribution kernels of the operators defined by \eqref{1.14} and respectively \eqref{1.15} are
\begin{equation}\label{1.18}
 K^w_a(x,y):=\omega^A(x,y)\int_{\mathbb{R}^d}e^{i<x-y,\eta>}a\left(\frac{x+y}{2},\eta\right)\,\dbar\eta,
\end{equation}
\begin{equation}\label{1.19}
 K_b(x,y):=\omega^A(x,y)\int_{\mathbb{R}^d}e^{i<x-y,\eta>}b(x,\eta)\,\dbar\eta,
\end{equation}
where the integrals are oscillating and define distributions in $\mathcal{S}^\prime(\mathbb{R}^d\times\mathbb{R}^d)$. The equality we want to prove will follow from the equality $K^w_a(x,y)=K_b(x,y)$ in $\mathcal{S}^\prime(\mathbb{R}^d\times\mathbb{R}^d)$ that is equivalent to
$$
\left[\mathcal{F}_\eta^{-1}\left(a\left(\frac{x+y}{2},\cdot\right)\right)\right](x-y)=\left[\mathcal{F}_\eta^{-1}\left(b(x,\cdot)\right)\right](x-y).
$$
The change of variables $y\mapsto z=x-y$ gives
$$
b(x,\eta)=\int_{\mathbb{R}^d}\int_{\mathbb{R}^d}e^{i<z,\zeta-\eta>}a\left(x-z/2,\zeta\right)dz\,\dbar\zeta=\int_{\mathbb{R}^d}\int_{\mathbb{R}^d}e^{i<z,\zeta>}a\left(x-z/2,\eta+\zeta\right)dz\,\dbar\zeta.
$$
We use a Taylor development (with $N\geq1$):
$$
a\left(x-z/2,\eta+\zeta\right)=\sum_{|\alpha|< N}\frac{\zeta^\alpha}{\alpha!}\left(\partial^\alpha_\eta a\right)\left(x-z/2,\eta\right)+\sum_{|\alpha|=N}\frac{\zeta^\alpha}{(N-1)!}\int_0^1(1-t)^{N-1}dt\left(\partial^\alpha_\eta a\right)\left(x-z/2,\eta+t\zeta\right),
$$
so that we get
$$
b(x,\eta)=\sum_{|\alpha|< N}\frac{1}{\alpha!}\left(\frac{1}{2}\right)^{|\alpha|}\left(D^\alpha_x\partial^\alpha_\eta a\right)(x,\eta)+r_N(x,\eta)
$$
with
$$
r_N(x,\eta):=\sum_{|\alpha|=N}\frac{1}{(N-1)!}\left(\frac{1}{2}\right)^{N}\int_0^1(1-t)^{N-1}dt\times
$$
$$
\times\int_{\mathbb{R}^d}\int_{\mathbb{R}^d}e^{i<z,\zeta>}<z>^{-2N_1}(1-\Delta_\zeta)^{N_1}\left[<\zeta>^{-2N_2}(1-\Delta_z)^{N_2}\left(D^\alpha_x\partial^\alpha_\eta a\right)\left(x-z/2,\eta+t\zeta\right)dz\,\dbar\zeta\right],
$$
with $N_1$ and $N_2$ large enough. It follows that one can find two positive constants $C$ and $C^\prime$ such that
$$
\left|r_N(x,\eta)\right|\leq C^\prime\underset{0\leq t\leq1}{\sup}\int_{\mathbb{R}^d}\int_{\mathbb{R}^d}<z>^{-2N_1}<\zeta>^{-2N_2}<\eta+t\zeta>^{m-N}dz\,\dbar\zeta\leq C<\eta>^{m-N}.
$$
The derivatives of $r_N$ are estimated in the same way and we obtain that $r_N\in S^{m-N}(\mathbb{R}^d)$ and the first conclusion of the Proposition follows. The proof of the second statement is quite similar.
\end{proof}

\begin{definition}\label{D1.15}
 We call {\it principal symbol} of the operator $\mathfrak{Op}^A(a)$, $a\in S^m(\mathbb{R}^d)$ (respectively $b^A(x,D)$, $b\in S^m(\mathbb{R}^d)$) any representative of the equivalence class of $a$ (respectively $b$) in $S^m(\mathbb{R}^d)/S^{m-1}(\mathbb{R}^d)$.
\end{definition}

Thus our Proposition \ref{P1.14} means that the principal symbol of a magnetic $\Psi$DO is the same in the Weyl and in the 'left' quantizations.

\begin{lemma}\label{L1.16}
 Let $a\in S^m(\mathbb{R}^d)$ and $K^w_a(x,y)$ the distribution kernel of $\mathfrak{Op}^A(a)$. Then $K^w_a(x,y)=\omega^A(x,y)L^w_a(x,y)$ where $L^w_a\in C^\infty(\mathbb{R}^d\times\mathbb{R}^d\setminus\Delta)$ with $\Delta$ the diagonal set of $\mathbb{R}^d\times\mathbb{R}^d$. Moreover we have that 
$$
(x-y)^\gamma\partial_x^\alpha\partial_y^\beta L^w_a\in BC(\mathbb{R}^d\times\mathbb{R}^d),\quad\forall(\alpha,\beta,\gamma)\in[\mathbb{N}^d]^3\ \text{with}\ |\gamma|>m+d+|\alpha+\beta|.
$$
For $m=-\infty$ we have $L^w_a\in BC^\infty(\mathbb{R}^d\times\mathbb{R}^d)$ and $(x-y)^\gamma\partial_x^\alpha\partial_y^\beta L^w_a\in L^\infty(\mathbb{R}^d\times\mathbb{R}^d),\quad\forall(\alpha,\beta,\gamma)\in[\mathbb{N}^d]^3$.
\end{lemma}
\begin{proof}
 Due to \eqref{1.18} we have that
$$
L^w_a(x,y)=\int_{\mathbb{R}^d}e^{i<x-y,\eta>}a\left(\frac{x+y}{2},\eta\right)\dbar\eta,
$$
so that for any $\gamma\in\mathbb{N}^d$ we have that
$$
(x-y)^\gamma L^w_a(x,y)=\int_{\mathbb{R}^d}e^{i<x-y,\eta>}\left[\left(-D_\eta\right)^\gamma a\right]\left(\frac{x+y}{2},\eta\right)\dbar\eta.
$$
All the statements in the Lemma follow now by straightforward arguments.
\end{proof}

Let us recall now the 'magnetic' Sobolev spaces introduced in \cite{IMP1,IMP2}.
\begin{definition}\label{D1.17}
 Let us fix $s\geq0$, and let $p_s(x,\xi):=<\xi>^s$ and $P_s:=\mathfrak{Op}^A(p_s)$.
\begin{enumerate}
 \item {\it The magnetic Sobolev space of order $s\geq0$} is
$$
\mathcal{H}^s_A:=\left\{u\in L^2(\mathbb{R}^d)\mid P_su\in L^2(\mathbb{R}^d)\right\},
$$
that is a Hilbert space for the norm: $\|u\|_{s,A}:=\sqrt{\|u\|^2+\|P_su\|^2}$.
\item {\it The magnetic Sobolev space of negative order $-s$} is the dual space of $\mathcal{H}^s_A$ endowed with the natural norm.
\end{enumerate}
\end{definition}
The following result has been proved in \cite{IMP1}.
\begin{lemma}\label{L1.18}
 Suppose $a\in S^m(\mathbb{R}^d)$ and $s\in\mathbb{R}$.
\begin{enumerate}
 \item $\mathfrak{Op}^A(a)$ belongs to $\mathbb{B}\big(\mathcal{H}^s_A;\mathcal{H}^{s-m}_A\big)$.
\item If $m>0$ and $a$ is real and elliptic, then $\mathfrak{Op}^A(a)$ defines a self-adjoint operator in $L^2(\mathbb{R}^d)$ with domain $\mathcal{H}^m_A$.
\end{enumerate}
\end{lemma}

We shall also need a criterion in terms of commutators for an operator in $L^2(\mathbb{R}^d)$ to be a magnetic $\Psi$DO with symbol of class $S^m_0(\mathbb{R}^d)$ (see Theorem 5.20 in \cite{IMP2}). Let us recall that in fact any tempered distribution $a\in\mathcal{S}^\prime(\mathbb{R}^d\times\mathbb{R}^d)$ may be considered the 'symbol' of a magnetic pseudodifferential operator $\mathfrak{Op}^A(a)\in\mathbb{B}\big(\mathcal{S}(\mathbb{R}^d);\mathcal{S}^\prime(\mathbb{R}^d)\big)$, so that the map $\mathfrak{Op}^A:\mathcal{S}^\prime(\mathbb{R}^d\times\mathbb{R}^d)\rightarrow\mathbb{B}\big(\mathcal{S}(\mathbb{R}^d);\mathcal{S}^\prime(\mathbb{R}^d)\big)$ to become a topological isomorphism. We shall denote by 
$$
\mathcal{B}^B:=\left[\mathfrak{Op}^A\right]^{-1}\big[\mathbb{B}\big(L^2(\mathbb{R}^d)\big)
$$
and endow it with the norm $\|a\|_{\mathcal{B}^B}:=\left\|\mathfrak{Op}^A(a)\right\|_{\mathbb{B}\big(L^2(\mathbb{R}^d)\big)}$.

Suppose given $a$ and $b$ in $\mathcal{S}^\prime(\mathbb{R}^d\times\mathbb{R}^d)$ such that $\mathfrak{Op}^A(b)$ belongs to $\mathbb{B}\big(\mathcal{S}(\mathbb{R}^d)\big)\bigcap\mathbb{B}\big(\mathcal{S}^\prime(\mathbb{R}^d)\big)$. Then the operator $\mathfrak{Op}^A(b)\mathfrak{Op}^A(a)$ and the commutator $\mathfrak{ad}_{\mathfrak{Op}^A(b)}\big(\mathfrak{Op}^A(a)\big)=\mathfrak{Op}^A(b)\mathfrak{Op}^A(a)-\mathfrak{Op}^A(a)\mathfrak{Op}^A(b)$ are well defined, the first having the symbol $b\sharp^B a$ and the second having the symbol ${\rm ad}^B_b(a):=b\sharp^B a-a\sharp^B b$. In particular, if $X=(x,\xi)\in\mathbb{R}^d\times\mathbb{R}^d$ and $l_X:\mathbb{R}^d\rightarrow\mathbb{R}$ is given by $l_X(Y):=<\xi,y>-<\eta,x>$ for any $Y=(y,\eta)\in\mathbb{R}^d\times\mathbb{R}^d$, we denote by $\text{\rm ad}^B_X[a]:=\text{\rm ad}^B_{l_X}(a)$, $a\in \mathcal{S}^\prime(\mathbb{R}^d\times\mathbb{R}^d)$. The following theorem is proved in \cite{IMP2}.
\begin{theorem}\label{T1.19}
 Suppose $a\in \mathcal{S}^\prime(\mathbb{R}^d\times\mathbb{R}^d)$ and let us fix $m\leq0$. Then $a\in S^m_0(\mathbb{R}^d)$ if and only if for any $N\in\mathbb{N}$ and for any set of $N$ points $X_1,\ldots,X_N$ in $\mathbb{R}^d\times\mathbb{R}^d$ we have that
$$
p_{-m}\sharp^B\left(\text{\rm ad}^B_{X_1}\cdots\text{\rm ad}^B_{X_N}[a]\right)\in\mathcal{B}^B.
$$Moreover, the following two families of seminorms:
\begin{itemize}
 \item $\left\{\left\|p_{-m}\partial_x^\alpha\partial_\xi^\beta a\right\|_\infty,\ (\alpha,\beta)\in[\mathbb{N}^d]^2\right\}$;
\item $\|a\|_{m,X_1,\ldots,X_N}:=\underset{0\leq k\leq N}{\max}\underset{\{j_1,\ldots,j_k\}\subset\{1,2,\ldots,N\}}{\max}\left\|p_{-m}\sharp^B\left(\text{\rm ad}^B_{X_{j_1}}\cdots\text{\rm ad}^B_{X_{j_k}}[a]\right)\right\|_{\mathcal{B}^B}$ indexed by $N\in\mathbb{N}$ and $X_1,\ldots,X_N\in\mathbb{R}^d\times\mathbb{R}^d$,
\end{itemize}
define equivalent topologies on $S^m_0(\mathbb{R}^d)$.
\end{theorem}

\subsection{The Generating Function of a Canonical Transformation}

Let $\Xi\equiv\mathbb{T}^*\mathcal{X}\equiv\Rd\times\Rd$ with points denoted using the convention $X=(x,\xi), Y=(y,\eta),\ldots$, and with the canonical 2-form $\sigma_X:=\underset{1\leq j\leq d}{\sum}d\xi_j\wedge dx_j$. We shall some times use the same notation for the canonical symplectic form raised on $\mathbb{T}^*\mathbb{R}^{2d}$, coresponding to the canonical projection $\mathbb{R}^{2d}\rightarrow\Rd$.

\begin{definition}\label{D1.20}
 
\begin{itemize}
 \item We call {\it a canonical relation} from $\mathbb{T}^*\Rd$ to $\mathbb{T}^*\Rd$, any $C^\infty$-submanifold $\Lambda\subset\mathbb{T}^*\mathbb{R}^{2d}$ of dimension $2d$ on which $\sigma_X-\sigma_Y=0$.
\item If the submanifold $\Lambda$ of the previous point is also a conic submanifold of $\left(\mathbb{T}^*\Rd\setminus0\right)\times\left(\mathbb{T}^*\Rd\setminus0\right)$ (with $0$ the zero section of $\mathbb{T}^*\Rd$ and conic meaning invariant under the multiplication with strictly positive scalars on the fibres of $\mathbb{T}^*\mathbb{R}^{2d}$), we shall call it {\it a homogeneous canonical relation} from $\mathbb{T}^*\Rd\setminus0$ to $\mathbb{T}^*\Rd\setminus0$.
\end{itemize}
\end{definition}
Let us make some coments about the previous definition. First we observe that we can define the following canonical projections
$$
\Pi_1:\mathbb{T}^*\mathbb{R}^{2d}\ni(X,Y)\mapsto X\in\mathbb{T}^*\Rd,\quad\Pi_2:\mathbb{T}^*\mathbb{R}^{2d}\ni(X,Y)\mapsto Y\in\mathbb{T}^*\Rd,
$$
so that in fact $\sigma_X-\sigma_Y\equiv\Pi_1^*\sigma-\Pi_2^*\sigma$. Moreover we remark that our definition is equivalent with asking that 
\begin{equation}\label{1.prime}
\Lambda^\prime:=\left\{(x,\xi,y,\eta)\in\mathbb{T}^*\mathbb{R}^{2d}\,\mid\,(x,\xi,y,-\eta)\in\Lambda\right\}
\end{equation}
is a Lagrangean submanifold of $\mathbb{T}^*\mathbb{R}^{2d}$.

\paragraph{Examples} 

\begin{enumerate}
 \item Suppose $\Phi:\left(\mathbb{T}^*\Rd,\sigma\right)\rightarrow\left(\mathbb{T}^*\Rd,\sigma\right)$ is a symplectomorphism (i.e. a $C^\infty$-diffeomorphism between the given symplectic spaces verifying $\Phi^*\big(\sigma_X\big)=\sigma_{\Phi^{-1}(X)}$). Then $\Lambda:=\text{\sf graph}\Phi$ is a canonical relation from $\mathbb{T}^*\Rd$ to $\mathbb{T}^*\Rd$. (We write $\text{\sf graph}\Phi=\left\{(\Phi(Y),Y)\mid Y\in\mathbb{T}^*\Rd\right\}$).
\item Suppose now that $\Phi:\mathbb{T}^*\Rd\setminus0\rightarrow\mathbb{T}^*\Rd\setminus0$ is a $C^\infty$-diffeomorphism such that $\Phi^*\big(\sigma_X\big)=\sigma_{\Phi^{-1}(X)}$ and homogeneous (i.e. commuting with the natural action of $\mathbb{R}^*_+$ on the fibres of $\mathbb{T}^*\Rd$, that means that for $X=\Phi(Y)$ with $x=\Phi_1(y,\eta)$ and $\xi=\Phi_2(y,\eta)$ we have that for any $(y,\eta)\in\mathbb{T}^*\Rd\setminus0$ and any $\lambda\in\mathbb{R}^*_+$, $\Phi_1(y,\lambda\eta)=\Phi_1(y,\eta)$ and $\Phi_2(y,\lambda\eta)=\lambda\Phi_2(y,\eta)$), then $\Lambda:=\text{\sf graph}\Phi$ is a homogeneous canonical relation from $\mathbb{T}^*\Rd\setminus0$ to $\mathbb{T}^*\Rd\setminus0$.
\item Let $U$ be a real valued map in $C^\infty(\Rd\times\Rd)$. Then the set
\begin{equation}\label{1.22}
 \Lambda:=\left\{\left(x,\big(\nabla_xU\big)(x,\eta),\big(\nabla_\eta U\big)(x,\eta),\eta\right)\,\mid\,x\in\Rd,\eta\in\Rd\right\}
\end{equation}
is a canonical relation from $\mathbb{T}^*\Rd$ to $\mathbb{T}^*\Rd$. The function $U$ is called {\it the generating function} of $\Lambda$. It is evident that any other generating function of $\Lambda$ differs of $U$ only by an additive constant.
\item Let $U$ be a real valued map in $C^\infty\big(\Rd\times\left(\Rd\setminus\{0\}\right)\big)$ that is homogeneous of order 1 in the second variable (i.e. for any $x\in\Rd$, $\eta\in\Rd\setminus\{0\}$ and $\lambda\in\mathbb{R}^*_+$, we have $U(x,\lambda\eta)=\lambda U(x,\eta)$) and such that the matrix $\left(\nabla^2_{x,\eta}U\right)(x,\eta)$ is non-singular for any $x\in\Rd$ and $\eta\in\Rd\setminus\{0\}$; then
\begin{equation}\label{1.23}
 \Lambda:=\left\{\left(x,\big(\nabla_xU\big)(x,\eta),\big(\nabla_\eta U\big)(x,\eta),\eta\right)\,\mid\,x\in\Rd,\eta\in\Rd\setminus\{0\}\right\}
\end{equation}
is a homogeneous canonical relation from $\mathbb{T}^*\Rd\setminus0$ to $\mathbb{T}^*\Rd\setminus0$. Such a function $U$ is called {\it a homogeneous generating function} for $\Lambda$ and it is unique.
\end{enumerate}
Let us remark that in Example 4, the submanifold $\Lambda$ is in fact contained in $\mathbb{T}^*\Rd\setminus0\times\mathbb{T}^*\Rd\setminus0$; suppose $x_0\in\Rd$, $\eta_0\in\Rd\setminus\{0\}$, then $\big(\nabla_xU\big)(x_0,\eta_0)=0$ would imply $\left(\nabla^2_{x,\eta}U\right)(x_0,\eta_0)\eta_0=0$ ($U$ being homogeneous of order 1 in the second variable) that is false by hypothesis.

\begin{definition}\label{D1.23}
 If a $C^\infty$-submanifold $\Lambda\subset\mathbb{T}^*\mathbb{R}^{2d}$ verifies simultaneously the conditions in Examples 1 and 3 (resp. 2 and 4) we say that $U$ is the {\it generating function of the symplectomorphism} $\Phi$.
\end{definition}

\begin{lemma}\label{L1.24}
 
\begin{enumerate}
 \item Suppose $U:\Rd\times\Rd\rightarrow\mathbb{R}$ is a $C^\infty$ function such that 
$$
\underset{x\in\Rd,\eta\in\Rd}{\inf}\left\|\left(\nabla^2_{x,\eta}U\right)(x,\eta)\right\|>0.
$$
Then $U$ is the generating function of a symplectomorphism $\Phi\equiv(\Phi_1,\Phi_2):\mathbb{T}^*\Rd\rightarrow\mathbb{T}^*\Rd$.
\item Reciprocally, if $\Phi\equiv(\Phi_1,\Phi_2):\mathbb{T}^*\Rd\rightarrow\mathbb{T}^*\Rd$ is a symplectomorphism verifying
$$
\underset{y\in\Rd,\eta\in\Rd}{\inf}\left\|\left(\nabla_{y}\Phi_1\right)(y,\eta)\right\|>0
$$
then $\Phi$ admits a generating function.
\end{enumerate}
\end{lemma}
\begin{proof}
 
\noindent
{\it 1.} Using the Schwartz global inversion theorem, the equality $y=\big(\nabla_\eta U\big)(x,\eta)$ defines a unique implicit function $x=\Phi_1(y,\eta)$, $\Phi_1\in C^\infty(\mathbb{T}^*\Rd)$ such that $y=\big(\nabla_\eta U\big)(\Phi_1(y,\eta),\eta)$.

Then let us define $\Phi_2(y,\eta):=\big(\nabla_xU\big)(\Phi_1(y,\eta),\eta)$, that will have components of class $C^\infty(\mathbb{T}^*\Rd)$. Due to our hypothesis it is easy to prove that $\Phi:=(\Phi_1,\Phi_2):\mathbb{T}^*\Rd\rightarrow\mathbb{T}^*\Rd$ is a $C^\infty$-diffeomorphism. Moreover with these definitions we have
$$
\left\{\big(\Phi_1(y,\eta),\Phi_2(y,\eta),y,\eta\big)\right\}=\left\{(x,\big(\nabla_xU\big)(x,\eta),\big(\nabla_\eta U\big)(x,\eta),\eta)\right\}=\Lambda
$$
of the example 3. Thus $\Lambda$ is a canonical relation and this implies that $\Phi$ is a symplectomorphism.
 
\noindent
{\it 2.} Using once again the global inversion theorem, we deduce that the equality $\Phi_1(y,\eta)=x$ defines an implicit $C^\infty$ function $y=f(x,\eta)$. Let us denote $g(x,\eta):=\Phi_2\big(f(x,\eta),\eta\big)$, that will have components of class $C^\infty$.
We have to find a real function $U\in C^\infty(\Rd\times\Rd)$ verifying
$$
\nabla_xU=g,\quad\nabla_\eta U=f,\quad\text{on}\ \Rd\times\Rd,
$$
or equivalently
$$
\d U=\sum_{1\leq j \leq d}\left\{g_jdx_j+f_jd\eta_j\right\},\quad\text{on}\ \Rd\times\Rd.
$$
But $\Phi$ being a symplectomorphism, $\text{\sf graph}\Phi=\left\{\left(x,g(x,\eta),f(x,\eta),\eta\right)\,\mid\,x\in\Rd,\eta\in\Rd\right\}
$ is a canonical relation, so that the right-hand side above is a closed 1-form and thus the Poincar\'{e} Lemma finishes the proof of the existence of $U$. In fact, the above manifold being a canonical relation we have
$$
0=\sum_{1\leq j \leq d}\left[\d g_j\wedge dx_j-\eta_j\wedge\d f_j\right],
$$
and thus
$$
\d\left\{\sum_{1\leq j \leq d}\left[g_jdx_j+f_jd\eta_j\right]\right\}=\sum_{1\leq j \leq d}\left[\d g_j\wedge dx_j+\d f_j\wedge \eta_j\right]=\sum_{1\leq j \leq d}\left[\d g_j\wedge dx_j-\eta_j\wedge\d f_j\right]=0.
$$
\end{proof}

\subsubsection{The Hamiltonian system}

Let us start with a real elliptic symbol $a\in S^m(\mathbb{R}^d)$, for $0<m\leq1$ to which we associate the following Hamiltonian system:
\begin{equation}\label{1.24}
 \left\{
\begin{array}{rcl}
 \dot{X}(t)&=&H_a\big(X(t)\big),\quad\forall t\in\mathbb{R},\\
&&\\
 X(0)&=&Y\in\mathbb{T}^*\mathbb{R}^d,
\end{array}
\right.
\end{equation}
where we have used the notations $Y:=(y,\eta)$, $X(t):=\big(x(t),\xi(t)\big)$, and $H_a:=\big(\nabla_\xi a,-\nabla_x a\big)$ (the Hamiltonian field associated to $a$ by the symplectic form $\sigma$). It is wellknown that for any $Y\in\mathbb{T}^*\mathbb{R}^d$ the system \eqref{1.24} has a unique global solution $X(t;Y)$ and the map 
$$
\mathbb{R}\times\mathbb{T}^*\mathbb{R}^d\ni(t,Y)\mapsto X(t;Y)\in\mathbb{T}^*\mathbb{R}^d
$$
is of class $C^\infty$ (see for example \cite{Ro,DG,IMP2}). The Hamiltonian flow of $a$ defined as:
$$
\Phi_t:\mathbb{T}^*\mathbb{R}^d\rightarrow\mathbb{T}^*\mathbb{R}^d,\quad\Phi_t(Y):=X(t;Y),
$$
is a symplectomorphism for any $t\in\mathbb{R}$.

\begin{lemma}\label{L1.25}
 The solution $X(t;Y)$ above has the properties:
\begin{enumerate}
 \item $a\big(X(t;Y)\big)=a(Y)$ for any $t\in\mathbb{R}$ and any $Y\in\mathbb{T}^*\mathbb{R}^d$.
\item There exists a constant $C>0$ such that $\left(<\xi(t;Y)>/<\eta>\right)^{\pm1}\leq C$, for any $t\in\mathbb{R}$ and any $Y\in\mathbb{T}^*\mathbb{R}^d$.
\item there exists two functions $b\in C^\infty\big(\mathbb{R};S^0(\mathbb{R}^d)\big)$ and $c\in C^\infty\big(\mathbb{R};S^1(\mathbb{R}^d)\big)$ such that
$$
x(t;Y)=y+tb(t;Y),\ \xi(t;Y)=\eta+tc(t;Y),\ \forall t\in\mathbb{R},\ \forall Y\in\mathbb{T}^*\mathbb{R}^d.
$$
\end{enumerate}
\end{lemma}
\begin{proof}
 The first two properties are very easy and have been proved in \cite{IMP2}. For the third property we use another result of \cite{IMP2}, saying that:
$$
\forall T>0,\ \forall(\alpha,\beta)\in[\mathbb{N}^d]^2,\text{ with }|\alpha+\beta|\geq1,\ \exists C(T,\alpha,\beta)>0\text{ such that:}
$$
\begin{equation}\label{1.26}
 \left|\left(\partial^\alpha_y\partial^\beta_\eta x\right)(t;Y)\right|+<\eta>^{-1}\left|\left(\partial^\alpha_y\partial^\beta_\eta\xi\right)(t;Y)\right|\leq C(T,\alpha,\beta)<\eta>^{-|\beta|},\quad\forall Y\in\mathbb{T}^*\mathbb{R}^d,\ \forall t\in[-T,T].
\end{equation}
Taking into account point (2) of this Lemma, the above inequality \eqref{1.26} and the following identities:
\begin{equation}\label{1.27}
 \left\{
\begin{array}{rcl}
 x(t;Y)&=&y+\int_0^t\big(\nabla_\xi a\big)\big(x(s;Y),\xi(s;Y)\big)ds,\\
&&\\
\xi(t;Y)&=&\eta-\int_0^t\big(\nabla_x a\big)\big(x(s;Y),\xi(s;Y)\big)ds,
\end{array}
\right.
\end{equation}
the point (3) in our Lemma follows from the evident identity: $\int_0^t\omega(s;Y)ds=t\int_0^1\omega(ts;Y)ds$ valid for any $\omega\in \mathcal{C}^\infty(\mathbb{R};S^m(\mathbb{R}^d))$, $m\in\mathbb{R}$.
\end{proof}

\begin{remark}\label{R1.26}
 In the case when $a$ does not depend on $x$, the solution of the Hamiltonian system \eqref{1.24} is
\begin{equation}\label{1.28}
x(t;Y)=y+t\big(\nabla_\xi a\big)(\eta),\quad\xi(t;Y)=\eta.
\end{equation}
\end{remark}

\begin{hypothesis}\label{H1.27}
We shall consider now the case of a homogeneous symbol of degree 1. More precisely we shall suppose that $a\in C^\infty\big(\mathbb{R}^d\times\big(\mathbb{R}^d\setminus\{0\}\big)\big)$ is real and such that
$$
a(x,\lambda\xi)=\lambda a(x,\xi)\ \forall\lambda>0,\ \forall(x,\xi)\in \mathbb{R}^d\times\big(\mathbb{R}^d\setminus\{0\}\big),
$$
$$
\underset{x\in\mathbb{R}^d}{\inf}\underset{|\xi|=1}{\inf}a(x,\xi)>0,
$$
$$
\underset{x\in\mathbb{R}^d}{\sup}\underset{|\xi|=1}{\sup}\left|\left(\partial^\alpha_x\partial^\beta_\xi a\right)(x,\xi)\right|<\infty,\quad\forall(\alpha,\beta)\in [\mathbb{N}^d]^2.
$$
\end{hypothesis}
\begin{lemma}\label{L1.27}
 Under the above assumptions of Hypothesis \ref{H1.27} on the symbol $a$, we have that:
\begin{enumerate}
 \item For any point $Y\in\mathbb{T}^*\mathbb{R}^d\setminus0$ the Hamiltonian system \eqref{1.24} has a unique global solution $X(t;Y)=\big(x(t;Y),\xi(t,Y)\big)$ such that $\xi(t;Y)\ne0$ and the map 
$$
\mathbb{R}\times\mathbb{T}^*\mathbb{R}^d\setminus0\ni(t,Y)\mapsto X(t;Y)\in\mathbb{T}^*\mathbb{R}^d\setminus0
$$
is of class $C^\infty$ 
\item The map $X(t;Y)$ commutes with the action of $\mathbb{R}^*_+$ on the fibres of $\mathbb{T}^*\mathbb{R}^d\setminus0$, i.e. 
$$
x(t;y,\lambda\eta)=x(t;y,\eta),\ \xi(t;y,\lambda\eta)=\lambda\xi(t;y,\eta),\quad\forall\lambda>0,\ \forall(x,\xi)\in\mathbb{T}^*\mathbb{R}^d\setminus0.
$$
\item The map $\Phi_t:\mathbb{T}^*\mathbb{R}^d\setminus0\rightarrow\mathbb{T}^*\mathbb{R}^d\setminus0$ given by $\Phi_t(Y):=X(t;Y)$ is a symplectomorphism for any $t\in\mathbb{R}$.
\item $\forall T>0$, $\forall k\in\mathbb{N}$ and $\forall \alpha\in\mathbb{N}^{2d}$, the map $\partial_t^k\partial_Y^\alpha\left[X(t;Y)-Y\right]$ is bounded on $[-T,T]\times\mathbb{R}^d\times\{|\eta|=1\}$.
\end{enumerate}
\end{lemma}
\begin{proof}
\noindent{\sf Point 1.} For any $Y\in\mathbb{T}^*\mathbb{R}^d\setminus0$ there exist $a_Y<0$ and $b_Y>0$ such that the Cauchy problem \eqref{1.24} has a maximal solution $X\in C^\infty\big((a_Y,b_Y)\times\big(\mathbb{T}^*\mathbb{R}^d\setminus0\big)\big)$ that stays inside $\mathbb{T}^*\mathbb{R}^d\setminus0$. We have to prove that $a_Y=-\infty$ and $b_Y=\infty$. Let us suppose for example that $b_Y<\infty$; in order to get a contradiction it will be sufficient to prove that the set $\{X(t;Y);t\in[0,b_Y)\}$ is included in a compact subset of $\mathbb{T}^*\mathbb{R}^d\setminus0$, because in this case the solution can be still continued past $b_Y$ and thus it is not a maximal solution. Let us denote by
$$
m:=\underset{x\in\mathbb{R}^d}{\inf}\underset{|\xi|=1}{\inf}a(x,\xi)>0,\quad M:=\underset{x\in\mathbb{R}^d}{\sup}\underset{|\xi|=1}{\sup}a(x,\xi)<\infty.
$$
Using the homogeneity of $a$ and the fact that $a\big(x(t;Y),\xi(t;Y)\big)=a(Y)$ for any $t\in(a_Y,b_Y)$ we deduce that
$$
\frac{m}{M}|\eta|\,\leq\,|\xi(t;Y)|\,\leq\,\frac{M}{m}|\eta|,\quad\forall t\in(a_Y,b_Y).
$$
On the other side we have that
$$
x(t;Y)\,=\,y\,+\,\int_0^t\big(\nabla_\xi a\big)\big(X(s;Y)\big)ds,
$$
so that it exists a constant $C(Y,b_Y)$ such that $|x(t;Y)|\leq C(Y,b_Y)$ for any $t\in[0,b_Y)$. We conclude thus that $\{X(t;Y);t\in[0,b_Y)\}$ is included in a compact subset of $\mathbb{T}^*\mathbb{R}^d\setminus0$.

\noindent{\sf Point 2.} Let $Y\in\mathbb{T}^*\mathbb{R}^d\setminus0$, $\lambda>0$ and let $X(t;Y)=\big(x(t;Y),\xi(t;Y)\big)$ be the unique solution of the Cauchy problem \eqref{1.24}. Then the function $\big(x(t;y,\lambda\eta),\lambda^{-1}\xi(t;y,\lambda\eta)\big)$ will also be a solution of the same problem due to the homogeneity of $a$. In conclusion we get
$$
x(t;y,\lambda\eta)=x(t;y,\eta),\quad \lambda^{-1}\xi(t;y,\lambda\eta)=\xi(t;y,\eta).
$$

\noindent{\sf Point 3.} is trivial.

\noindent{\sf Point 4.} We use the equalities \eqref{1.27}. The first implies that $x(t;Y)-y$ is a bounded function on $[-T,T]\times\mathbb{R}^d\times\{|\eta|=1\}$ and from the second equality in \eqref{1.27} we deduce the existence of two positive constants $C_1$ and $C_2$ such that
$$
\left|\xi(t;Y)\right|\,\leq\,C_1\,+\,C_2\left|\int_0^t\left|\xi(s;Y)\right|ds\right|,\quad\forall t\in\mathbb{R},\ \forall Y\in\mathbb{T}^*\mathbb{R}^d\setminus0,\ |\eta|=1.
$$
Let us suppose now that $t>0$ and use the Gronwall Lemma in order to conclude that also the function $\xi(t;Y)-\eta$ is bounded on $[-T,T]\times\mathbb{R}^d\times\{|\eta=1\}$. For the derivatives we differentiate the equalities in \eqref{1.27} and proceed by induction.
\end{proof}

\subsection{The Hamilton-Jacobi equation}\label{1.f}

Let us take $a\in S^1(\mathbb{R}^d)$ a real, elliptic symbol.
\begin{proposition}\label{P1.28}
 There exists $T>0$ such that the following Cauchy problem
\begin{equation}\label{1.29}
 \left\{
\begin{array}{rcl}
 U&\in&C^\infty\big((-T,T)\times\mathbb{R}^{2d}\big),\text{ real function},\\
&&\\
\partial_tU(t;x,\eta)&=&-a\left(x,\nabla_xU(t;x,\eta)\right),\ \forall t\in(-T,T),\ \forall(x,\eta)\in\mathbb{R}^{2d},\\
&&\\
U(0;x,\eta)&=&<x,\eta>,
\end{array}
\right.
\end{equation}
has a unique solution. Moreover, for any $t\in(-T,T)$ and for any $(x,\eta)\in\mathbb{R}^{2d}$ the following equalities are true:
\begin{equation}\label{1.30}
 \nabla_xU(t;x(t;Y),\eta)=\xi(t;Y),
\end{equation}
\begin{equation}\label{1.31}
 \nabla_\eta U(t;x(t;Y),\eta)=y,
\end{equation}
where $\big(x(t;Y),\xi(t;Y)\big)$ is the unique solution of the Cauchy problem \eqref{1.24}.
\end{proposition}
\begin{proof}
 We follow the standard proof (see for example \cite{Ro,DG}). First one verifies the unicity of the solution of the Cauchy problem \eqref{1.29} obtaining also an 'ansatz' for the solution.

So let us suppose that $U$ is a solution of \eqref{1.29}. We consider now the following Cauchy problem with a parameter $Y\in\mathbb{T}^*\mathbb{R}^d$:
\begin{equation}\label{1.32}
 \left\{
\begin{array}{rcl}
 z&\in&C^\infty\big((-T,T)\times\mathbb{T}^*\mathbb{R}^d\big),\\
&&\\
\dot{z}(t;Y)&=&\big(\nabla_\xi a\big)\left(z(t;Y),\nabla_xU(t;z(t;Y),\eta)\right),\\
&&\\
z(0;Y)&=&y.
\end{array}
\right.
\end{equation}
Due to the fact that $\nabla_\xi a$ is a bounded function it follows that the Cauchy problem \eqref{1.32} has a unique global solution. We denote by $\zeta(t;Y):=\big(\nabla_xU\big)(t,z(t;Y),\eta)$. Then $\zeta\in C^\infty\big((-T,T)\times\mathbb{T}^*\mathbb{R}^d\big)$ and by differentiation (and taking into account \eqref{1.32}) we get
\begin{equation}\label{1.33}
 \dot{\zeta}(t;Y)=\big(\partial_t\nabla_xU\big)(t,z(t;Y),\eta)+\big(\nabla^2_{x,x}U\big)(t,z(t;Y),\eta)\cdot\big(\nabla_\xi a\big)\left(z(t;Y),\nabla_xU(t;z(t;Y),\eta)\right).
\end{equation}
By differentiating \eqref{1.29} with respect to $x$ we get
\begin{equation}\label{1.34}
 \big(\partial_t\nabla_xU\big)(t;x,\eta)+\big(\nabla_x a\big)\left(x,\nabla_xU(t;x,\eta)\right)+\big(\nabla^2_{x,x}U\big)(t;x,\eta)\cdot\big(\nabla_\xi a\big)\left(x,\nabla_xU(t;x,\eta)\right)=0.
\end{equation}
Comparing these last two equations we get
\begin{equation}\label{1.35}
 \dot{\zeta}(t;Y)=-\big(\nabla_x a\big)\big(z(t;Y),\zeta(t;Y)\big).
\end{equation}
Using \eqref{1.32} and \eqref{1.35} and the initial condition in \eqref{1.29} we conclude that the pair of functions $\big(z(t;Y),\zeta(t;Y)\big)$ verifies the following Hamiltonian system:
\begin{equation}\label{1.36}
 \left\{
\begin{array}{rcl}
 \dot{z}(t;Y)&=&\big(\nabla_\xi a\big)\left(z(t;Y),\zeta(t;Y)\right),\\
&&\\
\dot{\zeta}(t;Y)&=&-\big(\nabla_x a\big)\big(z(t;Y),\zeta(t;Y)\big),\\
&&\\
z(0;Y)=y,&&\zeta(0;Y)=\big(\nabla_xU\big)(0;y,\eta)=\eta.
\end{array}
\right.
\end{equation}
Comparing with \eqref{1.24} and taking into account the unicity of the solution we conclude that
\begin{equation}\label{1.37}
 z(t;Y)=x(t;Y),\quad\zeta(t;Y)=\xi(t;Y).
\end{equation}
In particular, taking into account the definition of the function $\zeta$ we deduce that for any solution $U$ of problem \eqref{1.29} the following equality is verified:
\begin{equation}\label{1.38}
 \big(\nabla_xU\big)(t;x(t;Y),\eta)=\xi(t;Y),
\end{equation}
and thus \eqref{1.30} is true.

Let us compute now
$$
\frac{d}{dt}\left[U\big(t;x(t;Y),\eta\big)\right]=\big(\partial_tU\big)\big(t;x(t;Y),\eta\big)+\left<\big(\nabla_xU\big)(t;x(t;Y),\eta),\dot{x}(t;Y)\right>.
$$
Using \eqref{1.29}, \eqref{1.38} and \eqref{1.24} we deduce that
\begin{equation}\label{1.39}
 \frac{d}{dt}\left[U\big(t;x(t;Y),\eta\big)\right]=-a\big(x(t;Y),\xi(t;Y)\big)+\left<\xi(t;Y),\big(\nabla_\xi a\big)\big(x(t;Y),\xi(t;Y)\big)\right>,
\end{equation}
or by integration
\begin{equation}\label{1.40}
 U\big(t;x(t;Y),\eta\big)=<y,\eta>+\int_0^t\left[\left<\xi(s;Y),\big(\nabla_\xi a\big)\big(x(s;Y),\xi(s;Y)\big)\right>-a\big(x(s;Y),\xi(s;Y)\big)\right]ds.
\end{equation}
We denote the right hand side of the equality \eqref{1.40} by $Q(t;y,\eta)$. Using (3) of Lemma \ref{L1.25} we get that
$$
\forall\delta\in[0,1),\ \exists T>0\text{ such that }\left\|\boldsymbol{1}_d-\big(\nabla_yx\big)(t;Y)\right\|\leq\delta\ \forall t\in(-T,T),\ \forall Y\in\mathbb{T}^*\mathbb{R}^d.
$$
Using the Schwartz global inverse theorem we deduce that the equation $x(t;Y)=x$ has a unique solution $y=f(t;x,\eta)$ for any $t\in(-T,T)$ and any $(x,\eta)\in\mathbb{T}^*\mathbb{R}^d$ and this solution belongs to $C^\infty\big((-T,T)\times\mathbb{R}^{2d}\big)$. Due to \eqref{1.40}, for $T>0$ given above, the problem \eqref{1.29} may have at most one solution and this must be given by
\begin{equation}\label{1.41}
 U(t;x,\eta)=Q\big(t;f(t;x,\eta),\eta\big).
\end{equation}

Let us verify now that the function $U$ given by \eqref{1.41} is indeed a solution of the problem \eqref{1.29}. The condition $U\in C^\infty\big((-T,T)\times\mathbb{R}^{2d}\big)$ is evidently satisfied. We also have that
$$
\partial_tQ(t;y,\eta)=\left<\xi(t;Y),\big(\nabla_\xi a\big)\big(x(t;Y),\xi(t;Y)\big)\right>-a\big(x(t;Y),\xi(t;Y)\big).
$$
Differentiating with respect to $y$ and taking into account the first equation in \eqref{1.24} we get
$$
\partial_t\nabla_yQ(t;y,\eta)=-\nabla_yx(t;Y)\cdot\big(\nabla_x a\big)\big(x(t;Y),\xi(t;Y)\big)-\nabla_y\xi(t;Y)\cdot\big(\nabla_\xi a\big)\big(x(t;Y),\xi(t;Y)\big)+
$$
$$
+\nabla_y\xi(t;Y)\cdot\big(\nabla_\xi a\big)\big(x(t;Y),\xi(t;Y)\big)+\nabla_y\dot{x}(t;Y)\cdot\xi(t;Y).
$$
Taking into account \eqref{1.24} and the identity $\nabla_y\dot{x}(t;Y)=\partial_t\big(\nabla_y x\big)$, we get
$$
\partial_t\nabla_yQ(t;y,\eta)=\partial_t\left[\nabla_yx(t;Y)\cdot\xi(t;Y)\right].
$$
Moreover we evidently have that
$$
\nabla_yQ(0;y,\eta)=\eta=\nabla_yx(0;Y)\cdot\xi(0;Y),
$$
so that
\begin{equation}\label{1.42}
 \nabla_yQ(t;y,\eta)=\nabla_yx(t;Y)\cdot\xi(t;Y).
\end{equation}
Puting now together \eqref{1.41} and \eqref{1.42} we obtain that:
$$
\nabla_xU(t;x,\eta)=\big(\nabla_x f\big)(t;x,\eta)\cdot\big(\nabla_{y}Q\big)\big(t;f(t;x,\eta),\eta\big)=
$$
$$
=\big(\nabla_x f\big)(t;x,\eta)\cdot\big(\nabla_{y}x\big)\big(t;f(t;x,\eta),\eta\big)\,\xi\big(t;f(t;x,\eta),\eta\big).
$$
By the definition of the function $f$ we have that $x\big(t;f(t;x,\eta),\eta\big)=x$ and thus 
$$
\big(\nabla_x f\big)(t;x,\eta)\cdot\big(\nabla_yx\big)\big(t;f(t;x,\eta),\eta\big)=1_d
$$
and thus
\begin{equation}\label{1.43}
\nabla_xU\big(t;x(t;Y),\eta\big)=\xi(t;Y).
\end{equation}
In conclusion the function $U$ defined by \eqref{1.41} verifies \eqref{1.30}. The initial condition of \eqref{1.29} follows imediatly from \eqref{1.41} and the definition of $Q$ above. In order to also verify the equation we notice that
\begin{equation}\label{1.44}
 \partial_tQ(t;y,\eta)=-a\big(x(t;Y),\xi(t;Y)\big)+\left<\xi(t;Y),\dot{x}(t;Y)\right>.
\end{equation}
But we also have that $Q(t;y,\eta)=U\big(t;x(t;Y),\eta\big)$ so that using \eqref{1.43} we deduce that
$$
\partial_tQ(t;y,\eta)=\big( \partial_tU\big)\big(t;x(t;Y),\eta\big)+\left<\big( \nabla_xU\big)\big(t;x(t;Y),\eta\big),\dot{x}(t;Y)\right>=
$$
$$
=\big( \partial_tU\big)\big(t;x(t;Y),\eta\big)+\left<\xi(t;Y),\dot{x}(t;Y)\right>.
$$
Comparing this equality with \eqref{1.44} and using again \eqref{1.43} we conclude that
$$
\big( \partial_tU\big)\big(t;x(t;Y),\eta\big)=-a\big(x(t;Y),\xi(t;Y)\big)=-a\big(x(t;Y),\big( \nabla_xU\big)\big(t;x(t;Y),\eta\big)\big)
$$
and this implies now \eqref{1.29}.

Considering the equality\eqref{1.31}, it is verified for $t=0$ because we have $\nabla_\eta U(0;y,\eta)=y$. It remains to prove that
\begin{equation}\label{1.45}
 \frac{d}{dt}\left[\big(\nabla_\eta U\big)\big(t;x(t;Y),\eta\big)\right]=0.
\end{equation}
We develop the left hand side of the above equality obtaining:
$$
\big(\partial_t\nabla_\eta U\big)\big(t;x(t;Y),\eta\big)+\big(\nabla^2_{\eta,x} U\big)\big(t;x(t;Y),\eta\big)\cdot\dot{x}(t;Y).
$$
From \eqref{1.43} and \eqref{1.24} it follows that:
$$
\dot{x}(t;Y)=\big(\nabla_\xi a\big)\left(x(t;Y),\big( \nabla_xU\big)\big(t;x(t;Y),\eta\big)\right).
$$
Differentiating \eqref{1.29} with respect to $\eta$ we obtain
$$
\big(\partial_t\nabla_\eta U\big)\big(t;x,\eta\big)+\big(\nabla^2_{\eta,x} U\big)\big(t;x,\eta\big)\cdot\big(\nabla_\xi a\big)\left(x,\big( \nabla_xU\big)\big(t;x,\eta\big)\right)=0.
$$
The last three equalities imply \eqref{1.45}.
\end{proof}

\begin{remark}\label{R1.29}
 From the proof of Proposition \ref{P1.28} it follows that the parameter $T$ in the statement of Proposition \ref{P1.28} is chosen in order to verify the following inequality
\begin{equation}\label{1.46}
\underset{t\in(-T,T)}{\sup}\,\underset{Y\in\mathbb{T}^*\mathbb{R}^d}{\sup}\left\|1_d-\big(\nabla_y x\big)(t;Y)\right\|\leq1,
\end{equation}
where $\big(x(t;Y),\xi(t;Y)\big)$ is the solution of the canonical system \eqref{1.24}. In particular, if $a$ does not depend on the variable $x$, as in Remark \ref{R1.26}, we have $x(t;Y)=y+t\big(\nabla_\xi a\big)(\eta)$ and $\nabla_y x(t;Y)=1_d$, and thus in this case we can take $T=\infty$. In fact one can easily notice that in this case the solution of problem \eqref{1.29} is given by
\begin{equation}\label{1.47}
 U(t;x,\eta)=<x,\eta>-ta(\eta).
\end{equation}
\end{remark}

\begin{remark}\label{R1.30}
 If we fix $t\in(-T,T)$ with $T$ chosen in order to satisfy \eqref{1.46}, the equalities \eqref{1.30} and \eqref{1.31} mean that $U(t;x,\eta)$ is generating function for the symplectomorphism $\Phi_t$ defined by the canonical system \eqref{1.24}.
\end{remark}

In fact, a representation of $U$ of the form \eqref{1.47} may also be obtained in the general case.

\begin{lemma}\label{L1.31}
 Under the assumptions of Proposition \ref{P1.28}, there exists a function $d\in C^\infty\big((-T,T);S^1(\mathbb{R}^d)\big)$ such that
\begin{equation}\label{1.48}
 U(t;x,\eta)=<x,\eta>+t\,d(t;x,\eta),\quad\text{on }(-T,T)\times\mathbb{R}^{2d}.
\end{equation}
\end{lemma}
\begin{proof}
 We shall use the notations introduced in the proof of Proposition \ref{P1.28} and denote by
$$
e(t;x,\eta):=-a\left(x,\xi\big(t;f(t;x,\eta),\eta\big)\right).
$$
Following \eqref{1.29} and \eqref{1.30} we have that
$$
\partial_tU(t;x,\eta)=e(t;x,\eta)
$$
so that we obtain
\begin{equation}\label{1.49}
 U(t;x,\eta)=<x,\eta>+\int_0^t\,e(s;x,\eta)\,ds.
\end{equation}
As in the final part of the proof of Lemma \ref{L1.25}, the conclusion of this Lemma follows once we have proved that $e\in C^\infty\big((-T,T);S^1(\mathbb{R}^d)\big)$.

From the first point of Lemma \ref{L1.25} we deduce that
$$
e\big(t;x(t;Y),\eta\big)=-a\big(x(t;Y),\xi(t;Y)\big)=-a(Y),
$$
so that
\begin{equation}\label{1.50}
e\big(t;x,\eta\big)=-a\big(f(t;x,\eta),\eta\big).
\end{equation}
Differentiating the equality $x\big(t;f(t;x,\eta),\eta\big)=x$ we obtain the equalities:
\begin{equation}\label{1.51}
 \left\{
\begin{array}{rcl}
\big( \nabla_xf\big)(t;x,\eta)&=&\left[\big(\nabla_yx\big)\big(t;f(t;x,\eta),\eta\big)\right]^{-1},\\
&&\\
\big( \nabla_\eta f\big)(t;x,\eta)&=&-\big(\nabla_\eta x\big)\big(t;f(t;x,\eta),\eta\big)\cdot\left[\big(\nabla_yx\big)\big(t;f(t;x,\eta),\eta\big)\right]^{-1}.
\end{array}
\right.
\end{equation}
Differentiating further the equalities \eqref{1.51} and proceeding by induction using also point (3) of Lemma \ref{L1.25} one easily proves that 
$$
\partial_x^\alpha\partial_\xi^\beta f\in C^\infty\big((-T,T);S^{-|\beta|}(\mathbb{R}^d)\big),\quad\forall(\alpha,\beta)\in[\mathbb{N}^d]^2,\text{ with }|\alpha+\beta|\geq1.
$$
The equality \eqref{1.50} means then that $e\in C^\infty\big((-T,T);S^1(\mathbb{R}^d)\big)$.
\end{proof}

Let us consider now the case of a {\it homogeneous symbol} of degree 1 and having the properties assumed in Lemma \ref{L1.27}. Some evident modifications of the above arguments allow to prove the following Proposition.

\begin{proposition}\label{P1.32}
 Under the assumptions of Hypothesis \ref{H1.27} one has the following:
\begin{enumerate}
 \item There exists $T>0$ and a real function $U\in C^\infty\big((-T,T)\times\mathbb{R}^d\times(\mathbb{R}^d\setminus\{0\})\big)$ such that
\begin{equation}\label{1.52}
 \left\{
\begin{array}{rcl}
 \big(\partial_tU\big)(t;x,\eta)&=&-a\left(x,\big(\nabla_xU\big)(t;x,\eta)\right),\quad \forall t\in(-T,T),\ \forall(x,\eta)\in\mathbb{T}^*\mathbb{R}^d,\ \eta\ne0,\\
&&\\
U(0;x,\eta)&=&<x,\eta>.
\end{array}
\right.
\end{equation}
\item $U(t;x,\lambda\eta)=\lambda U(t;x,\eta),\quad\forall\lambda>0,\ \forall t\in(-T,T),\ \forall(x,\eta)\in\mathbb{T}^*\mathbb{R}^d,\ \eta\ne0$.
\item $\big(\nabla_xU\big)\big(t;x(t;Y),\eta\big)=\xi(t;Y),\quad\forall t\in(-T,T),\ \forall(x,\eta)\in\mathbb{T}^*\mathbb{R}^d,\ \eta\ne0$, where $\Phi_t(Y)\equiv\big(x(t;Y),(\xi(t;Y)\big)$ is the solution of the canonical system obtained in Lemma \ref{L1.27}.
\item $\big(\nabla_\eta U\big)\big(t;x(t;Y),\eta\big)=y,\quad\forall t\in(-T,T),\ \forall(x,\eta)\in\mathbb{T}^*\mathbb{R}^d,\ \eta\ne0$.
\item $U(t;x,\eta)=<x,\eta>+t\,d(t;x,\eta),\quad\forall t\in(-T,T),\ \forall(x,\eta)\in\mathbb{T}^*\mathbb{R}^d,\ \eta\ne0$, with $d$ a function of class $C^\infty$ on its domain of definition, positive homogeneous of degree 1 in the variable $\eta$, and for any $T_0\in(0,T)$, for any $k\in\mathbb{N}$ and for any $(\alpha,\beta)\in[\mathbb{N}^d]^2$, the function $\partial_t^k\partial_x^\alpha\partial_\xi^\beta d$ is bounded on $[-T_0,T_0]\times\mathbb{R}^d\times\{|\eta|=1\}$.
\item For any $t\in(-T,T)$, $U(t;x,\eta)$ is a generating function for the symplectomorphism $\Phi_t$.
\end{enumerate}
\end{proposition}

\begin{remark}\label{R1.33}
 Evidently the Remarks \ref{R1.26} and \ref{R1.29} remain true in the homogeneous case: if $a(x,\xi)=a(\xi)$ is a real, strictly positive  function of class $C^\infty(\mathbb{R}^d\setminus\{0\})$, that is homogeneous of degree 1, then $T=\infty$ and
\begin{equation}\label{1.53}
 x(t;Y)=y+t\big(\nabla_\xi a\big)(\eta),\ \xi(t;Y)=\eta,\ U(t;x,\eta)=<x,\eta>-t\,a(\eta),\eta\ne0.
\end{equation}
\end{remark}

\begin{lemma}\label{L1.34}
 Suppose the symbol $a$ verifies Hypothesis \ref{H1.27} and let $\chi\in C^\infty(\mathbb{R}^d)$ be a real function with $\chi(\xi)=0$ for $|\xi|\leq1$ and with $\chi(\xi)=1$ for $|\xi|\geq2$. We denote by $b(x,\xi):=\chi(\xi)a(x,\xi)$. We denote by $U_b$, respectively by $U_a$, the solutions of the Hamilton-Jacobi equations \eqref{1.29} (for $b$), respectively \eqref{1.52}. Then there exist $R>0$ and $T>0$ such that $U_a(t;x,\eta)$ and $U_b(t;x,\eta)$ exist both and are equal for any $t$ with $|t|<T$ and for any $(x,\eta)\in\mathbb{R}^{2d}$ with $|\eta|\geq R$.
\end{lemma}
\begin{proof}
 Let us denote by $X_a=(x_a,\xi_a)$, respectively by $X_b=(x_b,\xi_b)$, the solutions of the canonical system \eqref{1.24} associated to $a$, and respectively to $b$, with initial data in $\mathbb{T}^*\mathbb{R}^d\setminus0$. In the proof of Lemma \ref{L1.27} we proved that it exists $\epsilon\in(0,1]$ such that $|\xi_a(t;Y)|\geq\epsilon|\eta|$ for any $t\in\mathbb{R}$ and for any $Y\in\mathbb{T}^*\mathbb{R}^d\setminus0$. Let us choose $R:=2/\epsilon\geq2$; we have $|\xi_a(t;Y)|\geq2$ for any $t\in\mathbb{R}$ and for any $Y\in\mathbb{T}^*\mathbb{R}^d\setminus0$ with $|\eta|\geq R$. But for any $X\in\mathbb{T}^*\mathbb{R}^d\setminus0$ with $|\xi|\geq2$ we have $a(X)=b(X)$ and thus $H_a(X)=H_b(X)$. Then any integral curve of $H_a$, passing at $t=0$ through $Y=(y,\eta)\in\mathbb{T}^*\mathbb{R}^d$ with $|\eta|\geq R$ remains in the domain $\{|\xi|\geq2\}$ and will thus also be integral curve for $H_b$. It follows that $X_a(t;Y)=X_b(t;Y)$ for any $t\in\mathbb{R}$ and any $Y\in\mathbb{T}^*\mathbb{R}^d$ with $|\eta|\geq R$, due to the fact that they coincide for $t=0$. The proof is finished by following the construction of $U$ in Proposition \ref{P1.28}, choosing $T>0$ such that
$$
\underset{|t|<T}{\sup}\,\underset{y\in\mathbb{R}^d}{\sup}\,\underset{|\eta|=1}{\sup}\left\|1_d-\big(\nabla_y x_a\big)(t;Y)\right\|=\underset{|t|<T}{\sup}\,\underset{y\in\mathbb{R}^d}{\sup}\,\underset{|\eta|\geq2}{\sup}\left\|1_d-\big(\nabla_y x_b\big)(t;Y)\right\|< 1.
$$
\end{proof}

\section{Magnetic Fourier Integral Operators: definition and general properties}

\begin{hypothesis}\label{Hyp-U}
 We shall always work with a real function $U\in C^\infty(\mathbb{R}^d\times\mathbb{R}^d)$ verifying the following properties:
\begin{enumerate}
 \item $U(x,\eta)=<x,\eta>+\,d(x,\eta),\quad\text{with }d\in S^+\text{ real},\ \forall(x,\eta)\in\mathbb{R}^d\times\mathbb{R}^d$,
\item $\exists\delta\in[0,1)\text{ such that }\left\|\big(\nabla^2_{x,\eta}d\big)(x,\eta)\right\|\leq\delta,\ \forall(x,\eta)\in\mathbb{R}^d\times\mathbb{R}^d$.
\end{enumerate}
\end{hypothesis}
\begin{remark}\label{R2.1a}
 The properties (1) and (2) in the Hypothesis \ref{Hyp-U} above imply the inequality:
\begin{equation}\label{2.1}
 \left\|\big(\nabla^2_{x,\eta}U\big)(x,\eta)\right\|\geq1-\delta>0,\ \forall(x,\eta)\in\mathbb{R}^d\times\mathbb{R}^d,
\end{equation}
so that point (1) in Lemma \ref{L1.24} implies that $U$ is a generating function for a symplectomorphism $\Phi:\mathbb{T}^*\mathbb{R}^d\rightarrow\mathbb{T}^*\mathbb{R}^d$, being uniquely determined (modulo an additive constant) by this one.
\end{remark}

\begin{remark}\label{R2.1b}
 If we consider the phase function $\phi:\mathbb{R}^{3d}\rightarrow\mathbb{R}$ given by $\phi(x,y,\eta):=U(x,\eta)-<y,\eta>$ for all $(x,y,\eta)\in\big[\mathbb{R}^d\big]^3$, then the canonical relation $\Lambda_\phi$ defined by the phase function $\phi$ coincides with ${\sf graph}\Phi$. In fact we have that:
$$
\Lambda_\phi:=\left\{\left(x,\big(\nabla_x\phi\big)(x,y,\eta),y,-\big(\nabla_y\phi\big)(x,y,\eta)\right);\ \big(\nabla_\eta\phi\big)(x,y,\eta)=0,\ \forall(x,y,\eta)\in\big[\mathbb{R}^d\big]^3\right\}=
$$
$$
=\left\{\left(x,\big(\nabla_x U\big)(x,\eta),\big(\nabla_\eta U\big)(x,\eta),\eta\right);\  \forall(x,\eta)\in\big[\mathbb{R}^d\big]^2\right\}={\sf graph}\Phi.
$$

Suppose given a magnetic field $B$ on $\mathbb{R}^d$ with components of class $BC^\infty(\mathbb{R}^d)$ to which we associate a vector potential $A$ with components of class $C^\infty_{\text{\sf pol}}(\mathbb{R}^d)$ and the function $\omega^A\in C^\infty_{\text{\sf pol}}(\mathbb{R}^{2d})$ defined as in \eqref{omegaA}.
\end{remark}

\begin{lemma}\label{L2.2}
 Let us choose $\Phi$ and $\phi$ as in the Remarks \ref{R2.1a} and \ref{R2.1b} and $a\in S^m(\mathbb{R}^d)$ for some $m\in\mathbb{R}$. Then for any $u\in\mathcal{S}(\mathbb{R}^d)$ and for any $x\in\mathbb{R}^d$ the following oscillating integral
\begin{equation}\label{2.2}
 \left[\mathfrak{Op}^A_\Phi(a)u\right](x):=\int_{\mathbb{R}^d}\int_{\mathbb{R}^d}e^{i\phi(x,y,\eta)}\omega^A(x,y)a(x,\eta)u(y)\,dy\,\dbar\eta
\end{equation}
 is well defined, $\mathfrak{Op}^A_\Phi(a)u\in\mathcal{S}(\mathbb{R}^d)$ and the map
$$
\mathcal{S}(\mathbb{R}^d)\ni u\mapsto\mathfrak{Op}^A_\Phi(a)u\in\mathcal{S}(\mathbb{R}^d)
$$
is linear and continuous.
\end{lemma}
\begin{proof}
 Let us fix some function $\psi\in C^\infty_0(\mathbb{R}^d)$ such that $\psi(\eta)=1$ for $|\eta|\leq1$ and some parameter $\epsilon\in(0,1]$. Let us define:
\begin{equation}\label{2.3}
\big( I_\epsilon u\big)(x):=\int_{\mathbb{R}^d}\int_{\mathbb{R}^d}e^{i\phi(x,y,\eta)}\omega^A(x,y)\psi(\epsilon\eta)a(x,\eta)u(y)\,dy\,\dbar\eta.
\end{equation}
We integrate by parts in the above integral, using the identities:
$$
(1-\Delta_y)^Ne^{i\phi(x,y,\eta)}=<\eta>^{2N}e^{i\phi(x,y,\eta)},\qquad(1-\Delta_\eta)^Me^{i<x-y,\eta>}=<x-y>^{2M}e^{i<x-y,\eta>},
$$
for some $N\in\mathbb{N}$ and $M\in\mathbb{N}$, in order to obtain
$$
\big( I_\epsilon u\big)(x)=
$$
$$
=\int_{\mathbb{R}^d}\int_{\mathbb{R}^d}e^{i<x-y,\eta>}<x-y>^{-2M}(1-\Delta_\eta)^M\left\{e^{id(x,\eta)}<\eta>^{-2N}(1-\Delta_y)^N\left[\omega^A(x,y)\psi(\epsilon\eta)a(x,\eta)u(y)\right]\right\}\,dy\,\dbar\eta.
$$
We choose first $N$ such that $2N>m+d$, then we choose $M$ depending on the value of $N$ and taking into account that the derivatives of $d(x,\eta)$ of order $\geq1$ with respect to $\eta$ are bounded, we apply the Dominated Convergence Theorem in order to conclude that
$$
\underset{\epsilon\searrow0}{\lim}\big(I_\epsilon u\big)(x)=\big(I_0 u\big)(x).
$$
But by definition the value of the oscillating integral \eqref{2.2} is $\big(I_0 u\big)(x)$. The properties stated in the Lemma follow in a standard way by chosing $N$ and $M$ large enough.
\end{proof}

\begin{definition}\label{D2.3}
 The operator $\mathfrak{Op}^A_\Phi(a)\in\mathbb{B}\big(\mathcal{S}(\mathbb{R}^d)\big)$ is called {\it the magnetic Fourier Integral Operator} associated to the symplectomorphism $\Phi$, the magnetic field $B$ and the symbol $a\in S^m(\mathbb{R}^d)$.
\end{definition}
\begin{lemma}\label{L2.4}
 The map
$
S^m(\mathbb{R}^d)\ni a\mapsto \mathfrak{Op}^A_\Phi(a)\in\mathbb{B}\big(\mathcal{S}(\mathbb{R}^d)\big)
$
is injective.
\end{lemma}
\begin{proof}
 The distribution kernel $K^A_{\phi,a}(x,y)$ of the operator $\mathfrak{Op}^A_\Phi(a)$ is explicitely given by
\begin{equation}\label{2.4}
 K^A_{\phi,a}(x,y):=\omega^A(x,y)\int_{\mathbb{R}^d}e^{i\phi(x,y,\eta)}a(x,\eta)\dbar\eta,
\end{equation}
as an oscillating integral with values in $\mathcal{S}^\prime(\mathbb{R}^d)$. Thus we have that
$$
\omega^A(y,x)K^A_{\phi,a}(x,y)=\mathcal{F}^{-1}_\eta\left[e^{iU(x,\cdot)}a(x,\cdot)\right](-y),
$$
and thus
\begin{equation}\label{2.5}
 a(x,\eta)=e^{-iU(x,\eta)}\mathcal{F}_y\left[\omega^A(-.,x)K^A_{\phi,a}(x,-.)\right](\eta)=\int_{\mathbb{R}^d}e^{-i\phi(x,y,\eta)}\omega^A(y,x)K^A_{\phi,a}(x,y)\,dy,
\end{equation}
with the equality in the sense of tempered distributions.
\end{proof}

\begin{definition}\label{D2.5}
We call {\it the principal symbol} of $\mathfrak{Op}^A_\Phi(a)$ any representative of the class of $a\in S^m(\mathbb{R}^d)$ in the quotient $S^m(\mathbb{R}^d)/S^{m-1}(\mathbb{R}^d)$. The number $m\in\mathbb{R}$ is called {\it the order} of $\mathfrak{Op}^A_\Phi(a)$.
\end{definition}
\begin{remark}\label{R2.6}
 If $\Phi={\sf Id}$, the identity map on $\mathbb{T}^*\mathbb{R}^d$, then $U(x,\eta)=<x,\eta>$ (modulo an additive constant) and thus we have that $\mathfrak{Op}^A_{\sf Id}(a)=a^A(x,D)$, which due to Proposition \ref{P1.14} is a magnetic $\Psi$DO with principal symbol $a$.
\end{remark}
\begin{remark}\label{R2.7}
 The definition of the operator $\mathfrak{Op}^A_\Phi(a)$ is gauge covariant, in the sense that for any real function $\psi\in C^\infty_{\text{\sf pol}}(\mathbb{R}^d)$, if $A^\prime=A+\d\psi$ then $\mathfrak{Op}^{A^\prime}_\Phi(a)=e^{i\psi}\mathfrak{Op}^A_\Phi(a)e^{-i\psi}$.
\end{remark}

\begin{lemma}\label{L2.8}
 Under the assumptions of Lemma \ref{L2.2}, for any $v\in\mathcal{S}(\mathbb{R}^d)$ and for any $y\in\mathbb{R}^d$ the following oscillating integral
\begin{equation}\label{2.6}
 \big(Sv\big)(y):=\int_{\mathbb{R}^d}\int_{\mathbb{R}^d}e^{-i\phi(x,y,\eta)}\omega^A(y,x)\overline{a(x,\eta)}v(x)\,dx\,\dbar\eta
\end{equation}
is well defined, we have that $Sv\in\mathcal{S}(\mathbb{R}^d)$ and the map
$$
\mathcal{S}(\mathbb{R}^d)\ni v\mapsto Sv\in\mathcal{S}(\mathbb{R}^d)
$$
is linear and continuous.

Moreover we have that $S=\left[\mathfrak{Op}^A_\Phi(a)\right]^*$ (the formal adjoint of the operator $\mathfrak{Op}^A_\Phi(a)$); thus $\mathfrak{Op}^A_\Phi(a)\in\mathbb{B}\big(\mathcal{S}^\prime(\mathbb{R}^d)\big)$.
\end{lemma}
\begin{proof}
 For the first part of the statement we proceed like in the proof of Lemma \ref{L2.2} changing one of the operators used in the integration by parts; we define
$$
L_\phi:=\left<\big(\nabla_x\phi\big)(x,y,\eta)\right>^{-2}\left(1\,+\,i\left<\big(\nabla_x\phi\big)(x,y,\eta),\nabla_x\right>\right).
$$
In fact this operator satisfies the identity $L_\phi e^{-i\phi}=e^{-i\phi}$. We continue by noticing that 
$$
\big(\nabla_x\phi\big)(x,y,\eta)=\eta+\big(\nabla_x d\big)(x,y,\eta)
$$
and using the second property in Hypothesis \ref{Hyp-U} and Lemma \ref{L1.5} we deduce that $\nabla_x\phi\in S^1(\mathbb{R}^d)$ and that there exists a constant $C>0$ such that
\begin{equation}\label{2.7}
 C^{-1}<\eta>\leq\left<\big(\nabla_x\phi\big)(x,y,\eta)\right>\leq C<\eta>,\quad\forall(x,y,\eta)\in\big[\mathbb{R}^d\big]^3.
\end{equation}

The formal adjoint of $L_\phi$ is given by
\begin{equation}\label{2.8}
 ^tL_\phi=\left<\big(\nabla_x\phi\big)(x,y,\eta)\right>^{-1}\left[\alpha\,+\,\left<\beta,\nabla_x\right>\right],\text{ with }(\alpha,\beta)\in\left[BC^\infty(\mathbb{R}^{2d})\right]^2.
\end{equation}
Through the usual integration by parts procedure one obtains the following repesentation of $Sv$:
$$
Sv(y)=\int_{\mathbb{R}^d}\int_{\mathbb{R}^d}e^{i<y-x,\eta>}<y-x>^{-2M}(1-\Delta_\eta)^{M}\left\{e^{-id(x,\eta)}\left(^tL_\phi\right)^N\left[\omega^A(y,x)\overline{a(x,\eta)}v(x)\right]\right\}\,dx\,\dbar\eta,
$$
where we choose $N\in\mathbb{N}$ and $M\in\mathbb{N}$ sufficiently large. The first properties in the statement follow in a straightforward way. Considering the last statement, a direct computation shows that:
$$
\left(\mathfrak{Op}^A_\Phi(a)u,v\right)_{L^2(\mathbb{R}^d)}=\left(u,Sv\right)_{L^2(\mathbb{R}^d)},\qquad\forall(u,v)\in\big[\mathcal{S}(\mathbb{R}^d)\big]^2.
$$
\end{proof}
\begin{lemma}\label{L2.9}
 An operator $W\in\mathbb{B}\big(\mathcal{S}(\mathbb{R}^d);\mathcal{S}^\prime(\mathbb{R}^d)\big)$ is of the form $\mathfrak{Op}^A_\Phi(a)$ with a symbol $a\in S^{-\infty}(\mathbb{R}^d)$ if and only if there exists a symbol $b\in S^{-\infty}(\mathbb{R}^d)$ such that $W=b^A(x,D)$.
\end{lemma}
\begin{proof}
 If $W=\mathfrak{Op}^A_\Phi(a)$ with $a\in S^{-\infty}(\mathbb{R}^d)$, then for any $u\in\mathcal{S}(\mathbb{R}^d)$ we have that
$$
\big(Wu\big)(x)=\int_{\mathbb{R}^d}\int_{\mathbb{R}^d}e^{i\phi(x,y,\eta)}\omega^A(x,y)a(x,\eta)u(y)\,dy\,\dbar\eta=
$$
$$
=\int_{\mathbb{R}^d}\int_{\mathbb{R}^d}e^{i<x-y,\eta>}\omega^A(x,y)b(x,\eta)u(y)\,dy\,\dbar\eta=\big(b^A(x,D)u\big)(x)
$$
with $b(x,\eta):=e^{id(x,\eta)}a(x,\eta)$ so that evidently $b\in S^{-\infty}(\mathbb{R}^d)$.  The converse statement follows by the same arguments.
\end{proof}

\section{Product of a $\Psi$DO with a FIO}

Suppose given $a\in S^{m''}(\mathbb{R}^d)$ and $b\in S^{m'}(\mathbb{R}^d)$ and the associated magnetic operators: $E:=a^A(x,D)$ and $F:=\mathfrak{Op}^A_{\Phi}(b)$. Then, due to our previous analysis, the composed operator $E\circ F$ exists as a linear bounded operator from $\mathcal{S}(\mathbb{R}^d)$ to $\mathcal{S}(\mathbb{R}^d)$. We are interested to obtain some more precise properties of this operator.

First let us make explicit the form of our two magnetic operators:
\begin{equation}\label{3.1}
 \big(Eu\big)(x)=\int_{\mathbb{R}^d}\int_{\mathbb{R}^d}e^{i<x-z,\zeta>}\omega^A(x,z)a(x,\zeta)u(z)\,dz\,\dbar\zeta,\quad\forall u\in\mathcal{S}(\mathbb{R}^d),\ \forall x\in\mathbb{R}^d,
\end{equation}
\begin{equation}\label{3.2}
 \big(Fv\big)(z)=\int_{\mathbb{R}^d}\int_{\mathbb{R}^d}e^{i(U(z,\eta)-<y,\eta>}\omega^A(z,y)b(z,\eta)v(y)\,dy\,\dbar\eta,\quad\forall v\in\mathcal{S}(\mathbb{R}^d),\ \forall z\in\mathbb{R}^d.
\end{equation}

\subsection{The product is a FIO}

In order to have an explicit computation we shall approximate the two symbols by compactly supported ones. More precisely we fix some positive function $\chi\in C^\infty_0(\mathbb{R}^{2d})$ with $\chi(0,0)=1$ and for any $\epsilon\in(0,1]$ we define
$$
a_\epsilon(x,z,\zeta):=\chi(\epsilon z,\epsilon\zeta)a(x,\zeta),\quad b_\epsilon(y,z,\eta):=\chi(\epsilon y,\epsilon\eta)b(z,\eta).
$$
We denote by $E_\epsilon$ and $F_\epsilon$ the operators given respectively by \eqref{3.1} and \eqref{3.2} with the symbols $a$ and $b$ replaced by their approximants $a_\epsilon$, respectively $b_\epsilon$.

Returning to the proof of Lemma \ref{L2.2} we easily notice that for any pair of elements $u$ and $v$ from $\mathcal{S}(\mathbb{R}^d)$ we have that
$$
\underset{\epsilon\searrow0}{\lim}E_\epsilon u=Eu,\quad\underset{\epsilon\searrow0}{\lim}F_\epsilon v=Fv,\quad\text{in }\mathcal{S}(\mathbb{R}^d).
$$
Due to the fact that $E_\epsilon$ and $F_\epsilon$ are bounded operators on $\mathcal{S}(\mathbb{R}^d)$, the Banach-Steinhaus Theorem implies that 
$$
\underset{\epsilon\searrow0}{\lim}(E_\epsilon\circ F_\epsilon)v=(E\circ F)v,\quad\forall v\in\mathcal{S}(\mathbb{R}^d).
$$

A direct computation gives for any $v\in\mathcal{S}(\mathbb{R}^d)$ and any $x\in\mathbb{R}^d$
$$
\left[(E_\epsilon\circ F_\epsilon)v\right](x)=
$$
$$
=\int_{\mathbb{R}^{4d}}e^{i(<x-z,\zeta>-<y,\eta>+U(z,\eta))}\omega^A(x,y)\Omega^B(x,z,y)a_\epsilon(x,z,\zeta)b_\epsilon(y,z,\eta)v(y)\,dy\,dz\,\dbar\eta\,\dbar\zeta.
$$
Thus the distribution kernel $K_\epsilon\in\mathcal{S}(\mathbb{R}^{2d})$ of the operator $E_\epsilon\circ F_\epsilon$ is given by
$$
K_\epsilon(x,y)=\omega^A(x,y)\int_{\mathbb{R}^{3d}}e^{i(<x-z,\zeta>-<y,\eta>+U(z,\eta))}\Omega^B(x,z,y)a_\epsilon(x,z,\zeta)b_\epsilon(y,z,\eta)\,dz\,\dbar\eta\,\dbar\zeta.
$$
Using the equality \eqref{2.5} we obtain that $E_\epsilon\circ F_\epsilon=\mathfrak{Op}^A_{\Phi}(c_\epsilon)$ with $c_\epsilon\in\mathcal{S}(\mathbb{R}^{2d})$; explicitely we have
\begin{equation}\label{3.4}
 c_\epsilon(x,\xi):=\int_{\mathbb{R}^{4d}}e^{i\psi(x,y,z,\xi,\eta,\zeta)}\Omega^B(x,z,y)a_\epsilon(x,z,\zeta)b_\epsilon(y,z,\eta)\,dy\,dz\,\dbar\eta\,\dbar\zeta,
\end{equation}
with
\begin{equation}\label{3.5}
 \psi(x,y,z,\xi,\eta,\zeta):=<y,\xi-\eta>+<x-z,\zeta>+U(z,\eta)-U(x,\xi).
\end{equation}

\begin{proposition}\label{P3.2}
Suppose we have $f\in S^{m,m',m''}(\mathbb{R}^{3d}\times\mathbb{R}^d\times\mathbb{R}^d\times\mathbb{R}^d)$ and a phase function $\psi$ defined by \eqref{3.5}, then the following oscillating integral is well defined
\begin{equation}\label{3.7}
 c(x,\xi):=\int_{\mathbb{R}^{4d}}e^{i\psi(x,y,z,\xi,\eta,\zeta)}\Omega^B(x,z,y)f(x,y,z,\xi,\eta,\zeta)\,dy\,dz\,\dbar\eta\dbar\,\zeta.
\end{equation}
Moreover $c\in S^{m+m'+m''}(\mathbb{R}^d)$ and the map
$$
S^{m,m',m''}(\mathbb{R}^{3d}\times\mathbb{R}^d\times\mathbb{R}^d\times\mathbb{R}^d)\ni f\mapsto c\in S^{m+m'+m''}(\mathbb{R}^d)
$$
is continuous.
\end{proposition}
\begin{corollary}\label{C3.2}
 Taking $f(x,y,z,\xi,\eta,\zeta)=a_\epsilon(x,z,\zeta)b_\epsilon(y,z,\eta)$ we conclude that $c_\epsilon$ defined in \eqref{3.4} belongs to $S^{m'+m''}(\mathbb{R}^d)$ and the limit $c:=\underset{\epsilon\searrow0}{\lim}c_\epsilon$ exists in $S^{m'+m''}(\mathbb{R}^d)$.
\end{corollary}

\begin{proposition}\label{P3.1}
 If $a\in S^{m''}(\mathbb{R}^d)$ and $b\in S^{m'}(\mathbb{R}^d)$, then $a^A(x,D)\circ\mathfrak{Op}^A_\Phi(b)=\mathfrak{Op}^A_\Phi(c)$ with $c\in S^{m'+m''}(\mathbb{R}^d)$. Moreover $c$ is the limit in $S^{m'+m''}(\mathbb{R}^d)$ of the symbols defined by \eqref{3.4}.
\end{proposition}
\begin{proof}
 Due to Corollary \ref{C3.2} we know that $c:=\underset{\epsilon\searrow0}{\lim}c_\epsilon$ exists in $S^{m'+m''}(\mathbb{R}^d)$. It follows that $\underset{\epsilon\searrow0}{\lim}\mathfrak{Op}^A_\Phi(c_\epsilon)v=\mathfrak{Op}^A_\Phi(c)v$ for any $v\in\mathcal{S}(\mathbb{R}^d)$. This finishes the proof of the Proposition.
\end{proof}

\begin{proofP32} 
 Let us fix some positive function $\chi\in C^\infty_0(\mathbb{R}^{4d})$ with $\chi(0,0,0,0)=1$ and for any $\epsilon\in(0,1]$ let us define $f_\epsilon(x,y,z,\xi,\eta,\zeta)=\chi(\epsilon y,\epsilon z,\epsilon\eta,\epsilon\zeta)f(x,y,z,\xi,\eta,\zeta)$. We denote by $c_\epsilon$ the integral that \eqref{3.7} associates to $f_\epsilon$ and notice that it defines an element of $\mathcal{S}(\mathbb{R}^{2d})$.

The arguments given below prove that for any $(x,\xi)\in\mathbb{R}^{2d}$, $\exists\underset{\epsilon\rightarrow0}{\lim}c_\epsilon(x,\xi)$.
Thus the conclusions of the Proposition, derive from the following fact: {\it if $f$ has compact support with respect to the variables $(y,z,\eta,\zeta)$, and $c$ is defined by the integral \eqref{3.7}, then any seminorm of $c$ defining the topology of  $S^{m+m'+m''}(\mathbb{R}^d)$ can be bounded by a seminorm of $f$ defining the topology of $S^{m,m',m''}(\mathbb{R}^{3d}\times\mathbb{R}^d\times\mathbb{R}^d\times\mathbb{R}^d)$ multiplied by some constant that is ``well controlled''}. In proving this fact we shall generically denote by $C_f$ any seminorm of $f$ defining the topology of $S^{m,m',m''}(\mathbb{R}^{3d}\times\mathbb{R}^d\times\mathbb{R}^d\times\mathbb{R}^d)$ and by $C$ any constant, that may depend only on $\Phi$ and $B$. Let us prove now this last statement.

We begin with a change of variables:
\begin{equation}\label{3.8}
 y:=x-\tilde{y},\ z:=x-\overline{z},\ \eta:=\xi-\overline{\eta},\ \zeta:=\xi-\overline{\eta}+\tilde{\zeta}
\end{equation}
leading to the formula
\begin{equation}\label{3.9}
 c(x,\xi):=\int_{\mathbb{R}^{4d}}e^{i\tilde{\psi}}\tilde{\Omega^B}\tilde{f}\,d\tilde{y}\,d\overline{z}\,\dbar\overline{\eta}\,\dbar\tilde{\zeta}
\end{equation}
with
\begin{equation}\label{3.10}
 \tilde{\psi}(x,\tilde{y},\overline{z},\xi,\overline{\eta},\tilde{\zeta}):=\psi(x,x-\tilde{y},x-\overline{z},\xi,\xi-\overline{\eta},\xi-\overline{\eta}+\tilde{\zeta})=<\overline{z},\tilde{\zeta}>-<\tilde{y},\overline{\eta}>+d(x-\overline{z},\xi-\overline{\eta})-d(x,\xi),
\end{equation}
\begin{equation}\label{3.11}
 \tilde{f}(x,\tilde{y},\overline{z},\xi,\overline{\eta},\tilde{\zeta}):=f(x,x-\tilde{y},x-\overline{z},\xi,\xi-\overline{\eta},\xi-\overline{\eta}+\tilde{\zeta}),
\end{equation}
\begin{equation}\label{3.12}
 \tilde{\Omega^B}(x,\tilde{y},\overline{z}):=\Omega^B(x,x-\overline{z},x-\tilde{y}).
\end{equation}

We fix now another positive function $\rho\in C^\infty_0(\mathbb{R}^d)$ with $\rho(\xi)=1$ for $|\xi|\leq1$ and $\rho(\xi)=0$ for $|\xi|\geq2$. We shall use the following notations:
\begin{equation}\label{3.13a}
 h_0(x,\xi):=\big(\nabla_x d\big)(x,\xi),\qquad h_1(x,\xi):=\big(\nabla_\xi d\big)(x,\xi),
\end{equation}
\begin{equation}\label{3.13b}
\overline{h}_0\equiv \overline{h}_0(x,\overline{z},\xi,\overline{\eta}):=h_0(x-\overline{z},\xi-\overline{\eta}),\qquad\overline{h}_1\equiv \overline{h}_1(x,\overline{z},\xi,\overline{\eta}):=h_1(x-\overline{z},\xi-\overline{\eta}),
\end{equation}
\begin{equation}\label{3.14}
\rho_\infty(x,\tilde{y},\overline{z},\xi,\overline{\eta},\tilde{\zeta}):=\rho\left(\frac{\epsilon(\xi-\overline{\eta})}{<\overline{\eta}>+<\tilde{\zeta}-\overline{h}_0>}\right),\qquad\rho_0:=1-\rho_\infty,
\end{equation}
for some $\epsilon\in(0,\epsilon_0]$ with $\epsilon_0$ sufficiently small, to be fixed later. We notice that the derivatives of $\rho_\infty$ satisfy estimations of the following type (on $\mathbb{R}^{6d}$):
\begin{equation}\label{3.16}
 \left|\left(\partial_{(x,\overline{z})}^\alpha\partial_{(\tilde{y},\xi,\overline{\eta},\tilde{\zeta})}^\beta\rho_\infty\right)(x,\tilde{y},\overline{z},\xi,\overline{\eta},\tilde{\zeta})\right|\,\leq\,C\left[1+\left(\frac{<\xi-\overline{\eta}>}{<\tilde{\zeta}-\overline{h}_0>}\right)^{|\alpha|}\right],
\end{equation}
for any mutiindices $\alpha\in\mathbb{N}^{2d}$ and $\beta\in\mathbb{N}^{4d}$.
We make the decomposition $\tilde{f}=\rho_0\tilde{f}+\rho_\infty\tilde{f}$ in the integral \eqref{3.9} and write $c=c_0+c_\infty$.

\paragraph{\sf Point 1.} We prove that for any $p\in\mathbb{N}$ we have the estimation:
\begin{equation}\label{3.17}
 \left|c_\infty(x,\xi)\right|\,\leq\,C_f<\xi>^{-p},\quad\forall(x,\xi)\in\mathbb{T}^*\mathbb{R}^d.
\end{equation}

In order to prove this estimation we have to make a number of integration by parts  using the following operators:
\begin{equation}\label{3.18a}
 L_{\tilde{y}}:=<\overline{\eta}>^{-2}\left(1+i\left<\overline{\eta},\nabla_{\tilde{y}}\right>\right),\qquad L_{\overline{z}}:=<\tilde{\zeta}-\overline{h}_0>^{-2}\left(1-i\left<\tilde{\zeta}-\overline{h}_0,\nabla_{\overline{z}}\right>\right),
\end{equation}
\begin{equation}\label{3.18b}
 L_{\overline{\eta}}:=<\tilde{y}+\overline{h}_1>^{-2}\left(1+i\left<\tilde{y}+\overline{h}_1,\nabla_{\overline{\eta}}\right>\right),\qquad L_{\tilde{\zeta}}:=<\overline{z}>^{-2}\left(1-i\left<\overline{z},\nabla_{\tilde{\zeta}}\right>\right),
\end{equation}
that satisfy the identities:
\begin{equation}\label{3.19}
 L_{\tilde{y}}e^{i\tilde{\psi}}= L_{\overline{z}}e^{i\tilde{\psi}}= L_{\overline{\eta}} e^{i\tilde{\psi}}= L_{\tilde{\zeta}} e^{i\tilde{\psi}}=e^{i\tilde{\psi}}.
\end{equation}
For the explicit computations we shall also need the form of their formal adjoints:
\begin{equation}\label{3.20a}
 ^tL_{\tilde{y}}:=<\overline{\eta}>^{-2}\left(1-i\left<\overline{\eta},\nabla_{\tilde{y}}\right>\right),\qquad ^tL_{\overline{z}}:=<\tilde{\zeta}-\overline{h}_0>^{-1}\left(a_1+i\left<b_1,\nabla_{\overline{z}}\right>\right),
\end{equation}
\begin{equation}\label{3.20b}
^t L_{\overline{\eta}}:=<\tilde{y}+\overline{h}_1>^{-1}\left(a_2+i\left<b_2,\nabla_{\overline{\eta}}\right>\right),\qquad ^tL_{\tilde{\zeta}}:=<\overline{z}>^{-2}\left(1+i\left<\overline{z},\nabla_{\tilde{\zeta}}\right>\right),
\end{equation}
where 
\begin{equation}\label{3.21}
 \left\{
\begin{array}{lcl}
 a_2\in BC^\infty(\mathbb{R}^{6d}),&&b_2\in BC^\infty(\mathbb{R}^{6d}),\\
\left|\left(\partial_{(x,\overline{z})}^\alpha\partial_{(\tilde{y},\xi,\overline{\eta},\tilde{\zeta})}^\beta b_1\right)(x,\tilde{y},\overline{z},\xi,\overline{\eta},\tilde{\zeta})\right|&\leq&C\left[1+\left(\frac{<\xi-\overline{\eta}>}{<\tilde{\zeta}-\overline{h}_0>}\right)^{|\alpha|}\right]\quad\text{on }\mathbb{R}^{6d},\\
\left|\left(\partial_{(x,\overline{z})}^\alpha\partial_{(\tilde{y},\xi,\overline{\eta},\tilde{\zeta})}^\beta a_1\right)(x,\tilde{y},\overline{z},\xi,\overline{\eta},\tilde{\zeta})\right|&\leq&C\left[1+\left(\frac{<\xi-\overline{\eta}>}{<\tilde{\zeta}-\overline{h}_0>}\right)^{1+|\alpha|}\right]\quad\text{on }\mathbb{R}^{6d},
\end{array}
\right.
\end{equation}
for any multiindices $\alpha\in\mathbb{N}^{2d}$ and $\beta\in\mathbb{N}^{4d}$.

For any point $(x,\xi)\in\mathbb{T}^*\mathbb{R}^d$ we define the following domains in $\mathbb{R}^{4d}$ (depending on the parameter $\epsilon$):
\begin{equation}\label{3.22a}
 \mathcal{D}_\infty(x,\xi):=\left\{(\tilde{y},\overline{z},\overline{\eta},\tilde{\zeta})\in\mathbb{R}^{4d}\mid|\xi-\overline{\eta}|\leq(2/\epsilon)\big(<\overline{\eta}>+<\tilde{\zeta}-\overline{h}_0>\big)\right\},
\end{equation}
\begin{equation}\label{3.22b}
 \mathcal{D}_0(x,\xi):=\left\{(\tilde{y},\overline{z},\overline{\eta},\tilde{\zeta})\in\mathbb{R}^{4d}\mid\  <\overline{\eta}>+<\tilde{\zeta}-\overline{h}_0>\leq\epsilon|\xi-\overline{\eta}|\right\},
\end{equation}
and notice that $\text{supp}\rho_\infty(x,\cdot,\cdot,\xi,\cdot,\cdot)\subset \mathcal{D}_\infty(x,\xi)$ and $\text{supp}\rho_0(x,\cdot,\cdot,\xi,\cdot,\cdot)\subset \mathcal{D}_0(x,\xi)$.

We can write now
\begin{equation}\label{3.23}
 c_\infty(x,\xi)=\int_{\mathcal{D}_\infty(x,\xi)}e^{i\tilde{\psi}}\big(^tL_{\tilde{y}}\big)^{N_1}\big(^tL_{\overline{z}}\big)^{N_2}\big(^tL_{\overline{\eta}}\big)^{N_3}\big(^tL_{\tilde{\zeta}}\big)^{N_4}\left(\tilde{\Omega^B}\rho_\infty\tilde{f}\right)\,d\tilde{y}\,d\overline{z}\,\dbar\overline{\eta}\,\dbar\tilde{\zeta}
\end{equation}
where the natural numbers $N_j$ (with $1\leq j\leq4$) will be chosen later. First let us notice that the derivatives with respect to $\overline{\eta}$ and $\tilde{\zeta}$ commute with the function $\tilde{\Omega^B}$. Meanwhile the derivatives with respect to $\tilde{y}$ and $\overline{z}$ do not commute with $\tilde{\Omega^B}$ but produce factors that are monomials of the form $\tilde{y}^\alpha\overline{z}^\beta$ (with $(\alpha,\beta)\in\mathbb{N}^{2d}$) multiplied with some function of class $BC^\infty(\mathbb{R}^{3d})$ and with $\tilde{\Omega^B}$ itself. The monomials $\tilde{y}^\alpha\overline{z}^\beta$ are dealt with by some suplementary integrations by parts using the identities:
\begin{equation}\label{3.24}
 \tilde{y}^\alpha e^{-i<\tilde{y},\overline{\eta}>}=\big(-D_{\overline{\eta}}\big)^\alpha e^{-i<\tilde{y},\overline{\eta}>},\qquad\overline{z}^\beta e^{i<\overline{z},\tilde{\zeta}>}=D_{\tilde{\zeta}}^\beta e^{i<\overline{z},\tilde{\zeta}>}
\end{equation}
and the fact that the derivatives of $d(x-\overline{z},\xi-\overline{\eta})$ with respect to $\overline{\eta}$ are bounded. Using \eqref{3.23}, \eqref{3.16}, \eqref{3.20a} and \eqref{3.20b} one obtains that
$$
\left|c_\infty(x,\xi)\right|\,\leq\,C_f\int_{\mathcal{D}_\infty(x,\xi)}<\overline{\eta}>^{-N_1}<\tilde{\zeta}-\overline{h}_0>^{-N_2}<\tilde{y}>^{-N_3}<\overline{z}>^{-N_4}\times
$$
$$
\times\left[1+\left(\frac{<\xi-\overline{\eta}>}{<\tilde{\zeta}-\overline{h}_0>}\right)^{N_2}\right]<\xi>^m<\xi-\overline{\eta}>^{m'}<\xi-\overline{\eta}+\tilde{\zeta}>^{m''}\,d\tilde{y}\,d\overline{z}\,\dbar\overline{\eta}\,\dbar\tilde{\zeta}.
$$
Further we use the inequalities:
$$
<\xi-\overline{\eta}>\leq C<\overline{\eta}><\tilde{\zeta}-\overline{h}_0>,\text{ on }\mathcal{D}_\infty(x,\xi),
$$
$$
<\xi-\overline{\eta}+\tilde{\zeta}>^{m''}\leq C<\tilde{\zeta}-\overline{h}_0>^{m''}<\xi-\overline{\eta}>^{|m''|},
$$
$$
<\overline{\eta}>^{-p}\leq C<\xi>^{-p}<\xi-\overline{\eta}>^p,\text{ for }p\in\mathbb{N},
$$
obtaining that:
$$
\left|c_\infty(x,\xi)\right|\,\leq\,C_f<\xi>^{m-p}\int_{\mathbb{R}^{3d}}<\overline{\eta}>^{-N_1+N_2+2p+m'+|m''|}<\tilde{y}>^{-N_3}<\overline{z}>^{-N_4}\times
$$
$$
\times\left[\int_{\mathbb{R}^d}<\tilde{\zeta}-\overline{h}_0>^{-N_2+p+m'+m''+|m''|}\,\dbar\tilde{\zeta}\right]\,d\tilde{y}\,d\overline{z}\,\dbar\overline{\eta}.
$$
The inequality \eqref{3.17} follows by choosing $N_3\geq d+1$, $N_4\geq d+1$, $N_2\geq p+m'+m''+|m''|+d+1$ and $N_1\geq N_2+2p+m'+|m''|+d+1$.

\paragraph{\sf Point 2.} We shall now prove the estimation:
\begin{equation}\label{3.25}
 \left|c_0(x,\xi)\right|\,\leq\,C_f<\xi>^{m+m'+m''},\quad\forall(x,\xi)\in\mathbb{T}^*\mathbb{R}^d.
\end{equation}

We begin by the change of variables $\tilde{\zeta}=\overline{\zeta}+\overline{h}$ in the integral defining $c_0$, where:
\begin{equation}\label{3.27}
 \overline{h}\equiv \overline{h}(x,\overline{z},\xi,\overline{\eta}):=h(x,\overline{z},\xi-\overline{\eta}),\qquad h(x,\overline{z},\theta):=\int_0^1\big(\nabla_x d\big)(x-t\overline{z},\theta)\,dt.
\end{equation}
We also have to use the identity:
\begin{equation}\label{3.28}
 d(x-\overline{z},\xi-\overline{\eta})-d(x,\xi-\overline{\eta})=-<\overline{h}(x,\overline{z},\xi,\overline{\eta}),\overline{z}>,
\end{equation}
and with the notations:
\begin{equation}\label{3.29}
 \underline{\rho}_0(x,\tilde{y},\overline{z},\xi,\overline{\eta},\overline{\zeta}):=\rho_0(x,\tilde{y},\overline{z},\xi,\overline{\eta},\overline{\zeta}+\overline{h}),\quad \underline{f}(x,\tilde{y},\overline{z},\xi,\overline{\eta},\overline{\zeta}):=\tilde{f}(x,\tilde{y},\overline{z},\xi,\overline{\eta},\overline{\zeta}+\overline{h}),
\end{equation}
we obtain that:
\begin{equation}\label{3.30}
 c_0(x,\xi)=\int_{\mathbb{R}^{4d}}e^{i\big(<\overline{z},\overline{\zeta}>-<\tilde{y},\overline{\eta}>\big)}e^{i\big(d(x,\xi-\overline{\eta})-d(x,\xi)\big)}\tilde{\Omega^B}\underline{\rho}_0\underline{f}\,d\tilde{y}\,d\overline{z}\,\dbar\overline{\eta}\,\dbar\overline{\zeta}.
\end{equation}
In \eqref{3.30} we integrate by parts using the operators:
\begin{equation}\label{3.31a}
 \overline{L}_{\tilde{y}}:=<\overline{\eta}>^{-2}\left(1-\Delta_{\tilde{y}}\right),\qquad \overline{L}_{\overline{z}}:=<\overline{\zeta}>^{-2}\left(1-\Delta_{\overline{z}}\right),
\end{equation}
\begin{equation}\label{3.31b}
 \overline{L}_{\overline{\eta}}:=<\tilde{y}>^{-2}\left(1-\Delta_{\overline{\eta}}\right),\qquad \overline{L}_{\overline{\zeta}}:=<\overline{z}>^{-2}\left(1-\Delta_{\overline{\zeta}}\right),
\end{equation}
that satisfy the identities:
\begin{equation}\label{3.32}
 \overline{L}_{\tilde{y}}e^{i\overline{\psi}}= \overline{L}_{\overline{z}}e^{i\overline{\psi}}= \overline{L}_{\overline{\eta}} e^{i\overline{\psi}}= \overline{L}_{\overline{\zeta}} e^{i\overline{\psi}}=e^{i\overline{\psi}},\quad\text{with}\ \overline{\psi}(\tilde{y},\overline{z},\overline{\eta},\overline{\zeta}):=<\overline{z},\overline{\zeta}>-<\tilde{y},\overline{\eta}>.
\end{equation}
We obtain that
\begin{equation}\label{3.33}
 c_0(x,\xi)=\int_{\mathbb{R}^{4d}}e^{i\overline{\psi}}\big(\overline{L}_{\tilde{y}}\big)^{\underline{N}_1}\big(\overline{L}_{\overline{z}}\big)^{\underline{N}_2}\big(\overline{L}_{\overline{\eta}}\big)^{\underline{N}_3}\big(\overline{L}_{\overline{\zeta}}\big)^{\underline{N}_4}\left(e^{i\big(d(x,\xi-\overline{\eta})-d(x,\xi)\big)}\tilde{\Omega^B}\rho_0\underline{f}\right)\,d\tilde{y}\,d\overline{z}\,\dbar\overline{\eta}\,\dbar\overline{\zeta},
\end{equation}
where the positive integers $\underline{N}_1$, $\underline{N}_2$, $\underline{N}_3$ and $\underline{N}_4$ will be fixed later.

In order to continue the estimation of our integral, using the derivation of products formula we shall write it as a sum $c_0:=c'_0+c''_0$ with the first integral containing all the terms with $\underline{\rho}_0$ and no derivative of $\underline{\rho}_0$ and the second integral containing only the terms with derivatives of $\underline{\rho}_0$. In order to estimate $c'_0$ we notice that
\begin{equation}\label{3.35}
 \left|\partial^\alpha_{(x,\overline{z})}\partial^\beta_{(\tilde{y},\xi,\overline{\eta},\overline{\zeta})}\underline{f}\right|\leq C_f<\xi>^m<\xi-\overline{\eta}>^{m'}<\overline{\zeta}+\xi-\overline{\eta}+\overline{h}>^{m''}\left[1+\left(\frac{<\xi-\overline{\eta}>}{<\overline{\zeta}+\xi-\overline{\eta}+\overline{h}>}\right)^{|\alpha|}\right]
\end{equation}
on $\mathbb{R}^{6d}$, for any $\alpha\in\mathbb{N}^{2d}$ and any $\beta\in\mathbb{N}^{4d}$. We also notice that in fact the actual integration domain is given by:
$$
\tilde{\mathcal{D}}_0(x,\xi):=\left\{(\tilde{y},\overline{z},\overline{\eta},\overline{\zeta})\in\mathbb{R}^{4d}\mid\  <\overline{\eta}>+<\overline{\zeta}+\overline{h}-\overline{h}_0>\leq\epsilon|\xi-\overline{\eta}|\right\}.
$$
We take into account the fact that the derivatives of $d(x,\xi-\overline{\eta})$ with respect to $\overline{\eta}$ are bounded and that the monomials of the form $\tilde{y}^\alpha\overline{z}^\beta$ appearing when derivating $\tilde{\Omega^B}$ can be eliminated by the usual integration by parts (as in the case of $c_\infty$) in order to obtain the estimation:
$$
\left|c'_0(x,\xi)\right|\leq C_f\int_{\tilde{\mathcal{D}}_0(x,\xi)}<\tilde{y}>^{-2\overline{N}_3}<\overline{z}>^{-2\overline{N}_4}<\overline{\eta}>^{-2\overline{N}_1}<\overline{\zeta}>^{-2\overline{N}_2}<\xi>^m<\xi-\overline{\eta}>^{m'}<\overline{\zeta}+\xi-\overline{\eta}+\overline{h}>^{m''}\times
$$
$$
\times\left[1+\left(\frac{<\xi-\overline{\eta}>}{<\overline{\zeta}+\xi-\overline{\eta}+\overline{h}>}\right)^{2\overline{N}_2}\right]\,d\tilde{y}\,d\overline{z}\,\dbar\overline{\eta}\,\dbar\overline{\zeta}.
$$
Moreover, on $\tilde{\mathcal{D}}_0(x,\xi)$ we have the inequality:
\begin{equation}\label{3.36}
 <\overline{\eta}>+<\overline{\zeta}+\overline{h}-\overline{h}_0>\,\leq\,\epsilon|\xi-\overline{\eta}|.
\end{equation}
In particular, $|\overline{\eta}|\leq\epsilon|\xi-\overline{\eta}|$, so that for $\epsilon_0$ sufficiently small we have that
\begin{equation}\label{3.37}
 C^{-1}<\xi>\,\leq\,<\xi-\overline{\eta}>\,\leq\,C<\xi>.
\end{equation}
From \eqref{3.36} we also conclude that $|\overline{\zeta}+\overline{h}-\overline{h}_0|\leq\epsilon|\xi-\overline{\eta}|$ and using the equality $\overline{\zeta}+\xi-\overline{\eta}+\overline{h}=\overline{\zeta}+\overline{h}-\overline{h}_0+\xi-\overline{\eta}+\overline{h}_0$ and Lemma \ref{1.6} we conclude that (reducing may be the value of $\epsilon_0$):
\begin{equation}\label{3.38}
 C^{-1}<\xi-\overline{\eta}>\,\leq\,<\overline{\zeta}+\xi-\overline{\eta}+\overline{h}>\,\leq\,C<\xi-\overline{\eta}>.
\end{equation}
Choosing now $2\overline{N}_j\geq d+1$ and using \eqref{3.37} and \eqref{3.38} we conclude that
\begin{equation}\label{3.39}
 \left|c'_0(x,\xi)\right|\,\leq\,C_f<\xi>^{m+m'+m''},\qquad\forall(x,\xi)\in\mathbb{T}^*\mathbb{R}^d.
\end{equation}
In order to estimate $c''_0$ we come back to the variables $\tilde{\zeta}=\overline{\zeta}+\overline{h}$. Now the integration domain is given by $\mathcal{D}_0(x,\xi)\bigcap\mathcal{D}_\infty(x,\xi)$ due to the support of the derivatives of $\underline{\rho}_0$. Thus the arguments leading to the estimation \eqref{3.17} for $c_\infty$ may be repeated to obtain for any $p\in\mathbb{N}$,
\begin{equation}\label{3.40}
 \left|c''_0(x,\xi)\right|\,\leq\,C_f<\xi>^{-p},\qquad\forall(x,\xi)\in\mathbb{T}^*\mathbb{R}^d.
\end{equation}
Putting together \eqref{3.39} and \eqref{3.40} we obtain \eqref{3.25}.
\paragraph{\sf Conclusion.} 
\begin{equation}\label{3.41}
 \left|c(x,\xi)\right|\,\leq\,C_f<\xi>^{m+m'+m''},\qquad\forall(x,\xi)\in\mathbb{T}^*\mathbb{R}^d.
\end{equation}

The derivatives of $c$ are estimated by similar arguments; in order to illustrate the minor differences that may appear let us study $\nabla_xc$ and $\nabla_\xi c$.

\paragraph{\sf Point 3.} We prove that
\begin{equation}\label{3.42}
 \left|\big(\nabla_xc\big)(x,\xi)\right|\,\leq\,C_f<\xi>^{m+m'+m''},\qquad\forall(x,\xi)\in\mathbb{T}^*\mathbb{R}^d.
\end{equation}
We begin by noticing that differentiating the function $c$ given by \eqref{3.9} with respect to $x$ one gets
\begin{equation}\label{3.43}
 \nabla_xc\,=\,T_1+T_2+T_3+T_4
\end{equation}
where:
\begin{itemize}
 \item $T_1$ has the form \eqref{3.9} with $\nabla_x\tilde{f}$ replacing $\tilde{f}$; it satisfies thus an estimation of the form \eqref{3.41} and we conclude that
\begin{equation}\label{3.44}
 \left|T_1(x,\xi)\right|\,\leq\,C_f<\xi>^{m+m'+m''},\qquad\forall(x,\xi)\in\mathbb{T}^*\mathbb{R}^d.
\end{equation}
\item $T_2$ is also of the form \eqref{3.9} but with $\tilde{f}$ replaced by $-i\big(\nabla_x\tilde{F}\big)\tilde{f}$ with (using the notations of Lemma \ref{L1.11})
$$
\tilde{F}(x,\tilde{y},\overline{z}):=F(x,x-\overline{z},x-\tilde{y}).
$$
\item $T_3$ is also of the form \eqref{3.9} but with $\tilde{f}$ replaced by
\begin{equation}\label{3.48}
 i\left[\big(\nabla_xd\big)(x-\overline{z},\xi-\overline{\eta})-\big(\nabla_xd\big)(x,\xi-\overline{\eta})\right]\tilde{f}\,=\,-i\left[\int_0^1\big(\nabla^2_{x,x}d\big)(x-t\overline{z},\xi-\overline{\eta})dt\right]\cdot\overline{z}\tilde{f}.
\end{equation}
\item $T_4$ is also of the form \eqref{3.9} but with $\tilde{f}$ replaced by
\begin{equation}\label{3.50}
 i\left[\big(\nabla_xd\big)(x,\xi-\overline{\eta})-\big(\nabla_xd\big)(x,\xi)\right]\tilde{f}\,=\,-i\left[\int_0^1\big(\nabla^2_{x,\xi}d\big)(x,\xi-t\overline{\eta})dt\right]\cdot\overline{\eta}\tilde{f}.
\end{equation}
\end{itemize}
Let us deal with $T_2$. Using point (4) in Lemma \ref{L1.11} we get
\begin{equation}\label{3.45}
 \nabla_x\tilde{F}=D\tilde{y}+E\overline{z},\quad\text{ with $D$ and $E$ of class }BC^\infty(\mathbb{R}^{3d}).
\end{equation}
As previously we eliminate $\tilde{y}$ and $\overline{z}$ by integration by parts using the identities
\begin{equation}\label{3.46}
 \tilde{y}e^{-i<\tilde{y},\overline{\eta}>}=i\nabla_{\overline{\eta}}e^{-i<\tilde{y},\overline{\eta}>},\qquad\overline{z}e^{i<\overline{z},\tilde{\zeta}>}=-i\nabla_{\tilde{\zeta}}e^{i<\overline{z},\tilde{\zeta}>}
\end{equation}
and the fact that $\nabla_\xi d$ is a symbol of order 0. We conclude that $T_2$ may be written as a sum of integrals of the form \eqref{3.9} (with $m'$ or $m''$ maybe reduced by an unity) so that due to \eqref{3.41} we have that
\begin{equation}\label{3.47}
 \left|T_2(x,\xi)\right|\,\leq\,C_f<\xi>^{m+m'+m''},\qquad\forall(x,\xi)\in\mathbb{T}^*\mathbb{R}^d.
\end{equation}

For $T_3$ we proceed like in the case of $T_2$ using the second of the identities \eqref{3.46}; in this way the factor $\overline{z}\tilde{f}$ in \eqref{3.48} will be replaced by $i\widetilde{\big(\nabla_\zeta f\big)}$, and we notice that 
$$
\nabla_\zeta f\in S^{m,m',m''-1}\big(\mathbb{R}^{3d}\times\mathbb{R}^d\times\mathbb{R}^d\times\mathbb{R}^d\big).
$$
On the other hand, the entries of the matrix in the last paranthesis in \eqref{3.48} are of the form $\tilde{e}$ with $e\in S^{0,1,0}\big(\mathbb{R}^{3d}\times\mathbb{R}^d\times\mathbb{R}^d\times\mathbb{R}^d\big)$; thus $e\nabla_\zeta f\in S^{m,m'+1,m''-1}\big(\mathbb{R}^{3d}\times\mathbb{R}^d\times\mathbb{R}^d\times\mathbb{R}^d\big)$ and \eqref{3.41} implies that
\begin{equation}\label{3.49}
 \left|T_3(x,\xi)\right|\,\leq\,C_f<\xi>^{m+m'+m''},\qquad\forall(x,\xi)\in\mathbb{T}^*\mathbb{R}^d.
\end{equation}

In the integral defining $T_4$ we eliminate $\overline{\eta}$ by integration by parts using the identity
\begin{equation}\label{3.51}
 \overline{\eta}e^{-i<\tilde{y},\overline{\eta}>}=i\nabla_{\tilde{y}}e^{-i<\tilde{y},\overline{\eta}>},
\end{equation}
and producing factors of the form $\nabla_{\tilde{y}}\tilde{f}$ that have the right behaviour and also factors of the type $\tilde{y}\tilde{f}$ and $\overline{z}\tilde{f}$ which are dealt with like in the case of $T_2$. Due to the fact that $\nabla^2_{x,\xi}d$ are of class $BC^\infty(\mathbb{R}^{2d})$, the proof of \eqref{3.41} allow to see that
\begin{equation}\label{3.52}
 \left|T_4(x,\xi)\right|\,\leq\,C_f<\xi>^{m+m'+m''},\qquad\forall(x,\xi)\in\mathbb{T}^*\mathbb{R}^d.
\end{equation}
Thus \eqref{3.42} follows now from \eqref{3.43}, \eqref{3.44}, \eqref{3.47}, \eqref{3.49} and \eqref{3.52}.
\paragraph{\sf Point 4.} We prove that
\begin{equation}\label{3.53}
 \left|\big(\nabla_\xi c\big)(x,\xi)\right|\,\leq\,C_f<\xi>^{m+m'+m''-1},\qquad\forall(x,\xi)\in\mathbb{T}^*\mathbb{R}^d.
\end{equation}
We begin again by noticing that differentiating the function $c$ given by \eqref{3.9} with respect to $\xi$ one gets
\begin{equation}\label{3.54}
 \nabla_\xi c\,=\,S_1+S_2+S_3
\end{equation}
where:
\begin{itemize}
 \item $S_1$ has the form \eqref{3.9} with $\nabla_\xi \tilde{f}$ replacing $\tilde{f}$; but $\nabla_\xi \tilde{f}=\tilde{f}_1+\tilde{f}_2+\tilde{f}_3$ with $\tilde{f}_1$ in $S^{m-1,m',m''}\big(\mathbb{R}^{3d}\times\mathbb{R}^d\times\mathbb{R}^d\times\mathbb{R}^d\big)$,$\tilde{f}_2$ in $S^{m,m'-1,m''}\big(\mathbb{R}^{3d}\times\mathbb{R}^d\times\mathbb{R}^d\times\mathbb{R}^d\big)$ and $\tilde{f}_3$ in $S^{m,m',m''-1}\big(\mathbb{R}^{3d}\times\mathbb{R}^d\times\mathbb{R}^d\times\mathbb{R}^d\big)$, so that \eqref{3.41} implies that
\begin{equation}\label{3.55}
 \left|S_1(x,\xi)\right|\,\leq\,C_f<\xi>^{m+m'+m''-1},\qquad\forall(x,\xi)\in\mathbb{T}^*\mathbb{R}^d.
\end{equation}
\item $S_2$ is also of the form \eqref{3.9} but with $\tilde{f}$ replaced by
\begin{equation}\label{3.56}
 i\left[\big(\nabla_\xi d\big)(x-\overline{z},\xi-\overline{\eta})-\big(\nabla_\xi d\big)(x,\xi-\overline{\eta})\right]\tilde{f}\,=\,-i\left[\int_0^1\big(\nabla^2_{\xi,x}d\big)(x-t\overline{z},\xi-\overline{\eta})dt\right]\cdot\overline{z}\tilde{f}.
\end{equation}
\item $S_3$ is also of the form \eqref{3.9} but with $\tilde{f}$ replaced by
\begin{equation}\label{3.58}
 i\left[\big(\nabla_\xi d\big)(x,\xi-\overline{\eta})-\big(\nabla_\xi d\big)(x,\xi)\right]\tilde{f}\,=\,-i\left[\int_0^1\big(\nabla^2_{\xi,\xi}d\big)(x,\xi-t\overline{\eta})dt\right]\cdot\overline{\eta}\tilde{f}.
\end{equation}
\end{itemize}

In the integral defining $S_2$ we eliminate $\overline{z}$ as in the case of $T_3$ producing a factor of the form $\widetilde{\big(\nabla_\zeta f\big)}$, while the entries of the matrix in the last paranthesis in \eqref{3.56} are of the form $\tilde{e}$ with $e\in S^{0,0,0}\big(\mathbb{R}^{3d}\times\mathbb{R}^d\times\mathbb{R}^d\times\mathbb{R}^d\big)$. Thus
\begin{equation}\label{3.57}
 \left|S_2(x,\xi)\right|\,\leq\,C_f<\xi>^{m+m'+m''-1},\qquad\forall(x,\xi)\in\mathbb{T}^*\mathbb{R}^d.
\end{equation}

In the integral defining $S_3$ we eliminate $\overline{\eta}$ as in the case of $T_4$ and we continue as in the proof of \eqref{3.41}. Thus we use once again the partition of unity $1=\rho_0+\rho_\infty$ and write the associated decomposition of the integral defining $S_3$ as $S_3=S_{3,0}+S_{3,\infty}$. The entries $e_{jk}(x,\xi,\overline{\eta})$ of the matrix in the last paranthesis in \eqref{3.58} being of class $BC^\infty(\mathbb{R}^{3d})$ the same arguments that lead us to the inequality \eqref{3.17} allow us to prove that for any $p\in\mathbb{N}$
\begin{equation}\label{3.59}
 \left|S_{3,\infty}(x,\xi)\right|\,\leq\,C_f<\xi>^{-p},\qquad\forall(x,\xi)\in\mathbb{T}^*\mathbb{R}^d.
\end{equation}
In order to deal with $S_{3,0}$ we follow the proof of the inequality \eqref{3.25}. The essential remark is that for any $\alpha\in\mathbb{N}^d$ we have the inequality
\begin{equation}\label{3.60}
 \left|\big(\partial_\eta^\alpha e_{jk}\big)(x,\xi,\overline{\eta})\right|\,\leq\,C\int_0^1<\xi-t\overline{\eta}>^{-1-|\alpha|}dt.
\end{equation}
Recalling the proof of \eqref{3.37} we see that for $\epsilon_0$ small enough and $\epsilon\in(0,\epsilon]$ we have for any $\alpha\in\mathbb{N}^d$
\begin{equation}\label{3.61}
 \left|\big(\partial_\eta^\alpha e_{jk}\big)(x,\xi,\overline{\eta})\right|\,\leq\,C<\xi>^{-1-|\alpha|},\quad\text{ on }\overline{D_0}(x,\xi).
\end{equation}
In this way we have an extra power of $<\xi>^{-1}$ with respect to the estimations that lead to \eqref{3.25} and finally we get
\begin{equation}\label{3.62}
 \left|S_{3,0}(x,\xi)\right|\,\leq\,C_f<\xi>^{m+m'+m''-1},\qquad\forall(x,\xi)\in\mathbb{T}^*\mathbb{R}^d.
\end{equation}
 
The inequality \eqref{3.53} follows from \eqref{3.54}, \eqref{3.55}, \eqref{3.57}, \eqref{3.59} and \eqref{3.62}.

For the higher derivatives of $c$ we proceed similarly and obtain that for any $\alpha$ and $\beta$ in $\mathbb{N}^d$ the following estimations are true:
$$
\left|\big(\partial_x^\alpha\partial_\xi^\beta c\big)(x,\xi)\right|\,\leq\,C_f<\xi>^{m+m'+m''-|\beta|},\qquad\forall(x,\xi)\in\mathbb{T}^*\mathbb{R}^d.
$$
\end{proofP32}

\subsection{The symbol of the product}

Our aim here is to compute the symbol $c$ of the magnetic FIO $a^A(x,D)\circ\mathfrak{Op}^A_\Phi(b)$ defined under the assumptions of Proposition \ref{P3.1}, modulo $S^{m'+m''-2}(\mathbb{R}^d)$. In fact we shall compute the symbol $c$ defined in \eqref{3.7} under the assumptions of Proposition \ref{P3.2}, modulo $S^{m+m'+m''-2}(\mathbb{R}^d)$. In the rest of this subsection we shall use the notations from the proof of Proposition \ref{P3.2}.

We begin with the case $B=0$ and thus $\Omega^B=\tilde{\Omega^B}=1$ and we denote by
\begin{equation}\label{3.63}
e(x,\xi)\,:=\,\int_{\mathbb{R}^{4d}}e^{i\tilde{\Psi}}\tilde{f}\,d\tilde{y}\,d\overline{z}\,\dbar\overline{\eta}\,\dbar\tilde{\zeta}. 
\end{equation}
\begin{proposition}\label{P3.3}
 Under the assumptions of Proposition \ref{P3.2} we have the congruence
\begin{equation}\label{3.64}
 e\,\equiv\,c_0,\ \text{modulo}\ S^{m+m'+m''-2}(\mathbb{R}^d),
\end{equation}
where
\begin{equation}\label{3.65}
 c_0(x,\xi):=f\left(x,\big(\nabla_\xi U\big)(x,\xi),x,\xi,\xi,\big(\nabla_x U\big)(x,\xi)\right)-
\end{equation}
$$
-i\left[\big(\big<\nabla_y,\nabla_\eta\big>+\big<\nabla_y,\nabla_\zeta\big>+\big<\nabla_z,\nabla_\zeta\big>\big)f\right]\left(x,\big(\nabla_\xi U\big)(x,\xi),x,\xi,\xi,\big(\nabla_x U\big)(x,\xi)\right)-
$$
$$
-i\Tr\left[\big(\nabla^2_{y,\zeta}f\big)\left(x,\big(\nabla_\xi U\big)(x,\xi),x,\xi,\xi,\big(\nabla_x U\big)(x,\xi)\right)\cdot\big(\nabla^2_{x,\xi}d\big)(x,\xi)\right]-
$$
$$
-\frac{i}{2}\Tr\left[\big(\nabla^2_{y,y}f\big)\cdot\big(\nabla^2_{\xi,\xi}U\big)+\big(\nabla^2_{\zeta,\zeta}f\big)\cdot\big(\nabla^2_{x,x}U\big)\right]\left(x,\big(\nabla_\xi U\big)(x,\xi),x,\xi,\xi,\big(\nabla_x U\big)(x,\xi)\right).
$$
\end{proposition}
\begin{proof}
 We proceed as in the proof of Proposition \ref{P3.2}, write $\tilde{f}=\rho_0\tilde{f}+\rho_\infty\tilde{f}$ and associated to this we have the decomposition
\begin{equation}\label{3.66}
 e=e_0+e_\infty,
\end{equation}
and for any $p\in\mathbb{N}$ 
\begin{equation}\label{3.67}
 \left|e_\infty(x,\xi)\right|\,\leq\,C_f<\xi>^{-p},\quad\forall(x,\xi)\in\mathbb{T}^*\mathbb{R}^d.
\end{equation}
Similar estimations are satisfied by the derivatives of $e_\infty$.

In order to estimate $e_0$ we start with some more notations:
\begin{equation}\label{3.68}
\overline{g}_1\equiv\overline{g}_1(x,\overline{z},\xi):=\int_0^1\big(\nabla_x d\big)(x-t\overline{z},\xi)\,dt;\qquad\overline{g}_2\equiv\overline{g}_2(x,\overline{z},\xi,\overline{\eta}):=\int_0^1\big(\nabla_\xi d\big)(x-\overline{z},\xi-t\overline{\eta})\,dt
\end{equation}
and notice the following equality
\begin{equation}\label{3.69}
 d(x-\overline{z},\xi-\overline{\eta})-d(x,\xi)=-\big<\overline{g}_1,\overline{z}\big>-\big<\overline{g}_2,\overline{\eta}\big>.
\end{equation}
We make the change of variables:
\begin{equation}\label{3.70}
 \tilde{\zeta}:=\overline{\zeta}+\overline{g}_1(x,\overline{z},\xi),\qquad\tilde{y}:=\overline{y}-\overline{g}_2(x,\overline{z},\xi,\overline{\eta}).
\end{equation}
We obtain
\begin{equation}\label{3.71}
 e_0(x,\xi)=\int_{\mathbb{R}^{4d}}e^{i(<\overline{z},\overline{\zeta}>-<\overline{y},\overline{\eta}>)}\overline{\rho}_0\overline{f}\,d\overline{y}\,d\overline{z}\,\dbar\overline{\eta}\,\dbar\overline{\zeta}
\end{equation}
where
\begin{equation}\label{3.72}
 \overline{f}(x,\overline{y},\overline{z},\xi,\overline{\eta},\overline{\zeta}):=\tilde{f}\big(x,\overline{y}-\overline{g}_2(x,\overline{z},\xi,\overline{\eta}),\overline{z},\xi,\overline{\eta},\overline{\zeta}+\overline{g}_1(x,\overline{z},\xi)\big),
\end{equation}
\begin{equation}\label{3.73}
 \overline{\rho}_0(x,\overline{y},\overline{z},\xi,\overline{\eta},\overline{\zeta}):=\rho_0\big(x,\overline{y}-\overline{g}_2(x,\overline{z},\xi,\overline{\eta}),\overline{z},\xi,\overline{\eta},\overline{\zeta}+\overline{g}_1(x,\overline{z},\xi)\big).
\end{equation}
We consider also the function 
\begin{equation}\label{3.74}
 \overline{f}_t\equiv\overline{f}_t(x,\overline{y},\overline{z},\xi,\overline{\eta},\overline{\zeta}):= \overline{f}(x,t\overline{y},\overline{z},\xi,\overline{\eta},\overline{\zeta}),
\end{equation}
and its Taylor development of second order with respect to the variable $t\in\mathbb{R}$:
\begin{equation}\label{3.75}
 \overline{f}=\overline{f}_0+\left.\frac{d}{dt}\overline{f}_t\right|_{t=0}+\int_0^1(1-t)\left(\frac{d}{dt}\right)^2\overline{f}_t\,dt.
\end{equation}
We have the evident equalities
\begin{equation}\label{3.76}
 \left.\frac{d}{dt}\overline{f}_t\right|_{t=0}=\left<\big(\nabla_{\overline{y}}\overline{f}\big)(x,0,\overline{z},\xi,\overline{\eta},\overline{\zeta}),\overline{y}\right>,
\end{equation}
\begin{equation}\label{3.77}
 \left(\frac{d}{dt}\right)^2\overline{f}_t=\left<\big(\nabla^2_{\overline{y},\overline{y}}\overline{f}\big)(x,t\overline{y},\overline{z},\xi,\overline{\eta},\overline{\zeta})\cdot \overline{y},\overline{y}\right>.
\end{equation}
Using the equality \eqref{3.75} we obtain an associated decomposition
\begin{equation}\label{3.78}
 e_0=e_{0,0}+e_{0,1}+r.
\end{equation}
We integrate by parts in each of the integrals of type \eqref{3.63} appearing in the above relation, by using the identity 
$$
\overline{y}e^{-i<\overline{y},\overline{\eta}>}=i\nabla_{\overline{\eta}}e^{-i<\overline{y},\overline{\eta}>}
$$
and taking into account that $\overline{\rho}_0$ does not depend on $\overline{y}$ we obtain that
\begin{equation}\label{3.79}
 e_{0,0}(x,\xi)=\int_{\mathbb{R}^{2d}}e^{i<\overline{z},\overline{\zeta}>}\overline{\rho}_0(x,0,\overline{z},\xi,0,\overline{\zeta})\overline{f}(x,0,\overline{z},\xi,0,\overline{\zeta})\,d\overline{z}\,\dbar\overline{\zeta},
\end{equation}
\begin{equation}\label{3.80}
 e_{0,1}(x,\xi)=-i\int_{\mathbb{R}^{2d}}e^{i<\overline{z},\overline{\zeta}>}\left<\nabla_{\overline{\eta}},\overline{\rho}_0\nabla_{\overline{y}}\overline{f}\right>(x,0,\overline{z},\xi,0,\overline{\zeta})\,d\overline{z}\,\dbar\overline{\zeta},
\end{equation}
\begin{equation}\label{3.81}
 r(x,\xi)=\int_0^1(1-t)r_t(x,\xi)dt,\quad\text{with}
\end{equation}
\begin{equation}\label{3.82}
 r_t(x,\xi):=-\int_{\mathbb{R}^{4d}}e^{i(<\overline{z},\overline{\zeta}>-<\overline{y},\overline{\eta}>)}\Tr\left\{\nabla^2_{\overline{\eta},\overline{\eta}}\cdot\left[\big(\overline{\rho}_0 \nabla^2_{\overline{y},\overline{y}}\overline{f}\big)(x,t\overline{y},\overline{z},\xi,\overline{\eta},\overline{\zeta})\right]\right\}\,d\overline{y}\,d\overline{z}\,\dbar\overline{\eta}\,\dbar\overline{\zeta}.
\end{equation}

\paragraph{\sf The rest.} We begin with the study of $r_t$ for which we shall obtain estimations that are uniform with respect to $t\in[0,1]$. We shall integrate by parts using the operators $\overline{L}_{\tilde{y}},\overline{L}_{\overline{z}},\overline{L}_{\overline{\eta}},\overline{L}_{\overline{\zeta}}$ defined in \eqref{3.31a} and \eqref{3.31b} and we shall make the decomposition $r_t=r'_t+r''_t$, where under the integral defining $r'_t$ the function $\overline{\rho}_0$ is not differentiated, and thus the function $\overline{f}$ will be at least two times differentiated with respect to $\overline{\eta}$. Then under the integral defining $r''_t$ there are only derivatives of $\overline{\rho}_0$. For the derivatives of $\overline{f}$ we have the following estimations:
\begin{equation}\label{3.84}
 \left|\big(\partial^\alpha_{\overline{z}}\partial^\beta_{\overline{\eta}}\partial^\gamma_{(\overline{y},\overline{\zeta})}\overline{f}\big)(x,\overline{y},\overline{z},\xi,\overline{\eta},\overline{\zeta})\right|\,\leq\,C_f<\xi>^m<\xi-\overline{\eta}>^{m'}<\xi-\overline{\eta}+\overline{\zeta}+\overline{g}_1>^{m''}\times
\end{equation}
$$
\times\left[1+\left(\frac{<\xi>}{<\xi-\overline{\eta}+\overline{\zeta}+\overline{g}_1>}\right)^{|\alpha|}\right]\times\underset{0\leq s\leq1}{\sup}\left[<\xi-s\overline{\eta}>^{-|\beta|}+<\xi-\overline{\eta}>^{-|\beta|}+<\xi-\overline{\eta}+\overline{\zeta}+\overline{g}_1>^{-|\beta|}\right],
$$
for any $\alpha$ and $\beta$ in $\mathbb{N}^d$ and any $\gamma\in\mathbb{N}^{2d}$. It follows that
\begin{equation}\label{3.85}
 \left|r'_t(x,\xi)\right|\leq C_f\int_{\overline{D}_0(x,\xi)}<\overline{y}>^{-2\overline{N}_3}<\overline{z}>^{-2\overline{N}_4}<\overline{\eta}>^{-2\overline{N}_1}<\overline{\zeta}>^{-2\overline{N}_2}<\xi>^m<\xi-\overline{\eta}>^{m'}<\xi-\overline{\eta}+\overline{\zeta}+\overline{g}_1>^{m''}\times
\end{equation}
$$
\times\left[1+\left(\frac{<\xi>}{<\xi-\overline{\eta}+\overline{\zeta}+\overline{g}_1>}\right)^{2\overline{N}_2}\right]\underset{0\leq s\leq1}{\sup}\left[<\xi-s\overline{\eta}>^{-2}+<\xi-\overline{\eta}>^{-2}+<\xi-\overline{\eta}+\overline{\zeta}+\overline{g}_1>^{-2}\right] \,d\overline{y}\,d\overline{z}\,\dbar\overline{\eta}\,\dbar\overline{\zeta},
$$
where
\begin{equation}\label{3.86}
 \overline{D}_0(x,\xi)=\left\{(\overline{y},\overline{z},\overline{\eta},\overline{\zeta})\mid\,<\overline{\eta}>+<\overline{\zeta}+\overline{g}_1-\overline{h}_0>\leq\epsilon|\xi-\overline{\eta}|\right\},
\end{equation}
with $\overline{g}_1$ defined in \eqref{3.68} and $\overline{h}_0$ defined by \eqref{3.13b}. Thus on $ \overline{D}_0(x,\xi)$ we have the following inequality that is satisfied: $|\overline{\eta}|\leq\epsilon|\xi-\overline{\eta}|$ and choosing a small enough $\epsilon_0$ and $\epsilon\in(0,\epsilon_0]$ we shall obtain
\begin{equation}\label{3.87}
 C^{-1}<\xi-\overline{\eta}>\,\leq\,<\xi-s\overline{\eta}>\,\leq\,C<\xi-\overline{\eta}>,\quad\forall s\in[0,1],\,\forall(\xi,\overline{\eta})\in\mathbb{R}^{2d}.
\end{equation}
On $ \overline{D}_0(x,\xi)$ we also have $|\overline{\zeta}+\overline{g}_1-\overline{h}_0|\leq\epsilon|\xi-\overline{\eta}|$ and thus using the identity $\overline{\zeta}+\xi-\overline{\eta}+\overline{g}_1=\overline{\zeta}+\overline{g}_1-\overline{h}_0+\xi-\overline{\eta}+\overline{h}_0$ and Lemma \ref{L1.5}, reducing maybe $\epsilon_0$, we also get
\begin{equation}\label{3.88}
 C^{-1}<\xi-\overline{\eta}>\,\leq\,<\overline{\zeta}+\xi-\overline{\eta}+\overline{g}_1>\,\leq\,C<\xi-\overline{\eta}>.
\end{equation}
Choosing $2\overline{N}_j\geq d+1$ in \eqref{3.85} for $1\leq j\leq 4$ and using \eqref{3.87} and \eqref{3.88} we obtain the estimation
\begin{equation}\label{3.89}
 \left|r'_t(x,\xi)\right|\leq C_f<\xi>^{m+m'+m''-2},\quad\forall(x,\xi)\in\mathbb{T}^*\mathbb{R}^d,\ \forall t\in[0,1].
\end{equation}

We come now back to the variables $\tilde{y}$ and $\tilde{\zeta}$ for the estimation of $r''_t$. We proceed as in the estimation of $e_\infty$ and obtain that for any $p\in\mathbb{N}$ 
\begin{equation}\label{3.90}
 \left|r''_t(x,\xi)\right|\leq C_f<\xi>^{-p},\quad\forall(x,\xi)\in\mathbb{T}^*\mathbb{R}^d,\ \forall t\in[0,1].
\end{equation}
Putting together \eqref{3.81}, \eqref{3.89} and \eqref{3.90} we obtain that
\begin{equation}\label{3.91}
 \left|r(x,\xi)\right|\leq C_f<\xi>^{m+m'+m''-2},\qquad\forall(x,\xi)\in\mathbb{T}^*\mathbb{R}^d.
\end{equation}

Repeating the arguments in the proof of Proposition \ref{P3.2} we can obtain similar estimations for the derivatives of $r$ and conclude that
\begin{equation}\label{3.92}
 r\,\in\,S^{m+m'+m''-2}(\mathbb{R}^d).
\end{equation}

\paragraph{\sf The zero order contribution.} Let us introduce the notations:
\begin{equation}\label{3.93}
 \left\{
\begin{array}{rcccccl}
 \tilde{h}_0&\equiv&\tilde{h}_0(x,\overline{z},\xi)&:=&\overline{h}_0(x,\overline{z},\xi,0)&=&\big(\nabla_x d\big)(x-\overline{z},\xi),\\
 \tilde{g}_1&\equiv&\tilde{g}_1(x,\overline{z},\xi)&:=&\overline{g}_1(x,\overline{z},\xi)&=&\int_0^1\big(\nabla_x d\big)(x-t\overline{z},\xi)\,dt,\\
 \tilde{g}_2&\equiv&\tilde{g}_2(x,\overline{z},\xi)&:=&\overline{g}_2(x,\overline{z},\xi,0)&=&\big(\nabla_\xi d\big)(x-\overline{z},\xi),
\end{array}
\right.
\end{equation}
\begin{equation}\label{3.94}
 \left\{
\begin{array}{rcccccl}
 \tilde{\rho}_0&\equiv&\tilde{\rho}_0(x,\overline{z},\xi,\tilde{\zeta})&:=&\rho_0(x,0,\overline{z},\xi,0,\tilde{\zeta})&=&1-\rho\left(\frac{\epsilon\xi}{1+<\tilde{\zeta}-\tilde{h}_0>}\right),\\
 \bar{\bar{\rho}}_0&\equiv&\bar{\bar{\rho}}_0(x,\overline{z},\xi,\overline{\zeta})&:=&\bar{\rho}_0(x,0,\overline{z},\xi,0,\overline{\zeta})&=&\rho_0(x,-\tilde{g}_2,\overline{z},\xi,0,\overline{\zeta}+\tilde{g}_1),\\
  \tilde{\rho}_\infty&:=&1-\tilde{\rho}_0,&&&& \\
\bar{\bar{\rho}}_\infty&:=&1- \bar{\bar{\rho}}_0,&&&&
\end{array}
\right.
\end{equation}
\begin{equation}\label{3.95}
 \bar{\bar{f}}\equiv\bar{\bar{f}}(x,\overline{z},\xi,\overline{\zeta}):=\bar{f}(x,0,\overline{z},\xi,0,\overline{\zeta})=\tilde{f}(x,-\tilde{g}_2,\overline{z},\xi,0,\overline{\zeta}+\tilde{g}_1).
\end{equation}
Then the equality \eqref{3.79} may be written as
\begin{equation}\label{3.96}
 e_{0,0}(x,\xi)=\int_{\mathbb{R}^{2d}}e^{i<\overline{z},\overline{\zeta}>}\bar{\bar{\rho}}_0\bar{\bar{f}}\,d\overline{z}\,\dbar\overline{\zeta}
\end{equation}
and we decompose this integral as a sum of 3 terms
\begin{equation}\label{3.97}
 e_{0,0}=e'_{0,0}+e''_{0,0}+\int_0^1(1-t)\check{r}_t\,dt,
\end{equation}
associated to the Taylor development of second order with respect to $\overline{\zeta}\in\mathbb{R}^d$ of the function $\bar{\bar{f}}$:
\begin{equation}\label{3.98}
 \bar{\bar{f}}(x,\overline{z},\xi,\overline{\zeta})=\bar{\bar{f}}(x,\overline{z},\xi,0)+\left<\big(\nabla_{\overline{\zeta}}\bar{\bar{f}}\big)(x,\overline{z},\xi,0),\overline{\zeta}\right>+\int_0^1(1-t)\left<\big(\nabla^2_{\overline{\zeta},\overline{\zeta}}\bar{\bar{f}}\big)(x,\overline{z},\xi,t\overline{\zeta})\cdot\overline{\zeta},\overline{\zeta}\right>\,dt.
\end{equation}
In particular
\begin{equation}\label{3.100}
 \check{r}_t(x,\xi)=\int_{\mathbb{R}^{2d}}e^{i<\overline{z},\overline{\zeta}>}\left(\bar{\bar{\rho}}_0\circ\kappa_t\right)(x,\overline{z},\xi,\overline{\zeta})\,d\overline{z}\,\dbar\overline{\zeta}
\end{equation}
with
\begin{equation}\label{3.99}
 \kappa_t(x,\overline{z},\xi,\overline{\zeta}):=\left<\big(\nabla^2_{\overline{\zeta},\overline{\zeta}}\bar{\bar{f}}\big)(x,\overline{z},\xi,t\overline{\zeta})\cdot\overline{\zeta},\overline{\zeta}\right>.
\end{equation}

Let us prove the estimation:
\begin{equation}\label{3.101}
 \left| \check{r}_t(x,\xi)\right|\,\leq\,C_f<\xi>^{m+m'+m''-2},\quad\forall(x,\xi)\in\mathbb{T}^*\mathbb{R}^d,\ \forall t\in[0,1].
\end{equation}
We begin by integrating by parts in \eqref{3.100} using the identities:
\begin{equation}\label{3.102}
 \left\{
\begin{array}{rcl}
 \overline{\zeta}e^{i<\overline{z},\overline{\zeta}>}&=&-i\nabla_{\overline{z}}e^{i<\overline{z},\overline{\zeta}>},\\
<\overline{\zeta}>^{2N_1}e^{i<\overline{z},\overline{\zeta}>}&=&(1-\Delta_{\overline{z}})^{N_1}e^{i<\overline{z},\overline{\zeta}>},\\
<\overline{z}>^{2N_2}e^{i<\overline{z},\overline{\zeta}>}&=&(1-\Delta_{\overline{\zeta}})^{N_2}e^{i<\overline{z},\overline{\zeta}>},
\end{array}
\right.
\end{equation}
with $N_1$ and $N_2$ two natural numbers to be chosen later. We obtain:
\begin{equation}\label{3.103}
  \check{r}_t(x,\xi)=-\int_{\mathbb{R}^{2d}}e^{i<\overline{z},\overline{\zeta}>}<\overline{z}>^{-2N_2}(1-\Delta_{\overline{\zeta}})^{N_2}<\overline{\zeta}>^{-2N_1}(1-\Delta_{\overline{z}})^{N_1}\mathfrak{r}_t(x,\overline{z},\xi,\overline{\zeta})\,d\overline{z}\,\dbar\overline{\zeta}
\end{equation}
with
$$
\mathfrak{r}_t(x,\overline{z},\xi,\overline{\zeta})=\Tr\left[\nabla^2_{\overline{z},\overline{z}}\cdot\left(\bar{\bar{\rho}}_0(x,\overline{z},\xi,\overline{\zeta})\big(\nabla^2_{\overline{\zeta},\overline{\zeta}}\bar{\bar{f}}\big)(x,\overline{z},\xi,t\overline{\zeta})\right)\right].
$$
For any multiindices $\alpha$ and $\beta$ in $\mathbb{N}^d$ we have the inequalities:
\begin{equation}\label{3.104}
 \left|\big(\partial^\alpha_{\overline{z}}\partial^\beta_{\overline{\zeta}}\bar{\bar{f}}\big)(x,\overline{z},\xi,t\overline{\zeta})\right|\leq C_f<\xi>^{m+m'}<t\overline{\zeta}+\xi+\tilde{g}_1>^{m''-|\beta|}\left[1+\left(\frac{<\xi>}{<t\overline{\zeta}+\xi+\tilde{g}_1>}\right)^{|\alpha|}\right],
\end{equation}
\begin{equation}\label{3.105}
 \left|\big(\partial^\alpha_{\overline{z}}\partial^\beta_{\overline{\zeta}}\bar{\bar{\rho}}_0\big)(x,\overline{z},\xi,\overline{\zeta})\right|\leq C\left[1+\left(\frac{<\xi>}{<\overline{\zeta}+\tilde{g}_1-\tilde{h}_0>}\right)^{|\alpha|}\right].
\end{equation}
Moreover, on the support of $\bar{\bar{\rho}}_0$ the following inequality is satisfied:
\begin{equation}\label{3.106}
 1+<\overline{\zeta}+\tilde{g}_1-\tilde{h}_0>\leq\epsilon|\xi|\quad\text{and thus}\quad |\overline{\zeta}+\tilde{g}_1-\tilde{h}_0|\leq\epsilon|\xi|.
\end{equation}
We use the identity $t\overline{\zeta}+\xi+\tilde{g}_1=t(\overline{\zeta}+\tilde{g}_1-\tilde{h}_0)+\xi+t\tilde{h}_0+(1-t)\tilde{g}_1$ and Lemma \ref{L1.5} to conclude that for some small enough $\epsilon_0>0$ and for any $\epsilon\in(0,\epsilon_0]$ we have the estimations:
\begin{equation}\label{3.107}
 C^{-1}<\xi>\,\leq\,<t\overline{\zeta}+\xi+\tilde{g}_1>\,\leq\,C<\xi>,\qquad\forall t\in[0,1].
\end{equation}
On the support of the derivatives of the function $\bar{\bar{\rho}}_0$ we also have the inequality
\begin{equation}\label{3.108}
 1+<\overline{\zeta}+\tilde{g}_1-\tilde{h}_0>\,\geq\,(\epsilon/2)|\xi|.
\end{equation}

It is clear now that \eqref{3.103}, \eqref{3.104}, \eqref{3.105}, \eqref{3.107} and \eqref{3.108} imply \eqref{3.101} if we choose $2N_j\geq d+1$ for $j=1,2$. Estimating the derivatives of $\check{r}_t$ we obtain that
\begin{equation}\label{3.109}
\check{r}:= \int_0^1(1-t)\check{r}_t\,dt\,\in\,S^{m+m'+m''-2}(\mathbb{R}^d).
\end{equation}

Let us consider now the term
\begin{equation}\label{3.110}
 e'_{0,0}(x,\xi):=\int_{\mathbb{R}^{2d}}e^{i<\overline{z},\overline{\zeta}>}\bar{\bar{\rho}}_0(x,\overline{z},\xi,\overline{\zeta})\bar{\bar{f}}(x,\overline{z},\xi,0)\,d\overline{z}\,\dbar\overline{\zeta},
\end{equation}
from \eqref{3.97}. Let us denote by $ e'_{\infty,0}$ the integral similar to the one above in which we replace $\bar{\bar{\rho}}_0$ by $\bar{\bar{\rho}}_\infty$ and notice that
$$
e'_{0,0}(x,\xi)+e'_{\infty,0}(x,\xi)=\int_{\mathbb{R}^{2d}}e^{i<\overline{z},\overline{\zeta}>}\bar{\bar{f}}(x,\overline{z},\xi,0)\,d\overline{z}\,\dbar\overline{\zeta}=\bar{\bar{f}}(x,0,\xi,0)=\tilde{f}(x,-\tilde{g}_2(x,0,\xi),0,\xi,0,\tilde{g}_1(x,0,\xi))=$$
\begin{equation}\label{3.110a}
=f\left(x,x+\big(\nabla_\xi d\big)(x,\xi),x,\xi,\xi,\xi+\big(\nabla_x d\big)(x,\xi)\right).
\end{equation}

We prove now that $ e'_{\infty,0}$ is a symbol of order $-\infty$. For that we consider once again the change of variables $\tilde{\zeta}:=\overline{\zeta}+\tilde{g}_1$, we take into account the equality $d(x-\overline{z},\xi)-d(x,\xi)=-<\overline{z},\bar{g}_1(x,\overline{z},\xi)>$and the fact that $\rho_0$ does not in fact depend on its second variable, so that we obtain:
\begin{equation}\label{3.113}
 e'_{\infty,0}(x,\xi)=\int_{\mathbb{R}^{2d}}e^{i(<\overline{z},\tilde{\zeta}>+d(x-\overline{z},\xi)-d(x,\xi))}\tilde{\rho}_\infty(x,\overline{z},\xi,\tilde{\zeta})g(x,\overline{z},\xi)\,d\overline{z}\,\dbar\tilde{\zeta},
\end{equation}
with $g\equiv g(x,\overline{z},\xi):=f\left(x,x+\tilde{g}_2(x,\overline{z},\xi),x-\overline{z},\xi,\xi,\xi+\tilde{g}_1(x,\overline{z},\xi)\right)$. In this integral we integrate by parts using the operators 
\begin{equation}\label{3.114}
 \check{L}_{\overline{z}}:=<\tilde{\zeta}-\tilde{h}_0>^{-2}\left(1-i\left<\tilde{\zeta}-\tilde{h}_0,\nabla_{\overline{z}}\right>\right),\qquad\check{L}_{\tilde{\zeta}}:=<\overline{z}>^{-2}\left(1-i\left<\overline{z},\nabla_{\tilde{\zeta}}\right>\right).
\end{equation}
This operators have the following formal adjoints:
\begin{equation}\label{3.115}
 ^t\check{L}_{\overline{z}}:=<\tilde{\zeta}-\tilde{h}_0>^{-1}\left(\check{a}_1+\left<\check{b}_1,\nabla_{\overline{z}}\right>\right),\qquad ^t\check{L}_{\tilde{\zeta}}:=<\overline{z}>^{-2}\left(1+i\left<\overline{z},\nabla_{\tilde{\zeta}}\right>\right),
\end{equation}
with the coefficients $\check{a}_1$ and $\check{b}_1$ satisfying the estimations:
\begin{equation}\label{3.116}
 \left|\big(\partial^\alpha_{\overline{z}}\partial^\beta_{\tilde{\zeta}}\check{b}_1\big)(x,\overline{z},\xi,\tilde{\zeta})\right|\leq C\left[1+\left(\frac{<\xi>}{<\tilde{\zeta}-\tilde{h}_0>}\right)^{|\alpha|}\right],
\end{equation}
\begin{equation}\label{3.117}
 \left|\big(\partial^\alpha_{\overline{z}}\partial^\beta_{\tilde{\zeta}}\check{a}_1\big)(x,\overline{z},\xi,\tilde{\zeta})\right|\leq C\left[1+\left(\frac{<\xi>}{<\tilde{\zeta}-\tilde{h}_0>}\right)^{1+|\alpha|}\right],
\end{equation}
for any multiindices $\alpha$ and $\beta$ in $\mathbb{N}^d$. We obtain
\begin{equation}\label{3.118}
 e'_{\infty,0}(x,\xi)=\int_{\check{D}_\infty(x,\xi)}e^{i(<\overline{z},\tilde{\zeta}>+d(x-\overline{z},\xi)-d(x,\xi))}\big(^t\check{L}_{\overline{z}}\big)^{N_1}\big(^t\check{L}_{\tilde{\zeta}}\big)^{N_2}\big(\tilde{\rho}_\infty g\big)\,d\overline{z}\,\dbar\tilde{\zeta}
\end{equation}
with $N_1$ and $N_2$ suitable choosen and with
\begin{equation}\label{3.119}
 \check{D}_\infty(x,\xi):=\left\{(\overline{z},\tilde{\zeta})\in\mathbb{R}^{2d}\mid1+<\tilde{\zeta}-\tilde{h}_0>\geq(\epsilon/2)|\xi|\right\}.
\end{equation}
In order to proceed further we use again Lemma \ref{L1.5} that implies that the weights $<\xi>$ and $<\xi+\tilde{g}_1(x,\overline{z},\xi)>$ are equivalent and deduce that
\begin{equation}\label{3.121}
 \left|\big(\partial^\alpha_{\overline{z}}g\big)(x,\overline{z},\xi)\right|\,\leq\,C_f<\xi>^{m+m'+m''},\qquad\forall\alpha\in\mathbb{N}^d.
\end{equation}
We also have
\begin{equation}\label{3.120}
\left|\big(\partial^\alpha_{\overline{z}}\partial^\beta_{\tilde{\zeta}}\tilde{\rho}_\infty\big)(x,\overline{z},\xi,\tilde{\zeta})\right|\leq C\left[1+\left(\frac{<\xi>}{<\tilde{\zeta}-\tilde{h}_0>}\right)^{|\alpha|}\right],\qquad\forall(\alpha,\beta)\in[\mathbb{N}^d]^2.
\end{equation}
Putting all these results together we conclude that
\begin{equation}\label{3.122}
 \left|e'_{\infty,0}(x,\xi)\right|\leq C_f\int_{\check{D}_\infty(x,\xi)}<\tilde{\zeta}-\tilde{h}_0>^{-N_1}<\overline{z}>^{-N_2}<\xi>^{m+m'+m''}\,d\overline{z}\,\dbar\tilde{\zeta}.
\end{equation}
On $\check{D}_\infty(x,\xi)$ we have that $<\tilde{\zeta}-\tilde{h}_0>\geq C<\xi>$ so that after choosing $N_2\geq d+1$ and $N_1\geq p+m+m'+m''+d+1$ (for any fixed $p\in\mathbb{N}$), we obtain that $\left|e'_{\infty,0}\right|\leq C_f<\xi>^{-p}$ for any $(x,\xi)\in\mathbb{T}^*\mathbb{R}^d$. Repeating this procedure for any derivative of $e'_{\infty,0}$ we conclude that
\begin{equation}\label{3.111}
e'_{\infty,0}\,\in\,S^{-\infty}(\mathbb{R}^d).
\end{equation}

Putting together \eqref{3.110a} and \eqref{3.111} we have thus proved the following result
\begin{equation}\label{3.123}
 e'_{0,0}(x,\xi)-f\left(x,\big(\nabla_\xi U\big)(x,\xi),x,\xi,\xi,\big(\nabla_x U\big)(x,\xi)\right)\,\in\,S^{-\infty}(\mathbb{R}^d).
\end{equation}

Let us consider now the term $e''_{0,0}$ in \eqref{3.97}. We integrate by parts using the identity $\overline{\zeta}e^{i<\overline{z},\overline{\zeta}>}=-i\nabla_{\overline{z}}e^{i<\overline{z},\overline{\zeta}>}$ and obtain
\begin{equation}\label{3.124}
 e''_{0,0}(x,\xi)=i\int_{\mathbb{R}^{2d}}e^{i<\overline{z},\overline{\zeta}>}\left<\nabla_{\overline{z}},\bar{\bar{\rho}}_0(x,\overline{z},\xi,\overline{\zeta})\big(\nabla_{\overline{\zeta}}\bar{\bar{f}}\big)(x,\overline{z},\xi,0)\right>  \,d\overline{z}\,\dbar\overline{\zeta}.
\end{equation}
Taking now into account the evident identity: $\nabla_{\overline{z}}\bar{\bar{\rho}}_0=-\nabla_{\overline{z}}\bar{\bar{\rho}}_\infty$ we obtain by repeating the arguments concerning $ e'_{0,0}$, that
\begin{equation}\label{3.125}
 e''_{0,0}(x,\xi)-i\left(\left<\nabla_{\overline{z}},\nabla_{\overline{\zeta}}\right>\bar{\bar{f}}\right)(x,0,\xi,0)\,\in\,S^{-\infty}(\mathbb{R}^d).
\end{equation}
Let us compute the derivatives of $\bar{\bar{f}}$:
$$
\nabla_{\overline{\zeta}}\bar{\bar{f}}(x,\overline{z},\xi,\overline{\zeta})=\left(\nabla_{\zeta}f \right)\big(x,x+\tilde{g}_2(x,\overline{z},\xi),x-\overline{z},\xi,\xi,\overline{\zeta}+\xi+\tilde{g}_1(x,\overline{z},\xi)\big),
$$
\begin{equation}\label{3.126}
 \left<\nabla_{\overline{z}},\nabla_{\overline{\zeta}}\right>\bar{\bar{f}}=-\left<\nabla_{z},\nabla_{\zeta}\right>f+\Tr\left[\big(\nabla^2_{y,\zeta}f\big)\cdot\big(\nabla_{\overline{z}}\tilde{g}_2\big)\,+\,\big(\nabla^2_{\zeta,\zeta}f\big)\cdot\big(\nabla_{\overline{z}}\tilde{g}_1\big)\right].
\end{equation}
Considering now \eqref{3.93} we remark that
\begin{equation}\label{3.127}
 \left\{
\begin{array}{rcccl}
 \big(\nabla_{\overline{z}}\tilde{g}_2\big)(x,0,\xi)&=&-\big(\nabla^2_{x,\xi}d\big)(x,\xi)&=&1_d\,-\,\big(\nabla^2_{x,\xi}U\big)(x,\xi),\\
\big(\nabla_{\overline{z}}\tilde{g}_1\big)(x,0,\xi)&=&-(1/2)\big(\nabla^2_{x,x}d\big)(x,\xi)&=&-(1/2)\big(\nabla^2_{x,x}U\big)(x,\xi).
\end{array}
\right.
\end{equation}
We can thus conclude that:
\begin{equation}\label{3.128}
  e''_{0,0}(x,\xi)\,+\,i\left(\left<\nabla_{z},\nabla_{\zeta}\right>f\right)\big(x,\nabla_\xi  U(x,\xi),x,\xi,\xi,\nabla_x U(x,\xi)\big)-
\end{equation}
$$
-i\Tr\left[\big(\nabla^2_{y,\zeta}f\big)\big(x,\nabla_\xi  U(x,\xi),x,\xi,\xi,\nabla_x U(x,\xi)\big)\big(\nabla^2_{x,\xi}d\big)(x,\xi)+\right.
$$
$$
\left.+(1/2)\big(\nabla^2_{\zeta,\zeta}f\big)\big(x,\nabla_\xi  U(x,\xi),x,\xi,\xi,\nabla_x U(x,\xi)\big)\big(\nabla^2_{x,x}U\big)(x,\xi)\right]\,\in\,S^{-\infty}(\mathbb{R}^d).
$$

\paragraph{\sf The first order contribution.} Staring from \eqref{3.80} let us compute 
$$
\left<\nabla_{\overline{\eta}},\bar{\rho}_0\nabla_{\overline{y}}\bar{f}\right>=\left<\nabla_{\overline{\eta}}\bar{\rho}_0,\nabla_{\overline{y}}\bar{f}\right>+\bar{\rho}_0\left<\nabla_{\overline{\eta}},\nabla_{\overline{y}}\right>\bar{f}
$$
and observe that as previously, the term containing $\nabla_{\overline{\eta}}\bar{\rho}_0$ in the integral \eqref{3.80} is a symbol of class $S^{-\infty}(\mathbb{R}^d)$. For the integral containing $\bar{\rho}_0$ we notice that
\begin{equation}\label{3.129}
 \left( \left<\nabla_{\overline{\eta}},\nabla_{\overline{y}}\right>\bar{f}\right)(x,0,\overline{z},\xi,0,\overline{\zeta})=\left\{\left<\nabla_{\overline{\eta}},\nabla_{\tilde{y}}\right>\tilde{f}\,-\,\Tr\left[\big(\nabla^2_{\tilde{y},\tilde{y}}\tilde{f}\big)\cdot\big(\nabla_{\overline{\eta}}\bar{g}_2\big)\right]\right\}\big(x,-\bar{g}_2(x,\overline{z},\xi,0),\overline{z},\xi,0,\overline{\zeta}+\bar{g}_1(x,\overline{z},\xi)\big),
\end{equation}
that is of the same form as the formula of $e_{0,0}$ in \eqref{3.79} but with the order $m+m'+m''$ of the symbol replaced by $m+m'+m''-1$. Thus for its estimate we shall only use a first order Taylor development with respect to $\overline{\zeta}$; the remainder will thus be a symbol of order $m+m'+m''-2$ and the main term will be treated like in the case of $e'_{0,0}$ defined by \eqref{3.110}. This leads us to the conclusion
\begin{equation}\label{3.130}
 e_{0,1}(x,\xi)\,+\,i\left( \left<\nabla_{\overline{\eta}},\nabla_{\overline{y}}\right>\bar{f}\right)(x,0,0,\xi,0,0)\,\in\,S^{m+m'+m''-2}(\mathbb{R}^d).
\end{equation}
Let us notice that from \eqref{3.11} we obtain
\begin{equation}\label{3.131}
 \left<\nabla_{\overline{\eta}},\nabla_{\tilde{y}}\right>\tilde{f}\big(x,-\bar{g}_2(x,0,\xi,0),0,\xi,0,\bar{g}_1(x,0,\xi)\big)=
\end{equation}
\begin{equation*}
=
\left(\left<\nabla_{y},\nabla_{\eta}\right>f+\left<\nabla_{y},\nabla_{\zeta}\right>f\right)\big(x,x+\nabla_\xi d(x,\xi),x,\xi,\xi,\xi+\nabla_x d(x,\xi)\big).
\end{equation*}
We also notice that
\begin{equation}\label{3.132}
 \big(\nabla_{\overline{\eta}}\bar{g}_2\big)(x,0,\xi,0)=-(1/2)\nabla^2_{\xi,\xi}d(x,\xi)=-(1/2)\nabla^2_{\xi,\xi}U(x,\xi).
\end{equation}
Putting these remarks together with \eqref{3.129} we deduce that
\begin{equation}\label{3.133}
  e_{0,1}(x,\xi)\,+\,i\left( \left<\nabla_{y},\nabla_{\eta}\right>f+\left<\nabla_{y},\nabla_{\zeta}\right>f\right)\big(x,\nabla_\xi U(x,\xi),x,\xi,\xi,\nabla_x U(x,\xi)\big)+
\end{equation}
$$
+(i/2)\Tr\left[\big(\nabla^2_{y,y}f\big)\big(x,\nabla_\xi U(x,\xi),x,\xi,\xi,\nabla_x U(x,\xi)\big)\cdot\nabla^2_{\xi,\xi}U(x,\xi)\right]\,\in\,S^{m+m'+m''-2}(\mathbb{R}^d).
$$
From \eqref{3.66}, \eqref{3.78} and \eqref{3.97} we see that
\begin{equation}\label{3.134}
 e\,=\,e'_{0,0}+e''_{0,0}+e_{0,1}+e_\infty+r+\check{r}.
\end{equation}
But from \eqref{3.67}, \eqref{3.92} and \eqref{3.109} we deduce that $e_\infty+r+\check{r}$ is a symbol of class $S^{m+m'+m''-2}(\mathbb{R}^d)$ so that \eqref{3.64} follows from \eqref{3.65}, \eqref{3.123}, \eqref{3.128} and \eqref{3.133}.
\end{proof}

We shall now show that we can reduce our general situation to the case $B=0$ that we have studied. We shall suppose the assumptions of Propositions \ref{P3.2} and \ref{P3.3}, making also use of the notations used there. We denote by
\begin{equation}\label{3.135}
 e_B(x,\xi):=\left<M(x,\xi)\cdot\big(\nabla_\zeta f\big)\big(x,\nabla_\xi U(x,\xi),x,\xi,\xi,\nabla_x U(x,\xi)\big),\big(\nabla_\xi d\big)(x,\xi)\right>\,\in\,S^{m+m'+m''-1}(\mathbb{R}^d),
\end{equation}
where
\begin{equation}\label{3.136}
 M(x,\xi):=C\big(x,x,\nabla_\xi U(x,\xi)\big)=\int_0^1(1-t)B\big(x+t\nabla_\xi d(x,\xi)\big)\,dt.
\end{equation}
\begin{proposition}\label{P3.4}
 Under the assumptions of Proposition \ref{P3.2} and use the notations \eqref{3.135} above, we have that
$$
c\,-\,e\,-\,e_B\,\in\,S^{m+m'+m''-2}(\mathbb{R}^d).
$$
\end{proposition}
\begin{proof}
 We have 
\begin{equation}\label{3.138}
 \tilde{\Omega^B}(x,\tilde{y},\overline{z})=e^{-i\tilde{F}(x,\tilde{y},\overline{z})},\quad\tilde{F}(x,\tilde{y},\overline{z}):=F(x,x-\overline{z},x-\tilde{y})
\end{equation}
and from \eqref{1.13}
\begin{equation}\label{3.139}
 \tilde{F}(x,\tilde{y},\overline{z})\,=\,\left<\tilde{C}(x,\tilde{y},\overline{z})\overline{z}\,,\,\tilde{y}\right>,
\end{equation}
where
\begin{equation}\label{3.140}
\tilde{C}(x,\tilde{y},\overline{z})\,=\,-C(x,x-\overline{z},x-\tilde{y})\,=\,-\int_0^1ds\int_0^{1-s}dt\,B(x-t\overline{z}-s\tilde{y}).
\end{equation}
We use the equality 
\begin{equation}\label{3.141}
 \tilde{\Omega^B}=1-i\tilde{F}+\tilde{R}_1,\quad\text{with }\tilde{R}_1:=-\tilde{F}^2\int_0^1(1-t)e^{-it\tilde{F}}dt
\end{equation}
and the associated decomposition of $c(x,\xi)$ defined by \eqref{3.9}
\begin{equation}\label{3.142}
 c\,=\,e\,+\,\tilde{e}_B\,+\,\tilde{r}_1,
\end{equation}
with $e$ defined by \eqref{3.63}, $\tilde{e}_B$ obtained by replacing $\tilde{\Omega^B}$ with $-i\tilde{F}$ in the integral \eqref{3.9} and $\tilde{r}_1$ obtained  by replacing $\tilde{\Omega^B}$ with $\tilde{R}_1$ in the integral \eqref{3.9}.

Let us first prove that
\begin{equation}\label{3.143}
 \tilde{r}_1\,\in\,S^{m+m'+m''-2}(\mathbb{R}^d).
\end{equation}
In order to obtain this result we begin by integrating by parts in the integral defining $\tilde{r}_1$, using the identities
$$
\tilde{y}e^{-i<\tilde{y},\overline{\eta}>}=i\nabla_{\overline{\eta}}e^{-i<\tilde{y},\overline{\eta}>},\quad\overline{z}e^{i<\overline{z},\tilde{\zeta}>}=-i\nabla_{\tilde{\zeta}}e^{i<\overline{z},\tilde{\zeta}>}
$$and noticing that $\tilde{F}^2$ contains a factor $\overline{z}^\alpha$ with $|\alpha|=2$, while $d(x-\overline{z},\xi-\overline{\eta})-d(x,\xi)$ does not depend on $\tilde{\zeta}$ and $d(x,\overline{\eta})$ has derivatives with respect to $\overline{\eta}$ that are symbols of order 0. Thus we can write:
$$
\tilde{r}_1(x,\xi)\,=\,\int_0^1(1-t)\tilde{r}_1(t;x,\xi)\,dt,
$$
with $\tilde{r}_1(t;x,\xi)$ defined by an integral of type \eqref{3.9} in which $\tilde{\Omega^B}$ is replaced by $e^{-it\tilde{F}}$ and the factor $\tilde{f}$ is replaced by some symbol from $S^{m,m',m''-2}\big(\mathbb{R}^{3d}\times\mathbb{R}^d\times\mathbb{R}^d\times\mathbb{R}^d\big)$.
We use Proposition \ref{P3.2} obtaining estimations uniform with respect to $t\in[0,1]$ and thus \eqref{3.143}.

We come now to the study of $\tilde{e}_B$. Let us write down its explicit form:
\begin{equation}\label{3.144}
 \tilde{e}_B(x,\xi)\,=\,-i\int_{\mathbb{R}^{4d}}e^{i(<\overline{z},\tilde{\zeta}>-<\tilde{y},\overline{\eta}>)}\left<\tilde{C}(x,\tilde{y},\overline{z})\nabla_{\tilde{\zeta}},\nabla_{\overline{\eta}}\right>\left[e^{i(d(x-\overline{z},\xi-\overline{\eta})-d(x,\xi))}\tilde{f}\right]\,d\tilde{y}\,d\overline{z}\,\dbar\overline{\eta}\,\dbar\tilde{\zeta}.
\end{equation}
Let us notice that in the above expression the operator $\nabla_{\tilde{\zeta}}$ only applies to $\tilde{f}$. Applying $\nabla_{\overline{\eta}}$ produces two types of terms: those with derivatives of $\tilde{f}$ and with the exponential left non-differentiated (these integrals are symbols of class $S^{m+m'+m''-2}(\mathbb{R}^d)$) and the other terms containing derivatives of the exponential. We conclude that $\tilde{e}_B$ is congruent modulo $S^{m+m'+m''-2}(\mathbb{R}^d)$) with
\begin{equation}\label{3.145}
 e'_B(x,\xi)\,:=\,-\int_{\mathbb{R}^{4d}}e^{i\tilde{\Psi}}\left<\tilde{C}(x,\tilde{y},\overline{z})\big(\nabla_{\tilde{\zeta}}\tilde{f}\big),\big(\nabla_{\xi}d\big)(x-\overline{z},\xi-\overline{\eta})\right>\,d\tilde{y}\,d\overline{z}\,\dbar\overline{\eta}\,\dbar\tilde{\zeta}.
\end{equation}
The scalar product under the above integral defines a function of the form $\tilde{\Phi}(x,\tilde{y},\overline{z},\xi,\overline{\eta},\tilde{\zeta})$ with
\begin{equation}\label{3.146}
 \Phi(x,y,z,\xi,\eta,\zeta):=-\left<C(x,z,y)\cdot\big(\nabla_\zeta f\big)(x,y,z,\xi,\eta,\zeta)^{}_{},\big(\nabla_\xi d\big)(z,\eta)\right>
\end{equation}
being a symbol of class $S^{m,m',m''-1}(\mathbb{R}^{3d}\times\mathbb{R}^d\times\mathbb{R}^d\times\mathbb{R}^d\big)$. By the result of Proposition \ref{P3.3} the symbol $e'_B$ is congruent modulo $S^{m+m'+m''-2}(\mathbb{R}^d)$) with
$$
e_B(x,\xi):=-\Phi\big(x,\nabla_\xi U(x,\xi),x,\xi,\xi,\nabla_x U(x,\xi)\big)=
$$
$$
=\left<C\big(x,x,\nabla_\xi U(x,\xi)\big)\cdot\big(\nabla_\zeta f\big)\big(x,\nabla_\xi U(x,\xi),x,\xi,\xi,\nabla_x U(x,\xi)\big),\big(\nabla_\xi d\big)(x,\xi)\right>,
$$
with $C(x,y,z)$ defined by \eqref{1.12}. Using the definitions \eqref{3.135} and \eqref{3.136}we end the proof of the Proposition.
\end{proof}

By putting together the Propositions \ref{P3.3} and \ref{P3.4} we obtain the proof of the following statement.
\begin{theorem}\label{T3.5}
 Suppose $f\in S^{m,m',m''}(\mathbb{R}^{3d}\times\mathbb{R}^d\times\mathbb{R}^d\times\mathbb{R}^d\big)$ and define $c$ by the formula \eqref{3.7}. Then:
$$
c\,-\,(c_0+e_B)\,\in\,S^{m+m'+m''-2}(\mathbb{R}^d)
$$
with $c_0$ defined by \eqref{3.65} and $e_B$ defined by \eqref{3.135} and \eqref{3.136}.
\end{theorem}

The situation that is of main interest for us is the case when $c$ is the symbol of the compound operator $a^A(x,D)\circ\mathfrak{Op}^A_\Phi(b)$. In this case, following \eqref{3.4} and the Propositions \ref{P3.1} and \ref{P3.2}, the symbol $c$ is defined by formula \eqref{3.7} with
$$
f(x,y,z,\xi,\eta,\zeta)\,=\,a(x,\zeta)b(z,\eta);
$$
in particular, $m=0$ and the symbol $f$ does not in fact depend on the variables $y$ and $\xi$. For this particular case our Theorem \ref{T3.5} becomes
\begin{theorem}\label{T3.6}
 Suppose given $a\in S^{m''}(\mathbb{R}^d)$ and $b\in S^{m'}(\mathbb{R}^d)$. Then
$$
a^A(x,D)\circ\mathfrak{Op}^A_\Phi(b)\,=\,\mathfrak{Op}^A_\Phi(c)
$$
with $c\in S^{m'+m''}(\mathbb{R}^d)$ and 
\begin{equation}\label{3.148}
 c\,-\,(c_0+e_B)\,\in\,S^{m+m'+m''-2}(\mathbb{R}^d)
\end{equation}
with
\begin{equation}\label{3.149}
 c_0(x,\xi)\,:=\,a\big(x,\nabla_x U(x,\xi)\big)b(x,\xi)\,-\,i\left<\big(\nabla_\zeta a\big)\big(x,\nabla_x U(x,\xi)\big),\big(\nabla_z b\big)(x,\xi)\right>\,-
\end{equation}
$$
-\,(i/2)\Tr\left[\big(\nabla^2_{\zeta,\zeta}a\big)\big(x,\nabla_x U(x,\xi)\big)\cdot\big(\nabla^2_{x,x}U\big)(x,\xi)\right]b((x,\xi)
$$
and
\begin{equation}\label{3.150a}
e_B(x,\xi)\,:=\,\left<M(x,\xi)\cdot\big(\nabla_\zeta a\big)\big(x,\nabla_x U(x,\xi)\big),\big(\nabla_\xi d\big)(x,\xi)\right>b(x,\xi),
\end{equation}
\begin{equation}\label{3.150b}
 M(x,\xi)\,=\,\int_0^1(1-t)B\big(x+t\nabla_\xi d(x,\xi)\big)\,dt.
\end{equation}
\end{theorem}

\section{Product of a FIO with a $\Psi$DO}

Suppose once again given two symbols $a\in S^{m''}(\mathbb{R}^d)$ and $b\in S^{m'}(\mathbb{R}^d)$ and the operators $E:=\mathfrak{Op}^A_\Phi(a)$ and $F:=b^A(x,D)$. Then for any $u\in\mathcal{S}(\mathbb{R}^d)$ and any point $x\in\mathbb{R}^d$ we have
\begin{equation}\label{4.1}
 \big(Eu\big)(x)\,=\,\int_{\mathbb{R}^{2d}}e^{i\big(U(x,\zeta)-<z,\zeta>\big)}\omega^A(x,z)a(x,\zeta)u(z)\,dz\,\dbar\zeta,
\end{equation}
\begin{equation}\label{4.2}
  \big(Fu\big)(x)\,=\,\int_{\mathbb{R}^{2d}}e^{i<x-y,\eta>}\omega^A(x,y)b(x,\eta)u(y)\,dy\,\dbar\eta.
\end{equation}

Then $E\circ F$ is a well defined operator in $\mathbb{B}\big(\mathcal{S}(\mathbb{R}^d)\big)$. We want to prove that it is in fact a FIO and to compute explicitly its principal symbol. For that we proceed as in the previous section. We choose $\chi\in C^\infty_0(\mathbb{R}^{2d})$ with $\chi(0,0)=1$ and for $\epsilon\in(0,1]$ we define
$$
a_\epsilon(x,z,\zeta):=\chi(\epsilon z,\epsilon\zeta)a(x,\zeta),\quad b_\epsilon(y,z,\eta):=\chi(\epsilon y,\epsilon\eta)b(z,\eta).
$$
We still denote by $E_\epsilon$ the operator defined by \eqref{4.1} with $a$ replaced by $a_\epsilon$ and by $F_\epsilon$ the operator defined by \eqref{4.2} with $b$ replaced by $b_\epsilon$. Then, for any $v\in\mathcal{S}(\mathbb{R}^d)$ we have $\underset{\epsilon\rightarrow0}{\lim}\big(E_\epsilon\circ F_\epsilon\big)v=\big(E\circ F\big)v$. Moreover, a direct computation allows to obtain the distribution kernel $K_\epsilon\in\mathcal{S}(\mathbb{R}^{2d})$ of $E_\epsilon\circ F_\epsilon$:
\begin{equation}\label{4.3}
 K_\epsilon(x,y)=\omega^A(x,y)\int_{\mathbb{R}^{3d}}e^{i(<z-y,\eta>-<z,\zeta>+U(x,\zeta))}\Omega^B(x,z,y)a_\epsilon(x,z,\zeta)b_\epsilon(y,z,\eta)\,dz\,\dbar\eta\,\dbar\zeta.
\end{equation}
We use \eqref{2.5} to deduce that
$
 E_\epsilon\circ F_\epsilon=\mathfrak{Op}^A_\Phi(c_\epsilon),\quad\text{for }c_\epsilon\in\mathcal{S}(\mathbb{R}^{2d}),
$
where
\begin{equation}\label{4.5}
c_\epsilon(x,\xi)\,:=\,\int_{\mathbb{R}^{4d}}e^{i\psi(x,y,z,\xi,\eta,\zeta)}\Omega^B(x,z,y)a_\epsilon(x,z,\zeta)b_\epsilon(y,z,\eta)\,dy\,dz\,\dbar\eta\,\dbar\zeta,
\end{equation}
\begin{equation}\label{4.6}
 \psi(x,y,z,\xi,\eta,\zeta)\,:=\,<y,\xi-\eta>+<z,\eta-\zeta>+U(x,\zeta)-U(x,\xi).
\end{equation}

Suppose that we have proved the existence of a symbol $c\in S^{m'+m''}(\mathbb{R}^d)$ such that
\begin{equation}\label{4.7}
\underset{\epsilon\rightarrow0}{\lim}c_\epsilon\,=\,c\quad\text{in }S^{m'+m''}(\mathbb{R}^d).
\end{equation}
It follows that
\begin{proposition}\label{P4.1}
 If $a\in S^{m''}(\mathbb{R}^d)$ and $b\in S^{m'}(\mathbb{R}^d)$, then $\mathfrak{Op}^A_\Phi(a)\circ b^A(x,D)=\mathfrak{Op}^A_\Phi(c)$ with $c\in S^{m'+m''}(\mathbb{R}^d)$ defined by \eqref{4.7}.
\end{proposition}
Coming now back to \eqref{4.7}, we notice that it will follow from the next statement that we shall now prove.
\begin{proposition}\label{P4.2}
 Assume $f\in S^{m,m',m''}\big(\mathbb{R}^{3d}\times\mathbb{R}^d\times\mathbb{R}^d\times\mathbb{R}^d\big)$ and $\psi$ is the phase function defined in \eqref{4.6}. Then the oscillatory integral
\begin{equation}\label{4.8}
 c(x,\xi)\,:=\,\int_{\mathbb{R}^{4d}}e^{i\psi(x,y,z,\xi,\eta,\zeta)}\Omega^B(x,z,y)f(x,y,z,\xi,\eta,\zeta)\,dy\,dz\,\dbar\eta\,\dbar\zeta
\end{equation}
is well defined as an element from $S^{m+m'+m''}(\mathbb{R}^d)$ and the map
$$
S^{m,m',m''}\big(\mathbb{R}^{3d}\times\mathbb{R}^d\times\mathbb{R}^d\times\mathbb{R}^d\big)\ni f\mapsto c\in S^{m+m'+m''}(\mathbb{R}^d)
$$
is continuous.
\end{proposition}
\begin{proof}
 As in the proof of Proposition \ref{P3.2}, we can work with a function $f$ with compact support with respect to the variables $(y,z,\eta,\zeta)$ and estimate the seminorms of $c\in S^{m+m'+m''}(\mathbb{R}^d)$ with some seminorms of $f\in S^{m,m',m''}\big(\mathbb{R}^{3d}\times\mathbb{R}^d\times\mathbb{R}^d\times\mathbb{R}^d\big)$ (that we shall generically denote by $C_f$).

We make the change of variables:
\begin{equation}\label{4.9}
 \bar{y}:=x-y,\quad\bar{z}:=x-z,\quad\bar{\eta}:=\eta-\xi,\quad\bar{\zeta}:=\zeta-\eta
\end{equation}
and obtain
\begin{equation}\label{4.10}
  c(x,\xi)\,:=\,\int_{\mathbb{R}^{4d}}e^{i\bar{\psi}}\bar{\Omega^B}\bar{f}\,d\bar{y}\,d\bar{z}\,\dbar\bar{\eta}\,\dbar\bar{\zeta}
\end{equation}
with the notations
\begin{equation}\label{4.11}
 \bar{\psi}(x,\bar{y},\bar{z},\xi,\bar{\eta},\bar{\zeta}):=\psi(x,x-\bar{y},x-\bar{z},\xi,\xi+\bar{\eta},\xi+\bar{\eta}+\bar{\zeta})=<\bar{y},\bar{\eta}>+<\bar{z},\bar{\zeta}>+\bar{d}(x,\xi,\bar{\eta},\bar{\zeta}),
\end{equation}
\begin{equation}\label{4.12}
 \bar{d}(x,\xi,\bar{\eta},\bar{\zeta}):=d(x,\xi+\bar{\eta}+\bar{\zeta})-d(x,\xi),
\end{equation}
\begin{equation}\label{4.13}
 \bar{f}(x,\bar{y},\bar{z},\xi,\bar{\eta},\bar{\zeta}):=f(x,x-\bar{y},x-\bar{z},\xi,\xi+\bar{\eta},\xi+\bar{\eta}+\bar{\zeta}),
\end{equation}
\begin{equation}\label{4.14}
 \bar{\Omega^B}(x,\bar{y},\bar{z}):=\Omega^B(x,x-\bar{z},x-\bar{y})=e^{-i\bar{F}(x,\bar{y},\bar{z})},\qquad\bar{F}(x,\bar{y},\bar{z}):=F(x,x-\bar{z},x-\bar{y}).
\end{equation}
We continue by integrating by parts using the identities:
\begin{equation}\label{4.15}
 \big(1-\Delta_{\bar{y}}\big)e^{i<\bar{y},\bar{\eta}>}=<\bar{\eta}>^2e^{i<\bar{y},\bar{\eta}>},\qquad\big(1-\Delta_{\bar{\eta}}\big)e^{i<\bar{y},\bar{\eta}>}=<\bar{y}>^2e^{i<\bar{y},\bar{\eta}>},
\end{equation}
$$
 \big(1-\Delta_{\bar{z}}\big)e^{i<\bar{z},\bar{\zeta}>}=<\bar{\zeta}>^2e^{i<\bar{z},\bar{\zeta}>},\qquad\big(1-\Delta_{\bar{\zeta}}\big)e^{i<\bar{z},\bar{\zeta}>}=<\bar{z}>^2e^{i<\bar{z},\bar{\zeta}>}.
$$
Notice that $\bar{d}$ does not depend on the variables $(\bar{y},\bar{z})$, and its derivatives with respect to the variables $(\bar{\eta},\bar{\zeta})$ are functions of class $BC^\infty(\mathbb{R}^{4d})$. Although the derivatives of $\bar{\Omega^B}$ with respect with $(\bar{y},\bar{z})$ produce monomials in $(\bar{y},\bar{z})$ multiplied with functions of class $BC^\infty(\mathbb{R}^{3d})$, we can get rid of these unbounded monomials by further integrating by parts. Finally we obtain that
$$
\left|c(x,\xi)\right|\leq C_f\int_{\mathbb{R}^{4d}}<\bar{y}>^{-2N_1}<\bar{z}>^{-2N_2}<\bar{\eta}>^{-2N_3}<\bar{\zeta}>^{-2N_4}<\xi>^m<\xi+\bar{\eta}>^{m'}<\xi+\bar{\eta}+\bar{\zeta}>^{m''}\,d\bar{y}\,d\bar{z}\,\dbar\bar{\eta}\,\dbar\bar{\zeta},
$$
with $N_j$ for $1\leq j\leq4$ natural numbers that we can choose so that: $2N_1> d$, $2N_2> d$, $2N_3> d+|m'|+|m''|$ and $2N_4> d+|m''|$ so that we obtain the estimation
\begin{equation}\label{4.16}
 \left|c(x,\xi)\right|\leq C_f<\xi>^{m+m'+m''},\qquad\forall(x,\xi)\in\mathbb{T}^*\mathbb{R}^d.
\end{equation}

Let us estimate now the first order derivatives of the symbol $c$. We have that
\begin{equation}\label{4.17}
 \big(\nabla_x c\big)(x,\xi)=\int_{\mathbb{R}^{4d}}e^{i\bar{\psi}}\bar{\Omega^B}\left[\left(i\nabla_x \bar{d}-i\nabla_x\bar{F}\right)\bar{f}+\nabla_x\bar{f}\right]\,d\bar{y}\,d\bar{z}\,\dbar\bar{\eta}\,\dbar\bar{\zeta}.
\end{equation}
The integrals containing factors $\nabla_x\bar{F}$ and $\nabla_x\bar{f}$ are treated as before and lead to estimations of the form \eqref{4.16}. For the integral containing the factor $\nabla_x\bar{d}$ one has to take into account that
$$
\nabla_x\bar{d}(x,\xi,\bar{\eta},\bar{\zeta})=\left[\int_0^1\big(\nabla^2_{x,\xi}d\big)\big(x,\xi+t(\bar{\eta}+\bar{\zeta})\big)dt\right](\bar{\eta}+\bar{\zeta}),
$$
and $\bar{\eta}$ and $\bar{\zeta}$ may be eliminated by integration by parts. Moreover, the entries of the matrix in paranthesis are function of class $BC^\infty(\mathbb{R}^{4d})$. It follows that
\begin{equation}\label{4.18}
 \left|\nabla_x c(x,\xi)\right|\leq
C_f<\xi>^{m+m'+m''},\qquad\forall(x,\xi)\in\mathbb{T}^*\mathbb{R}^d.
\end{equation}

Similarly we have that
\begin{equation}\label{4.19}
 \big(\nabla_\xi c\big)(x,\xi)=\int_{\mathbb{R}^{4d}}e^{i\bar{\psi}}\bar{\Omega^B}\left[\left(i\nabla_\xi \bar{d}\right)\bar{f}+\nabla_\xi\bar{f}\right]\,d\bar{y}\,d\bar{z}\,\dbar\bar{\eta}\,\dbar\bar{\zeta},
\end{equation}
$$
\nabla_\xi\bar{d}(x,\xi,\bar{\eta},\bar{\zeta})=\left[\int_0^1\big(\nabla^2_{\xi,\xi}d\big)\big(x,\xi+t(\bar{\eta}+\bar{\zeta})\big)dt\right](\bar{\eta}+\bar{\zeta}).
$$
We once again eliminate $\bar{\eta}$ and $\bar{\zeta}$ by integration by parts, use \eqref{4.15} for some more integration by parts in order to get
$$
 \left|\nabla_\xi c(x,\xi)\right|\leq
$$
$$
\leq C_f\int_{\mathbb{R}^{4d}}<\bar{y}>^{-2N_1}<\bar{z}>^{-2N_2}<\bar{\eta}>^{-2N_3}<\bar{\zeta}>^{-2N_4}<\xi>^m<\xi+\bar{\eta}>^{m'}<\xi+\bar{\eta}+\bar{\zeta}>^{m''}\times
$$
$$
\times\left[<\xi>^{-1}+<\xi+\bar{\eta}>^{-1}+<\xi+\bar{\eta}+\bar{\zeta}>^{-1}+\underset{0\leq t\leq1}{\sup}<\xi+t(\bar{\eta}+\bar{\zeta})>^{-1}\right]\,d\bar{y}\,d\bar{z}\,\dbar\bar{\eta}\,\dbar\bar{\zeta}.
$$
We choose $2N_1> d$, $2N_2> d$, $2N_3> d+|m'|+|m''|+1$ and $2N_4> d+|m''|+1$ so that we obtain the estimation
\begin{equation}\label{4.20}
 \left|\nabla_\xi c(x,\xi)\right|\,\leq\,C_f<\xi>^{m+m'+m''-1},\qquad\forall(x,\xi)\in\mathbb{T}^*\mathbb{R}^d.
\end{equation}

The statement of the Proposition follows after similar estimations for any higher order derivative of $c$.
\end{proof}

In order to compute the principal part of the symbol $c$ we consider first the case $B=0$, i.e. $\Omega^B=1$.
\begin{lemma}\label{L4.3}
 Let us denote by
$$
e(x,\xi):=\int_{\mathbb{R}^{4d}}e^{i\bar{\psi}}\bar{f}\,d\bar{y}\,d\bar{z}\,\dbar\bar{\eta}\,\dbar\bar{\zeta}.
$$
Then $c-e\,\in\,S^{m+m'+m''-1}(\mathbb{R}^d)$.
\end{lemma}
\begin{proof}
 We notice that
$$
\bar{\Omega^B}=e^{-i\bar{F}}=1\,-\,i\bar{F}\int_0^1e^{-it\bar{F}}dt,
$$
so that we can write
\begin{equation}\label{4.23}
 c=e\,-\,i\int_0^1r_t\,dt
\end{equation}
with
\begin{equation}\label{4.24}
 r_t(x,\xi)\,:=\,\int_{\mathbb{R}^{4d}}e^{i(<\bar{z},\bar{\zeta}>+<\bar{y},\bar{\eta}>)}e^{-it\bar{F}}\bar{F}e^{i\bar{d}}\bar{f}\,d\bar{y}\,d\bar{z}\,\dbar\bar{\eta}\,\dbar\bar{\zeta}.
\end{equation}
If we use now formulae \eqref{1.12} and \eqref{1.13} we obtain that
\begin{equation}\label{4.26}
 \bar{F}(x,\bar{y},\bar{z})\,=\,\left<\bar{C}(x,\bar{y},\bar{z})\bar{z},\bar{y}\right>,\quad\bar{C}\in BC^\infty\big(\mathbb{R}^{3d}\big),\text{ skewsymmetric}.
\end{equation}
In the integral appearing in \eqref{4.24} we make the change of variables:
\begin{equation}\label{4.27}
\tilde{y}:=(1/2)(\bar{y}+\bar{z}),\quad\tilde{z}:=(1/2)(\bar{y}-\bar{z}),\quad\tilde{\eta}:=\bar{\eta}+\bar{\zeta},\quad\tilde{\zeta}:=\bar{\eta}-\bar{\zeta}.
\end{equation}
We also introduce some more notations
\begin{equation}\label{4.28}
 \tilde{f}(x,\tilde{y},\tilde{z},\xi,\tilde{\eta},\tilde{\zeta}):=\bar{f}\big(x,\tilde{y}+\tilde{z},\tilde{y}-\tilde{z},\xi,(1/2)(\tilde{\eta}+\tilde{\zeta}),(1/2)(\tilde{\eta}-\tilde{\zeta})\big),
\end{equation}
$$
\tilde{d}(x,\xi,\tilde{\eta}):=d(x,\xi+\tilde{\eta})-d(x,\xi),
$$
$$
\tilde{C}(x,\tilde{y},\tilde{z}):=\bar{C}(x,\tilde{y}+\tilde{z},\tilde{y}-\tilde{z}),
$$
$$
\tilde{F}(x,\tilde{y},\tilde{z}):=\bar{F}(x,\tilde{y}+\tilde{z},\tilde{y}-\tilde{z})=\left<\tilde{C}(x,\tilde{y},\tilde{z})(\tilde{y}-\tilde{z}),\tilde{y}+\tilde{z}\right>=2\left<\tilde{C}(x,\tilde{y},\tilde{z})\tilde{y},\tilde{z}\right>,
$$
where in the last equality we have used the skewsymmetry of $\tilde{C}$. It follows that
\begin{equation}\label{4.29}
 r_t(x,\xi)\,:=\,2\int_{\mathbb{R}^{4d}}e^{i(<\tilde{z},\tilde{\zeta}>+<\tilde{y},\tilde{\eta}>)}e^{-it\bar{F}}\left<\tilde{C}(x,\tilde{y},\tilde{z})\tilde{y},\tilde{z}\right>e^{i\tilde{d}}\tilde{f}\,d\tilde{y}\,d\tilde{z}\,\dbar\tilde{\eta}\,\dbar\tilde{\zeta}.
\end{equation}
We take into account that $\tilde{F}$ does not depend on $(\tilde{\eta},\tilde{\zeta})$ and $\tilde{d}$ does not depend on $\tilde{\zeta}$ and get rid of the linear factors $\tilde{y}$ and $\tilde{z}$ by integration by parts. Then
\begin{equation}\label{4.30}
 r_t(x,\xi)\,:=\,-2\int_{\mathbb{R}^{4d}}e^{i(<\tilde{z},\tilde{\zeta}>+<\tilde{y},\tilde{\eta}>)}e^{-it\bar{F}}\left<\tilde{C}(x,\tilde{y},\tilde{z})\nabla_{\tilde{\eta}},e^{i\tilde{d}}\nabla_{\tilde{\zeta}}\tilde{f}\right>\,d\tilde{y}\,d\tilde{z}\,\dbar\tilde{\eta}\,\dbar\tilde{\zeta}.
\end{equation}
 Application of $\nabla_{\tilde{\zeta}}$ to $\tilde{f}$ lowers the sum $m+m'+m''$ by one unity. Moreover, $\nabla_{\tilde{\eta}}\tilde{d}$ is a function of class $BC^\infty(\mathbb{R}^{3d})$. Thus, using the proof of Proposition \ref{P4.2} we conclude that $r_t\in S^{m+m'+m''-1}(\mathbb{R}^d)$ uniformly with respect to $t\in[0,1]$. In conclusion we have proved that $c-e\in  S^{m+m'+m''-1}(\mathbb{R}^d)$.
\end{proof}
\begin{lemma}\label{L4.4}
 Let us denote by
$$
e_0(x,\xi):=\int_{\mathbb{R}^{2d}}e^{i<\tilde{y},\tilde{\eta}>+\tilde{d}(x,\xi,\tilde{\eta}))}\tilde{f}(x,\tilde{y},0,\xi,\tilde{\eta},0)\,d\tilde{y}\,\dbar\tilde{\eta}.
$$
Then $e-e_0\,\in\,S^{m+m'+m''-1}(\mathbb{R}^d)$.
\end{lemma}
\begin{proof}
 We start from the equality:
$$
\tilde{f}(x,\tilde{y},\tilde{z},\xi,\tilde{\eta},\tilde{\zeta})=\tilde{f}(x,\tilde{y},0,\xi,\tilde{\eta},\tilde{\zeta})+\left<\tilde{z},\int_0^1\big(\nabla_{\tilde{z}}\tilde{f}\big)(x,\tilde{y},t\tilde{z},\xi,\tilde{\eta},\tilde{\zeta})\,dt
\right>
$$
and the corresponding decomposition
\begin{equation}\label{4.31}
 e\,=\,e_0\,+\,i\int_0^1\tilde{r}_tdt,
\end{equation}
\begin{equation}\label{4.32}
 \tilde{r}_t(x,\xi)=\int_{\mathbb{R}^{4d}}e^{i(<\tilde{z},\tilde{\zeta}>+<\tilde{y},\tilde{\eta}>)}e^{i\tilde{d}(x,\xi,\tilde{\eta})}\left(\left<\nabla_{\tilde{\zeta}},\nabla_{\tilde{z}}\right>\tilde{f}\right)(x,\tilde{y},t\tilde{z},\xi,\tilde{\eta},\tilde{\zeta})\,d\tilde{y}\,d\tilde{z}\,\dbar\tilde{\eta}\,\dbar\tilde{\zeta}.
\end{equation}
Using once again the arguments in the proof of Proposition \ref{P4.2} we deduce that $\tilde{r}_t\in S^{m+m'+m''-1}(\mathbb{R}^d)$ uniformly with respect to $t\in[0,1]$ and in conclusion we have proved that $e-e_0\in  S^{m+m'+m''-1}(\mathbb{R}^d)$.
\end{proof}

\begin{proposition}\label{P4.5}
 Assume $f\in S^{m,m',m''}\big(\mathbb{R}^{3d}\times\mathbb{R}^d\times\mathbb{R}^d\times\mathbb{R}^d\big)$ and the symbol $c$ is defined by \eqref{4.8}, then
$$
c(x,\xi)\,-\,f\big(x,\nabla_\xi U(x,\xi),\nabla_\xi U(x,\xi),\xi,\xi,\xi\big)\,\in\,S^{m+m'+m''-1}(\mathbb{R}^d).
$$
\end{proposition}
\begin{proof}
 We start from the definition of $\tilde{d}$ given in \eqref{4.28} and write:
\begin{equation}\label{4.36}
 \tilde{d}(x,\xi,\tilde{\eta})=\left<\tilde{\eta},h(x,\xi,\tilde{\eta})\right>,\qquad h(x,\xi,\tilde{\eta}):=\int_0^1\big(\nabla_\xi d\big)(x,\xi+t\tilde{\eta})dt.
\end{equation}
In the integral defining $e_0$ in the statement of Lemma \ref{L4.4} we make the change of variables
\begin{equation}\label{4.37}
 \tilde{\bar{y}}:=\tilde{y}+h(x,\xi,\tilde{\eta})
\end{equation}
in order to obtain
\begin{equation}\label{4.38}
 e_0(x,\xi):=\int_{\mathbb{R}^{2d}}e^{i<\tilde{\bar{y}},\tilde{\eta}>}\tilde{f}(x,\tilde{\bar{y}}-h(x,\xi,\tilde{\eta}),0,\xi,\tilde{\eta},0)\,d\tilde{\bar{y}}\,\dbar\tilde{\eta}=
\end{equation}
\begin{equation}\label{4.39}
=\tilde{f}(x,-h(x,\xi,0),0,\xi,0,0)\,+\,i\int_{\mathbb{R}^{2d}}e^{i<\tilde{\bar{y}},\tilde{\eta}>}\left(\int_0^1\left<\nabla_{\tilde{\eta}},\big(\nabla_{\tilde{y}}\tilde{f}\big)(x,s\tilde{\bar{y}}-h(x,\xi,\tilde{\eta}),0,\xi,\tilde{\eta},0)\right>\,ds\right)d\tilde{\bar{y}}\,\dbar\tilde{\eta}.
\end{equation}
But we can write
$$
\tilde{f}(x,\tilde{y},0,\xi,\tilde{\eta},0)=\bar{f}\big(x,\tilde{y},\tilde{y},\xi,(1/2)\tilde{\eta},(1/2)\tilde{\eta}\big)=f\big(x,x-\tilde{y},x-\tilde{y},\xi,\xi+(1/2)\tilde{\eta},\xi+\tilde{\eta}\big).
$$
We integrate now by parts using the operators $<\tilde{\bar{y}}>^{-2N_1}(1-\Delta_{\tilde{\eta}})^{N_1}$ and $<\tilde{\eta}>^{-2N_2}(1-\Delta_{\tilde{\bar{y}}})^{N_2}$ and obtain that the last integral in \eqref{4.39} is bounded from above by
$$
C_f\int_{\mathbb{R}^{2d}}<\tilde{\bar{y}}>^{-2N_1}<\tilde{\eta}>^{-2N_2}<\xi>^m<\xi+(1/2)\tilde{\eta}>^{m'}<\xi+\tilde{\eta}>^{m''}\underset{0\leq t\leq1}{\sup}<\xi+t\tilde{\eta}>^{-1}d\tilde{\bar{y}}\,\dbar\tilde{\eta}
$$
$$
\leq C_f<\xi>^{m+m'+m''-1},\quad\forall(x,\xi)\in\mathbb{T}^*\mathbb{R}^d,
$$
by choosing $2N_1> d$ and $2N_2>d+|m'|+|m''|+1$. Thus
\begin{equation}\label{4.41}
 e_0(x,\xi)\,-\,\tilde{f}(x,-h(x,\xi,0),0,\xi,0,0)
\end{equation}
defines a symbol of class $S^{m+m'+m''-1}(\mathbb{R}^d)$. Further we notice that $h(x,\xi,0)=\big(\nabla_\xi d\big)(x,\xi)$ so that we can write
\begin{equation}\label{4.42}
 \tilde{f}(x,-h(x,\xi,0),0,\xi,0,0)=f\big(x,\nabla_\xi U(x,\xi),\nabla_\xi U(x,\xi),\xi,\xi,\xi\big).
\end{equation}
The proof of the Proposition follows now by using the equalities \eqref{4.39}, \eqref{4.41} and \eqref{4.42} and by applying Lemmas \ref{L4.3} and \ref{L4.4}.
\end{proof}

We come back now to the situation in Proposition \ref{P4.1} with $f(x,y,z,\xi,\eta,\zeta):=a(x,\zeta)b(z,\eta)$ and obtain as a Corollary of Proposition \ref{P4.5} the main result of this section.

\begin{theorem}\label{T4.6}
 Suppose that $a\in S^{m''}(\mathbb{R}^d)$ and $b\in S^{m'}(\mathbb{R}^d)$, then $\mathfrak{Op}^A_\Phi(a)\circ b^A(x,D)=\mathfrak{Op}^A_\Phi(c)$ with $c\in S^{m'+m''}(\mathbb{R}^d)$. Moreover we have that
\begin{equation}\label{4.43}
 c(x,\xi)\,-\,a(x,\xi)b\big(\nabla_\xi U(x,\xi),\xi\big)\,\in\,S^{m'+m''-1}(\mathbb{R}^d).
\end{equation}
\end{theorem}

\section{Composing a FIO with the adjoint of a FIO}

Suppose that $a\in S^{m''}(\mathbb{R}^d)$ and $b\in S^{m'}(\mathbb{R}^d)$, and let us define $E:=\mathfrak{Op}^A_\Phi(a)$ and $F:=\big[\mathfrak{Op}^A_\Phi(b)\big]^*$. Explicitely we have that for any $u\in\mathcal{S}(\mathbb{R}^d)$ and any point $x\in\mathbb{R}^d$
\begin{equation}\label{5.1}
 \big(Eu\big)(x)\,=\,\int_{\mathbb{R}^{2d}}e^{i\big(U(x,\zeta)-<z,\zeta>\big)}\omega^A(x,z)a(x,\zeta)u(z)\,dz\,\dbar\zeta,
\end{equation}
\begin{equation}\label{5.2}
 \big(Fu\big)(x)\,=\,\int_{\mathbb{R}^{2d}}e^{i\big(<x,\eta>-U(y,\eta)\big)}\omega^A(x,y)\overline{b(y,\eta)}u(y)\,dy\,\dbar\eta.
\end{equation}

We know that $E\circ F$ is an operator in $\mathbb{B}\big(\mathcal{S}(\mathbb{R}^d)\big)$; we shall prove in this section that it is in fact a $\Psi$DO and we shall compute a principal symbol for it.

We proceed as in the previous sections and fix a cut-off function $\chi\in C^\infty_0(\mathbb{R}^{2d})$ with $\chi(0,0)=1$ and for $\epsilon\in(0,1]$ we define $a_\epsilon(x,z,\zeta):=\chi(\epsilon z,\epsilon\zeta)a(x,\zeta)$ and $b_\epsilon(y,\eta):=\chi(\epsilon y,\epsilon\eta)\overline{b(y,\eta)}$. Then we denote by $E_\epsilon$ and $F_\epsilon$ the operators defined respectively by \eqref{5.1} with $a$ replaced by $a_\epsilon$ and by equation \eqref{5.2} with $b$ replaced by $b_\epsilon$. The distribution kernel $K_\epsilon\in\mathcal{S}(\mathbb{R}^{2d})$ of the product $E_\epsilon\circ F_\epsilon$ is given by the following integral
\begin{equation}\label{5.3}
 K_\epsilon(x,y)\,:=\,\omega^A(x,y)\int_{\mathbb{R}^{3d}}e^{i\big(<z,\eta-\zeta>+U(x,\zeta)-U(y,\eta)\big)}\Omega^B(x,z,y)a_\epsilon(x,z,\zeta)b_\epsilon(y,\eta)\,dz\,\dbar\eta\,\dbar\zeta.
\end{equation}
We use now the equation relating the distribution kernel of a $\Psi$DO to its symbol (the special case of \eqref{2.5} when $\Phi=1$, i.e. $U(x,\eta)=<x,\eta>$) and obtain that
\begin{equation}\label{5.4}
 E_\epsilon\circ F_\epsilon\,=\,c^A_\epsilon(x,D),
\end{equation}
with $c_\epsilon\in\mathcal{S}(\mathbb{R}^d)$ given by
\begin{equation}\label{5.5}
 c_\epsilon(x,\xi)=\int_{\mathbb{R}^{4d}}e^{i\psi(x,y,z,\xi,\eta,\zeta)}\Omega^B(x,z,y)a_\epsilon(x,z,\zeta)b_\epsilon(y,\eta)\,dy\,dz\,\dbar\eta\,\dbar\zeta,
\end{equation}
\begin{equation}\label{5.6}
\psi(x,y,z,\xi,\eta,\zeta):=<z,\eta-\zeta>-<x-y,\xi>+U(x,\zeta)-U(y,\eta).
\end{equation}

Once again we notice that once we proved the existence of the following limit
\begin{equation}\label{5.7}
 \underset{\epsilon\searrow0}{\lim}c_\epsilon\,=\,c,\text{ in }S^{m'+m''}(\mathbb{R}^d),
\end{equation}
taking into account that for any $v\in\mathcal{S}(\mathbb{R}^d)$ we have $\underset{\epsilon\searrow0}{\lim}\big(E_\epsilon\circ F_\epsilon\big)v=\big(E\circ F\big)v$, we have that the following statement is true.
\begin{proposition}\label{P5.1}
 If $a\in S^{m''}(\mathbb{R}^d)$ and $b\in S^{m'}(\mathbb{R}^d)$, then $\mathfrak{Op}^A_\Phi(a)\circ\left[\mathfrak{Op}^A_\Phi(b)\right]^*=c^A(x,D)$ with $c\in S^{m'+m''}(\mathbb{R}^d)$ defined by \eqref{5.7}.
\end{proposition}

Now, the existence of the limit \eqref{5.7} follows from the next Proposition.
\begin{proposition}\label{P5.2}
 Suppose $f\in S^{m,m',m''}\big(\mathbb{R}^{3d}\times\mathbb{R}^d\times\mathbb{R}^d\times\mathbb{R}^d\big)$ and $\psi$ is the function defined in \eqref{5.6}. Then the following oscillating integral
\begin{equation}\label{5.8}
 c(x,\xi)\,:=\,\int_{\mathbb{R}^{4d}}e^{i\psi(x,y,z,\xi,\eta,\zeta)}\Omega^B(x,z,y)f(x,y,z,\xi,\eta,\zeta)\,dy\,dz\,\dbar\eta\,\dbar\zeta,
\end{equation}
defines an element in $S^{m+m'+m''}(\mathbb{R}^d)$ and the map
$$
S^{m,m',m''}\big(\mathbb{R}^{3d}\times\mathbb{R}^d\times\mathbb{R}^d\times\mathbb{R}^d\big)\ni f\mapsto c\in S^{m+m'+m''}(\mathbb{R}^d)
$$
is continuous.
\end{proposition}
\begin{proof}
 As before, we can work with $f$ with compact support wih respect to the variables $(y,z,\eta,\zeta)$. We remark that
\begin{equation}\label{5.9}
\psi(x,y,z,\xi,\eta,\zeta)=<x-y,\eta-\xi>+<x-z,\zeta-\eta>+d(x,\zeta)-d(y,\eta)
\end{equation}
and make the change of variables
\begin{equation}\label{5.10}
 \bar{y}:=x-y,\quad\bar{z}:=x-z,\quad\bar{\eta}:=\xi-\eta,\quad\bar{\zeta}:=\zeta-\eta
\end{equation}
and the notations:
\begin{equation}\label{5.11}
 \bar{\psi}\equiv\bar{\psi}(x,\bar{y},\bar{z},\xi,\bar{\eta},\bar{\zeta}):=<\bar{z},\bar{\zeta}>-<\bar{y},\bar{\eta}>+d(x,\xi-\bar{\eta}+\bar{\zeta})-d(x-\bar{y},\xi-\bar{\eta}),
\end{equation}
\begin{equation}\label{5.12}
 \overline{\Omega^B}\equiv\overline{\Omega^B}(x,\bar{y},\bar{z}):=\Omega^B(x,x-\bar{z},x-\bar{y}),
\end{equation}
\begin{equation}\label{5.13}
 \bar{f}\equiv\bar{f}(x,\bar{y},\bar{z},\xi,\bar{\eta},\bar{\zeta}):=f(x,x-\bar{y},x-\bar{z},\xi,\xi-\bar{\eta},\xi-\bar{\eta}+\bar{\zeta}),
\end{equation}
so that we can write
\begin{equation}\label{5.14}
 c(x,\xi)\,:=\,\int_{\mathbb{R}^{4d}}e^{i\bar{\psi}}\overline{\Omega^B}\bar{f}\,d\bar{y}\,d\bar{z}\,\dbar\bar{\eta}\,\dbar\bar{\zeta}.
\end{equation}

Let us consider the identity
$$
d(x,\xi-\bar{\eta}+\bar{\zeta})-d(x-\bar{y},\xi-\bar{\eta})=d(x,\xi-\bar{\eta}+\bar{\zeta})-d(x,\xi-\bar{\eta})+d(x,\xi-\bar{\eta})-d(x-\bar{y},\xi-\bar{\eta})=
$$
$$
=d(x,\xi-\bar{\eta}+\bar{\zeta})-d(x,\xi-\bar{\eta})+\left<\bar{y},\int_0^1\big(\nabla_x d\big)(x-t\bar{y},\xi-\bar{\eta})dt\right>
$$
and make the following change of variables
\begin{equation}\label{5.15}
\tilde{\eta}:=-\bar{\eta}+\int_0^1\big(\nabla_x d\big)(x-t\bar{y},\xi-\bar{\eta})dt,
\end{equation}
using also Lemma \ref{L1.5} and Hypothesis \ref{Hyp-U}. We still denote by
$$
\Lambda(x,\bar{y},\eta):=\eta+\int_0^1\big(\nabla_x d\big)(x-t\bar{y},\eta)dt
$$
and let $\lambda(x,\bar{y},\theta)$ be the inverse of the diffeomorphism
$$
\mathbb{R}^d\ni\eta\mapsto\theta:=\Lambda(x,\bar{y},\eta)\in\mathbb{R}^d.
$$
Due to our Hypothesis \ref{Hyp-U} and to Lemma \ref{L1.5}, $<\eta>$ is equivalent with $<\Lambda(x,\bar{y},\eta)>$ and $<\theta>$ is equivalent with $<\lambda(x,\bar{y},\theta)>$, uniformly with respect to $(x,\bar{y})\in\mathbb{R}^{2d}$. Moreover, we have that $\lambda$ belongs to $S^1\big(\mathbb{R}^{2d}\times\mathbb{R}^d\big)$ and the functional determinant $D(x,\bar{y},\theta)$ of the map
$$
\mathbb{R}^d\ni\theta\mapsto\lambda(x,\bar{y},\theta)\in\mathbb{R}^d
$$
is a symbol of class $S^0\big(\mathbb{R}^{2d}\times\mathbb{R}^d\big)$. Let us write \eqref{5.15} in the form
$$
\tilde{\eta}+\xi\,=\,\Lambda(x,\bar{y},\xi-\bar{\eta})
$$
and deduce that
\begin{equation}\label{5.16}
 \bar{\eta}\,=\,\xi-\lambda(x,\bar{y},\xi+\tilde{\eta}).
\end{equation}
Let us introduce the following notations:
\begin{equation}\label{5.17}
 \tilde{d}\equiv\tilde{d}(x,\bar{y},\xi,\tilde{\eta},\bar{\zeta}):=d\big(x,\lambda(x,\bar{y},\xi+\tilde{\eta})+\bar{\zeta}\big)-d\big(x,\lambda(x,\bar{y},\xi+\tilde{\eta})\big),
\end{equation}
\begin{equation}\label{5.18}
 \tilde{f}\equiv\tilde{f}(x,\bar{y},\bar{z},\xi,\tilde{\eta},\bar{\zeta}):=\bar{f}\big(x,\bar{y},\bar{z},\xi,\xi-\lambda(x,\bar{y},\xi+\tilde{\eta}),\bar{\zeta}\big),
\end{equation}
so that we can write
\begin{equation}\label{5.19}
 c(x,\xi)\,:=\,\int_{\mathbb{R}^{4d}}e^{i(<\bar{z},\bar{\zeta}>+<\bar{y},\tilde{\eta}>+\tilde{d})}\overline{\Omega^B}\tilde{f}D(x,\bar{y},\xi+\tilde{\eta})\,d\bar{y}\,d\bar{z}\,\dbar\tilde{\eta}\,\dbar\bar{\zeta}.
\end{equation}

We shall need to estimate the derivatives of $e^{i\tilde{d}}$ and $\tilde{f}$. It is evident that for any $(\alpha,\beta)\in\mathbb{N}^{2d}\times\mathbb{N}^{3d}$ we have that
\begin{equation}\label{5.20}
 \left|\partial^\alpha_{(x,\bar{y})}\partial^\beta_{(\xi,\tilde{\eta},\bar{\zeta})}e^{i\tilde{d}}\right|\leq C\left(<\bar{\zeta}>+<\xi+\tilde{\eta}>\right)^{|\alpha|},\quad\forall(x,\bar{y},\xi,\tilde{\eta},\bar{\zeta})\in\mathbb{R}^{5d}.
\end{equation}
Moreover, on the domain $|\bar{\zeta}|\leq\epsilon<\xi+\tilde{\eta}>$, for $\epsilon>0$ small enough, this inequality may still be improved. In fact we use the identity
$$
\tilde{d}(x,\bar{y},\xi,\tilde{\eta},\bar{\zeta})=\left<\bar{\zeta},\int_0^1\big(\nabla_\xi d\big)\big(x,\lambda(x,\bar{y},\xi+\tilde{\eta})+t\bar{\zeta}\big)dt\right>
$$
and the fact that on the above domain $<\lambda(x,\bar{y},\xi+\tilde{\eta})+t\bar{\zeta}>$ is equivalent with $<\xi+\tilde{\eta}>$, uniformly with respect to the parameter $t\in[0,1]$. We obtain that for $(\alpha,\beta)\in\mathbb{N}^{2d}\times\mathbb{N}^{3d}$ we have that
\begin{equation}\label{5.21}
 \left|\partial^\alpha_{(x,\bar{y})}\partial^\beta_{(\xi,\tilde{\eta},\bar{\zeta})}e^{i\tilde{d}}\right|\leq C<\bar{\zeta}>^{|\alpha|},\quad\text{if }|\bar{\zeta}|\leq\epsilon<\xi+\tilde{\eta}>.
\end{equation}
Let us consider now
$$
\tilde{f}(x,\bar{y},\bar{z},\xi,\tilde{\eta},\bar{\zeta})=f\big(x,x-\bar{y},x-\bar{z},\xi,\lambda(x,\bar{y},\xi+\tilde{\eta}),\lambda(x,\bar{y},\xi+\tilde{\eta})+\bar{\zeta}\big).
$$
For any $(\alpha,\beta)\in\mathbb{N}^{2d}\times\mathbb{N}^{4d}$ we have that
\begin{equation}\label{5.22}
 \left|\partial^\alpha_{(x,\bar{y})}\partial^\beta_{(\bar{z},\xi,\tilde{\eta},\bar{\zeta})}\tilde{f}\right|\leq C_f<\xi>^m<\xi+\tilde{\eta}>^{m'}<\lambda(x,\bar{y},\xi+\tilde{\eta})+\bar{\zeta}>^{m''}\left(1+\frac{<\xi+\tilde{\eta}>}{<\lambda(x,\bar{y},\xi+\tilde{\eta})+\bar{\zeta}>}\right)^{|\alpha|}.
\end{equation}
This inequality may also be improved on the domein $|\bar{\zeta}|\leq\epsilon<\xi+\tilde{\eta}>$, where we have for any  $(\alpha,\beta)\in\mathbb{N}^{3d}\times\mathbb{N}^{3d}$
\begin{equation}\label{5.23}
 \left|\partial^\alpha_{(x,\bar{y},\bar{z})}\partial^\beta_{(\xi,\tilde{\eta},\bar{\zeta})}\tilde{f}\right|\leq C_f<\xi>^m<\xi+\tilde{\eta}>^{m'+m''}\left(<\xi>^{-|\beta|}+<\xi+\tilde{\eta}>^{-|\beta|}\right),\quad\text{for }|\bar{\zeta}|\leq\epsilon<\xi+\tilde{\eta}>.
\end{equation}
We integrate by parts in \eqref{5.19} using the identities \eqref{4.15} and getting rid of the terms appearing after the derivation of $\overline{\Omega^B}$. Then we decompose the integration domain as the union of the regions $|\bar{\zeta}|\leq\epsilon<\xi+\tilde{\eta}>$ and respectively $|\bar{\zeta}|>\epsilon<\xi+\tilde{\eta}>$, writing $c=c_0+c_\infty$. For these two terms we obtain estimations of the form:
$$
\left|c_0(x,\xi)\right|\leq
$$
$$
\leq C_f\hspace{-0.5cm}\int\limits_{|\bar{\zeta}|\leq\epsilon<\xi+\tilde{\eta}>}\hspace{-0.5cm}<\bar{y}>^{-2N_1+2N_3}<\bar{z}>^{-2N_2+2N_4}<\bar{\zeta}>^{-2N_4+2N_3}<\tilde{\eta}>^{-2N_3}<\xi>^m<\xi+\tilde{\eta}>^{m'+m''}\,d\bar{y}\,d\bar{z}\,\dbar\tilde{\eta}\,\dbar\bar{\zeta}\leq
$$
$$
\leq C_f<\xi>^{m+m'+m''},
$$
for $2N_1\geq2N_3+d+1,\ 2N_2\geq2N_4+d+1,\ 2N_3\geq|m'+m''|+d+1,\ 2N_4\geq2N_3+d+1$;
$$
\left|c_\infty(x,\xi)\right|\leq
$$
$$
\leq C_f\hspace{-0.5cm}\int\limits_{|\bar{\zeta}|>\epsilon<\xi+\tilde{\eta}>}\hspace{-0.5cm}<\bar{y}>^{-2N_1+2N_3}<\bar{z}>^{-2N_2+2N_4}<\bar{\zeta}>^{-2N_4}<\tilde{\eta}>^{-2N_3}\times
$$
$$
\times<\xi>^m<\xi+\tilde{\eta}>^{m'}<\lambda(x,\bar{y},\xi+\tilde{\eta})+\bar{\zeta}>^{m''}\left(<\bar{\zeta}>+<\xi+\tilde{\eta}>\right)^{2N_3}\,d\bar{y}\,d\bar{z}\,\dbar\tilde{\eta}\,\dbar\bar{\zeta}\leq
$$
$$
\leq C_f\int_{\mathbb{R}^{2d}}<\xi>^{m-p}<\tilde{\eta}>^{-2N_3+p}<\bar{\zeta}>^{-2N_4+2N_3+p+m'_++m''_+}\dbar\tilde{\eta}\,\dbar\bar{\zeta}\ \leq\ C_f<\xi>^{m-p},\ \forall p\in\mathbb{N},
$$
for $2N_1\geq2N_3+d+1$, $2N_2\geq2N_4+d+1$, $2N_3\geq p+d+1$, $2N_4\geq2N_3+p+d+1+m'_++m''_+$ (here $m_+:=\max\{m,0\}$).
For $p\geq|m'+m''|$ the estimations for $c_0$ and $c_\infty$ are satisfied simultaneously and we conclude that
\begin{equation}\label{5.24}
 \left|c(x,\xi)\right|\,\leq\,C_f<\xi>^{m+m'+m''},\quad\forall(x,\xi)\in\mathbb{T}^*\mathbb{R}^d.
\end{equation}
Estimating similarly the derivatives of $c$ we obtain the statement of the Proposition.
\end{proof}

We shall now compute a principal part of the symbol $c$. The first step in this direction is the reduction to the case $B=0$.
\begin{lemma}\label{L5.3}
 If we define
$$
e(x,\xi):=\int_{\mathbb{R}^{4d}}e^{i(<\bar{z},\bar{\zeta}>+<\bar{y},\tilde{\eta}>+\tilde{d})}\tilde{f}D(x,\bar{y},\xi+\tilde{\eta})\,d\bar{y}\,d\bar{z}\,\dbar\tilde{\eta}\,\dbar\bar{\zeta},
$$
we have that $c-e\,\in\,S^{m+m'+m''-1}(\mathbb{R}^d)$.
\end{lemma}
\begin{proof}
 As in the proof of Lemma \ref{L4.3} we start from the equality
\begin{equation}\label{5.27}
 \bar{\Omega^B}=e^{-i\bar{F}}=1-i\bar{F}\int_0^1e^{-it\bar{F}}dt,\quad\bar{F}(x,\bar{y},\bar{z})=\left<\bar{C}(x,\bar{y},\bar{z})\bar{z},\bar{y}\right>
\end{equation}
with $\bar{C}\in BC^\infty(\mathbb{R}^{3d})$ as in \eqref{3.139}. Using this decomposition and integrating by parts to get rid of the factors $\bar{y}$ and $\bar{z}$ we get
\begin{equation}\label{5.29}
c=e+i\int_0^1r_t(x,\xi)dt,
\end{equation}
\begin{equation}\label{5.30}
 r_t(x,\xi)=\int_{\mathbb{R}^{4d}}e^{i(<\bar{z},\bar{\zeta}>+<\bar{y},\tilde{\eta}>)}e^{-it\bar{F}}\left<\bar{C}\nabla_{\bar{\zeta}},\nabla_{\tilde{\eta}}\right>\big(e^{i\tilde{d}}\tilde{f}D\big)\,d\bar{y}\,d\bar{z}\,\dbar\tilde{\eta}\,\dbar\bar{\zeta}.
\end{equation}
We end the proof with the same arguments as in the proof of Proposition \ref{5.2} obtaining that $r_t\in S^{m+m'+m''-1}(\mathbb{R}^d)$ uniformly with respect with $t\in[0,1]$ and concluding that $c-e\in S^{m+m'+m''-1}(\mathbb{R}^d)$.
\end{proof}

We reduce now further to an integral on $\mathbb{R}^{2d}$.
\begin{lemma}\label{L5.4}
 If we define
$$
e_0(x,\xi):=\int_{\mathbb{R}^{2d}}e^{i<\bar{y},\tilde{\eta}>}\tilde{f}(x,\bar{y},-\big(\nabla_\xi d\big)\big(x,\lambda(x,\bar{y},\xi+\tilde{\eta}),\xi,\tilde{\eta},0\big)D(x,\bar{y},\xi+\tilde{\eta})\,d\bar{y}\,\dbar\tilde{\eta},
$$
then $e-e_0\in S^{m+m'+m''-1}(\mathbb{R}^d)$.
\end{lemma}
\begin{proof}
 Starting from \eqref{5.17} we can write
\begin{equation}\label{5.34}
 \tilde{d}(x,\bar{y},\xi,\tilde{\eta},\bar{\zeta})=\left<\bar{\zeta},h\right>
\end{equation}
with
\begin{equation}\label{5.35}
 h\equiv h(x,\bar{y},\xi,\tilde{\eta},\bar{\zeta}):=\int_0^1\big(\nabla_\xi d\big)\big(x,\lambda(x,\bar{y},\xi+\tilde{\eta})+t\bar{\zeta}\big)dt.
\end{equation}
In the integral defining $e$ (in the statement of Lemma \ref{L5.3}) we make the change of variables:
\begin{equation}\label{5.36}
 \tilde{\bar{z}}:=\bar{z}+h(x,\bar{y},\xi,\tilde{\eta},\bar{\zeta})
\end{equation}
and obtain
\begin{equation}\label{5.37}
 e(x,\xi)=\int_{\mathbb{R}^{4d}}e^{i(<\tilde{\bar{z}},\bar{\zeta}>+<\bar{y},\tilde{\eta}>)}\tilde{f}(x,\bar{y},\tilde{\bar{z}}-h,\xi,\tilde{\eta},\bar{\zeta})D(x,\bar{y},\xi+\tilde{\eta})\,d\bar{y}\,d\tilde{\bar{z}}\,\dbar\tilde{\eta}\,\dbar\bar{\zeta}.
\end{equation}
We decompose the above integral as a sum of two terms corresponding to the following equality:
$$
\tilde{f}(x,\bar{y},\tilde{\bar{z}}-h,\xi,\tilde{\eta},\bar{\zeta})=\tilde{f}(x,\bar{y},-h,\xi,\tilde{\eta},\bar{\zeta})
+\left<\tilde{\bar{z}},\int_0^1\big(\nabla_{\bar{z}}\tilde{f}\big)(x,\bar{y},s\tilde{\bar{z}}-h,\xi,\tilde{\eta},\bar{\zeta})ds\right>.
$$
We obtain
\begin{equation}\label{5.38}
 e\,=\,e_0\,+\,i\int_0^1\left[\int_{\mathbb{R}^{4d}}e^{i(<\tilde{\bar{z}},\bar{\zeta}>+<\bar{y},\tilde{\eta}>)}\left<\nabla_{\bar{\zeta}},\big(\nabla_{\bar{z}}\tilde{f}\big)(x,\bar{y},s\tilde{\bar{z}}-h,\xi,\tilde{\eta},\bar{\zeta})\right>D(x,\bar{y},\xi+\tilde{\eta})\,d\bar{y}\,d\tilde{\bar{z}}\,\dbar\tilde{\eta}\,\dbar\bar{\zeta}\right]ds.
\end{equation}
We finish the proof as that of Proposition \ref{P5.2}.
\end{proof}
\begin{lemma}\label{L5.5}
If we denote by
$$
c_0(x,\xi):=f\big(x,x,\nabla_\xi U(x,\lambda(x,0,\xi)),\xi,\lambda(x,0,\xi),\lambda(x,0,\xi)\big)D(x,0,\xi),
$$
then $e_0-c_0\,\in\,S^{m+m'+m''-1}(\mathbb{R}^d)$.
\end{lemma}
\begin{proof}
 We denote by 
\begin{equation}\label{5.43}
 g(x,\bar{y},\xi,\theta):=f\big(x,x-\bar{y},x+\big(\nabla_x d\big)\big(x,\lambda(x,\bar{y},\theta)\big),\xi,\lambda(x,\bar{y},\theta),\lambda(x,\bar{y},\theta)\big)D(x,\bar{y},\theta)
\end{equation}
and notice that for any $(\alpha,\beta,\gamma)\in\mathbb{N}^{2d}\times\mathbb{N}^d\times\mathbb{N}^d$ we have the inequality
\begin{equation}\label{5.44}
 \left|\big(\partial^\alpha_{(x,\bar{y})}\partial^\beta_\xi\partial^\gamma_\theta g\big)(x,\bar{y},\xi,\theta)\right|\leq C_f<\xi>^{m-|\beta|}<\theta>^{m'+m''-|\gamma|}.
\end{equation}
We write
\begin{equation}\label{5.45}
 g(x,\bar{y},\xi,\xi+\tilde{\eta})=g(x,\bar{y},\xi,\xi)+\left<\tilde{\eta},\int_0^1\big(\nabla_\theta g\big)(x,\bar{y},\xi,\xi+t\tilde{\eta})dt\right>
\end{equation}
and the associated decomposition of the integral defining $e_0$ in Lemma \ref{L5.4}
\begin{equation}\label{5.46}
 e_0(x,\xi)=g(x,0,\xi,\xi)+\rho(x,\xi),\quad\rho(x,\xi):=i\int_0^1\rho_t(x,\xi)dt,
\end{equation}
\begin{equation}\label{5.48}
 \rho_t(x,\xi):=\int_{\mathbb{R}^{2d}}e^{i<\bar{y},\tilde{\eta}>}\left<\nabla_{\bar{y}},\big(\nabla_{\theta}g\big)(x,\bar{y},\xi,\xi+t\tilde{\eta})\right>\,d\bar{y}\,\dbar\tilde{\eta}.
\end{equation}
Usual integration by parts using the operators $<\bar{y}>^{-2N_1}(1-\Delta_{\tilde{\eta}})^{N_1}$ and $<\tilde{\eta}>^{-2N_2}(1-\Delta_{\bar{y}})^{N_2}$ and the use of inequality \eqref{5.44} allow us to obtain the estimation
$$
\left|\rho_t(x,\xi)\right|\,\leq\,C_f\int_{\mathbb{R}^{2d}}<\bar{y}>^{-2N_1}<\tilde{\eta}>^{-2N_2}<\xi>^m<\xi+t\tilde{\eta}>^{m'+m''-1}\,d\bar{y}\,\dbar\tilde{\eta}\,\leq
$$
$$
\leq\,C_f<\xi>^m\left[\int\limits_{|\tilde{\eta}|\leq(1/2)|\xi|}\hspace{-0.5cm}<\tilde{\eta}>^{-2N_2}<\xi+t\tilde{\eta}>^{m'+m''-1}\,\dbar\tilde{\eta}\,+\hspace{-0.5cm}\int\limits_{|\tilde{\eta}|>(1/2)|\xi|}\hspace{-0.5cm}<\tilde{\eta}>^{-2N_2}<\xi+t\tilde{\eta}>^{m'+m''-1}\,\dbar\tilde{\eta}\right]\,\leq
$$
$$
\leq\,C_f<\xi>^{m+m'+m''-1},
$$
choosing $2N_1\geq d+1,\ 2N_2\geq d+1+|m'+m''-1|$.
We proceed in a similar way with the derivatives of $\rho_t$ and deduce that $\rho_t\in S^{m+m'+m''-1}(\mathbb{R}^d)$ uniformly with respect to $t\in[0,1]$ and in conclusion $\rho\in S^{m+m'+m''-1}(\mathbb{R}^d)$.
\end{proof}

\begin{proposition}\label{P5.6}
 If $f\in S^{m,m',m''}\big(\mathbb{R}^{3d}\times\mathbb{R}^d\times\mathbb{R}^d\times\mathbb{R}^d\big)$ and the symbol $c$ is defined by \eqref{5.8}, then
$$
c(x,\xi)-f\big(x,x,\nabla_\xi U(x,\lambda(x,0,\xi)),\xi,\lambda(x,0,\xi),\lambda(x,0,\xi)\big)D(x,0,\xi)\,\in\,S^{m+m'+m''-1}(\mathbb{R}^d).
$$
\end{proposition}
\begin{proof}
 We use Lemmas \ref{L5.3}, \ref{L5.4} and \ref{L5.5} and the following evident facts:
\begin{itemize}
 \item $\lambda(x,0,\xi)$ is the inverse of the $C^\infty$-diffeomorphism:
\begin{equation}\label{5.49}
 \mathbb{R}^d\ni\xi\mapsto\big(\nabla_x U\big)(x,\xi)\in\mathbb{R}^d;
\end{equation}
\item $D(x,0,\xi)$ is the functional determinant of the transformation 
\begin{equation}\label{5.50}
 \mathbb{R}^d\ni\xi\mapsto\lambda(x,0,\xi)\in\mathbb{R}^d,
\end{equation}
so that
\begin{equation}\label{5.51}
 D(x,0,\xi)\,=\,\left|\det\left(\nabla^2_{\xi,x}U\right)\big(x,\lambda(x,0,\xi)\big)\right|^{-1}.
\end{equation}
\end{itemize}
\end{proof}

Considering now in Proposition \ref{P5.6} the particular case $f(x,y,z,\xi,\eta,\zeta)=a(x,\zeta)\overline{b(y,\eta)}$ and using Propositions \ref{P5.1} and \ref{P5.2} we obtain the main result of this section.
\begin{theorem}\label{T5.7}
  If $a\in S^{m''}(\mathbb{R}^d)$ and $b\in S^{m'}(\mathbb{R}^d)$, then $\mathfrak{Op}^A_\Phi(a)\circ\left[\mathfrak{Op}^A_\Phi(b)\right]^*=c^A(x,D)$ with $c\in S^{m'+m''}(\mathbb{R}^d)$. Moreover
\begin{equation}\label{5.53}
 c(x,\xi)-a\big(x,\lambda(x,0,\xi)\big)\overline{b\big(x,\lambda(x,0,\xi)\big)}\left|\det\left(\nabla^2_{\xi,x}U\right)\big(x,\lambda(x,0,\xi)\big)\right|^{-1}\,\in\,S^{m'+m''-1}(\mathbb{R}^d),
\end{equation}
with $\lambda(x,0,\xi)$ the inverse of the diffeomorphism \eqref{5.49}.
\end{theorem}

\section{Composing the adjoint of a FIO with a FIO}

Suppose that $a\in S^{m''}(\mathbb{R}^d)$ and $b\in S^{m'}(\mathbb{R}^d)$, and let us define $E:=\big[\mathfrak{Op}^A_\Phi(a)\big]^*$ and $F:=\mathfrak{Op}^A_\Phi(b)$. Explicitely we have that for any $u\in\mathcal{S}(\mathbb{R}^d)$ and any point $x\in\mathbb{R}^d$
\begin{equation}\label{6.1}
 \big(Eu\big)(x)\,=\,\int_{\mathbb{R}^{2d}}e^{i\big(<x,\zeta>-U(z,\zeta)\big)}\omega^A(x,z)\overline{a(z,\zeta)}u(z)\,dz\,\dbar\zeta,
\end{equation}
\begin{equation}\label{6.2}
 \big(Fu\big)(x)\,=\,\int_{\mathbb{R}^{2d}}e^{i\big(U(x,\eta)-<y,\eta>\big)}\omega^A(x,y)b(x,\eta)u(y)\,dy\,\dbar\eta.
\end{equation}
We know that $E\circ F$ is an operator in $\mathbb{B}\big(\mathcal{S}(\mathbb{R}^d)\big)$; we shall prove in this section that it is in fact a $\Psi$DO and we shall compute a principal symbol for it.

We proceed as in the previous sections and fix a cut-off function $\chi\in C^\infty_0(\mathbb{R}^{2d})$ with $\chi(0,0)=1$ and for $\epsilon\in(0,1]$ we define $a_\epsilon(z,\zeta):=\chi(\epsilon z,\epsilon\zeta)\overline{a(z,\zeta)}$ and $b_\epsilon(y,z,\eta):=\chi(\epsilon y,\epsilon\eta)b(z,\eta)$. Then we denote by $E_\epsilon$ and $F_\epsilon$ the operators defined respectively by \eqref{6.1} with $a$ replaced by $a_\epsilon$ and by equation \eqref{6.2} with $b$ replaced by $b_\epsilon$. The distribution kernel $K_\epsilon\in\mathcal{S}(\mathbb{R}^{2d})$ of the product $E_\epsilon\circ F_\epsilon$ is given by the following integral
\begin{equation}\label{6.3}
 K_\epsilon(x,y)\,:=\,\omega^A(x,y)\int_{\mathbb{R}^{3d}}e^{i\big(<x,\zeta>-<y,\eta>+U(z,\eta)-U(z,\zeta)\big)}\Omega^B(x,z,y)a_\epsilon(z,\zeta)b_\epsilon(y,z,\eta)\,dz\,\dbar\eta\,\dbar\zeta.
\end{equation}
As in the previous section we use the equation relating the distribution kernel of a $\Psi$DO to its symbol (the special case of \eqref{2.5} when $\Phi=1$, i.e. $U(x,\eta)=<x,\eta>$) and obtain that
\begin{equation}\label{6.4}
 E_\epsilon\circ F_\epsilon\,=\,c^A_\epsilon(x,D),
\end{equation}
with $c_\epsilon\in\mathcal{S}(\mathbb{R}^d)$ given by
\begin{equation}\label{6.5}
 c_\epsilon(x,\xi)=\int_{\mathbb{R}^{4d}}e^{i\psi(x,y,z,\xi,\eta,\zeta)}\Omega^B(x,z,y)a_\epsilon(z,\zeta)b_\epsilon(y,z,\eta)\,dy\,dz\,\dbar\eta\,\dbar\zeta,
\end{equation}
\begin{equation}\label{6.6}
\psi(x,y,z,\xi,\eta,\zeta):=<x,\zeta>-<y,\eta>-<x-y,\xi>+U(z,\eta)-U(z,\zeta).
\end{equation}

Once again we notice that once we proved the existence of the following limit
\begin{equation}\label{6.7}
 \underset{\epsilon\searrow0}{\lim}c_\epsilon\,=\,c,\text{ in }S^{m'+m''}(\mathbb{R}^d),
\end{equation}
taking into account that for any $v\in\mathcal{S}(\mathbb{R}^d)$ we have $\underset{\epsilon\searrow0}{\lim}\big(E_\epsilon\circ F_\epsilon\big)v=\big(E\circ F\big)v$, we have that the following statement is true.
\begin{proposition}\label{P6.1}
 If $a\in S^{m''}(\mathbb{R}^d)$ and $b\in S^{m'}(\mathbb{R}^d)$, then $\left[\mathfrak{Op}^A_\Phi(a)\right]^*\circ\mathfrak{Op}^A_\Phi(b)=c^A(x,D)$ with $c\in S^{m'+m''}(\mathbb{R}^d)$ defined by \eqref{6.7}.
\end{proposition}

Now, the existence of the limit \eqref{6.7} follows from the next Proposition.
\begin{proposition}\label{P6.2}
 Suppose $f\in S^{m,m',m''}\big(\mathbb{R}^{3d}\times\mathbb{R}^d\times\mathbb{R}^d\times\mathbb{R}^d\big)$ and $\psi$ is the function defined in \eqref{6.6}. Then the following oscillating integral
\begin{equation}\label{6.8}
 c(x,\xi)\,:=\,\int_{\mathbb{R}^{4d}}e^{i\psi(x,y,z,\xi,\eta,\zeta)}\Omega^B(x,z,y)f(x,y,z,\xi,\eta,\zeta)\,dy\,dz\,\dbar\eta\,\dbar\zeta,
\end{equation}
defines an element in $S^{m+m'+m''}(\mathbb{R}^d)$ and the map
$$
S^{m,m',m''}\big(\mathbb{R}^{3d}\times\mathbb{R}^d\times\mathbb{R}^d\times\mathbb{R}^d\big)\ni f\mapsto c\in S^{m+m'+m''}(\mathbb{R}^d)
$$
is continuous.
\end{proposition}

\begin{proof}
 As before, we can work with $f$ with compact support with respect to the variables $(y,z,\eta,\zeta)$. We remark that
\begin{equation}\label{6.1.9}
\psi(x,y,z,\xi,\eta,\zeta)=<x-y,\eta-\xi>+<x-z,\zeta-\eta>+d(z,\eta)-d(z,\zeta)
\end{equation}
and make the change of variables
\begin{equation}\label{6.10}
 \bar{y}:=x-y,\quad\bar{z}:=x-z,\quad\bar{\eta}:=\eta-\xi,\quad\bar{\zeta}:=\zeta-\eta
\end{equation}
and the notations:
\begin{equation}\label{6.11}
 \bar{\psi}\equiv\bar{\psi}(x,\bar{y},\bar{z},\xi,\bar{\eta},\bar{\zeta}):=<\bar{z},\bar{\zeta}>+<\bar{y},\bar{\eta}>+\bar{d}(x,\bar{z},\xi,\bar{\eta},\bar{\zeta}),
\end{equation}
\begin{equation}\label{6.12}
\bar{d}\equiv \bar{d}(x,\bar{z},\xi,\bar{\eta},\bar{\zeta}):=d(x-\bar{z},\xi+\bar{\eta})-d(x-\bar{z},\xi+\bar{\eta}+\bar{\zeta}),
\end{equation}
\begin{equation}\label{6.13}
 \bar{\Omega^B}\equiv\bar{\Omega^B}(x,\bar{y},\bar{z}):=\Omega^B(x,x-\bar{z},x-\bar{y}),
\end{equation}
\begin{equation}\label{6.14}
 \bar{f}\equiv\bar{f}(x,\bar{y},\bar{z},\xi,\bar{\eta},\bar{\zeta}):=f(x,x-\bar{y},x-\bar{z},\xi,\xi+\bar{\eta},\xi+\bar{\eta}+\bar{\zeta}),
\end{equation}
so that we can write
\begin{equation}\label{6.15}
 c(x,\xi)\,:=\,\int_{\mathbb{R}^{4d}}e^{i\bar{\psi}}\bar{\Omega^B}\bar{f}\,d\bar{y}\,d\bar{z}\,\dbar\bar{\eta}\,\dbar\bar{\zeta}.
\end{equation}
We integrate by parts using the operators:
\begin{equation}\label{6.16}
 L_{\bar{y}}:=<\bar{\eta}>^{-2}(1-i<\bar{\eta},\nabla_{\bar{y}}>),\quad L_{\bar{z}}:=\left<\bar{\zeta}+\nabla_{\bar{z}}\bar{d}\right>^{-2}\left(1-i\left<\bar{\zeta}+\nabla_{\bar{z}}\bar{d},\nabla_{\bar{z}}\right>\right),
\end{equation}
$$
L_{\bar{\eta}}:=\left<\bar{y}+\nabla_{\bar{\eta}}\bar{d}\right>^{-2}\left(1-i\left<\bar{y}+\nabla_{\bar{\eta}}\bar{d},\nabla_{\bar{\eta}}\right>\right),\quad L_{\bar{\zeta}}:=\left<\bar{z}+\nabla_{\bar{\zeta}}\bar{d}\right>^{-2}\left(1-i\left<\bar{z}+\nabla_{\bar{\zeta}}\bar{d},\nabla_{\bar{\zeta}}\right>\right).
$$
We notice that $\nabla_{\bar{\eta}}\bar{d}$ and $\nabla_{\bar{\zeta}}\bar{d}$ are functions of class $BC^\infty(\mathbb{R}^{5d})$, so that only the operator $L_{\bar{z}}$ needs some special care. But we can write
\begin{equation}\label{6.17}
 \bar{d}(x,\bar{z},\xi,\bar{\eta},\bar{\zeta})=\left<\bar{h}(x,\bar{z},\xi,\bar{\eta},\bar{\zeta}),\bar{\zeta}\right>,
\end{equation}
with
\begin{equation}\label{6.18}
 \bar{h}\equiv\bar{h}(x,\bar{z},\xi,\bar{\eta},\bar{\zeta}):=-\int_0^1\big(\nabla_\xi d\big)(x-\bar{z},\xi+\bar{\eta}+t\bar{\zeta})dt
\end{equation}
so that
\begin{equation}\label{6.19}
 \bar{\zeta}+\nabla_{\bar{z}}\bar{d}\,=\,\left[1_d\,+\,\int_0^1\big(\nabla^2_{x,\xi}d\big)(x-\bar{z},\xi+\bar{\eta}+t\bar{\zeta})dt\right]\cdot\bar{\zeta}.
\end{equation}
We use Hypothesis \ref{Hyp-U} to deduce that there exists a constant $C>0$ such that
\begin{equation}\label{6.20}
 C^{-1}<\bar{\zeta}>\,\leq\,<\bar{\zeta}+\nabla_{\bar{z}}\bar{d}>\,\leq\,C<\bar{\zeta}>,\forall(x,\bar{z},\xi,\bar{\eta},\bar{\zeta})\in\mathbb{R}^{5d}.
\end{equation}
We conclude from the above analysis that the formal adjoint of $L_{\bar{z}}$ has the form
\begin{equation}\label{6.21}
 ^tL_{\bar{z}}=<\bar{\zeta}+\nabla_{\bar{z}}\bar{d}>^{-1}\left(a+<b,\nabla_{\bar{z}}>\right),
\end{equation}
with $a$ and $b$ functions of class $BC^\infty(\mathbb{R}^{5d})$.

Let us notice that differentiating $\overline{\Omega^B}$ with respect to $\bar{y}$ and $\bar{z}$ produces monomials in these variables and we have to get rid of them by usual integration by parts with respect to $\bar{\eta}$ and $\bar{\zeta}$. Finally with a convenient choice for $N_j\in\mathbb{N}$ ($1\leq j\leq4$) we get
\begin{equation}\label{6.22}
 \left|c(x,\xi)\right|\,\leq
\end{equation}
$$
\leq\,C_f\int_{\mathbb{R}^{4d}}<\bar{\eta}>^{-2N_1}<\bar{\zeta}>^{-2N_2}<\bar{y}>^{-2N_3}<\bar{z}>^{-2N_4}<\xi>^m<\xi+\bar{\eta}>^{m'}<\xi+\bar{\eta}+\bar{\zeta}>^{m''}d\bar{y}d\bar{z}\dbar\bar{\eta}\dbar\bar{\zeta}\,\leq
$$
$$
\leq\,C_f<\xi>^{m+m'+m''},\quad\forall(x,\xi)\in\mathbb{T}^*\mathbb{R}^d.
$$
Proceeding similarly with the derivatives of $c$ we end the proof of the Proposition.
\end{proof}

In order to compute a principal part of the symbol $c$ we proceed as in the previous section.
\begin{lemma}\label{L6.3}
 If we define
$$
e(x,\xi):=\int_{\mathbb{R}^{4d}}e^{i(<\bar{z},\bar{\zeta}>+<\bar{y},\tilde{\eta}>+\tilde{d})}\bar{f}\,d\bar{y}\,d\bar{z}.
$$
we have that $c-e\,\in\,S^{m+m'+m''-1}(\mathbb{R}^d)$.
\end{lemma}
\begin{proof}
 As in the proof of Lemma \ref{L4.3} we start from the equality
\begin{equation}\label{6.25}
 \bar{\Omega^B}=e^{-i\bar{F}}=1-i\bar{F}\int_0^1e^{-it\bar{F}}dt,\quad\bar{F}(x,\bar{y},\bar{z})=\left<\bar{C}(x,\bar{y},\bar{z})\bar{z},\bar{y}\right>
\end{equation}
with $\bar{C}\in BC^\infty(\mathbb{R}^{3d})$ as in \eqref{3.139}. Using this decomposition and integrating by parts to get rid of the factors $\bar{y}$ and $\bar{z}$ we get
\begin{equation}\label{6.27}
c=e+i\int_0^1r_t(x,\xi)dt,
\end{equation}
\begin{equation}\label{6.29}
 r_t(x,\xi)=\int_{\mathbb{R}^{4d}}e^{i(<\bar{z},\bar{\zeta}>+<\bar{y},\bar{\eta}>)}e^{-it\bar{F}}\left<\bar{C}\nabla_{\bar{\zeta}},\nabla_{\bar{\eta}}\right>\big(e^{i\bar{d}}\bar{f}\big)\,d\bar{y}\,d\bar{z}\,\dbar\bar{\eta}\,\dbar\bar{\zeta}.
\end{equation}
We develop $\left<\bar{C}\nabla_{\bar{\zeta}},\nabla_{\bar{\eta}}\right>\big(e^{i\bar{d}}\bar{f}\big)$ and notice that we can divide the terms into two classes:
\begin{itemize}
 \item terms containing at least one derivative of $\bar{f}$ or exactly two derivatives of $\bar{d}$ (with respect to $\bar{\eta}$ and $\bar{\zeta}$); the corresponding integrals are estimated as in the proof of Proposition \ref{P6.2} and give symbols of class $S^{m+m'+m''-1}(\mathbb{R}^d)$ uniformly with respect to $t\in[0,1]$;
\item the term $\big(e^{i\bar{d}}\bar{f}\big)\left<\bar{C}\nabla_{\bar{\zeta}}\bar{d},\nabla_{\bar{\eta}}\bar{d}\right>$; in order to estimate the corresponding integral we write the equality:
$$
\left<\bar{C}\nabla_{\bar{\zeta}}\bar{d},\nabla_{\bar{\eta}}\bar{d}\right>=-\left<\bar{C}\nabla_{\xi}d(x-\bar{z},\xi+\bar{\eta}+\bar{\zeta}),\nabla_{\xi}d(x-\bar{z},\xi+\bar{\eta})-\nabla_{\xi}d(x-\bar{z},\xi+\bar{\eta}+\bar{\zeta})\right>=
$$
$$
=\int_0^1\left<\bar{C}\nabla_{\xi}d(x-\bar{z},\xi+\bar{\eta}+\bar{\zeta}),\big(\nabla^2_{\xi,\xi}d\big)(x-\bar{z},\xi+\bar{\eta}+s\bar{\zeta})\cdot\bar{\zeta}\right>\,ds
$$
and use it to estimate the corresponding integral by
$$
C_f\underset{0\leq s\leq1}{\sup}\int_{\mathbb{R}^{4d}}<\bar{\eta}>^{-2N_1}<\bar{\zeta}>^{-2N_2}<\bar{y}>^{-2N_3}<\bar{z}>^{-2N_4}\times
$$
$$
\times<\xi>^m<\xi+\bar{\eta}>^{m'}<\xi+\bar{\eta}+\bar{\zeta}>^{m''}<\bar{\zeta}><\xi+\bar{\eta}+s\bar{\zeta}>^{-1}d\bar{y}d\bar{z}\dbar\bar{\eta}\dbar\bar{\zeta}\,\leq
$$
$$
\leq\,C_f<\xi>^{m+m'+m''-1},\quad\forall(x,\xi)\in\mathbb{T}^*\mathbb{R}^d,\ \text{uniformly in}\ t\in[0,1],
$$
after we choose $2N_3\geq d+1$, $2N_4\geq d+1$, $2N_1\geq d+2+|m'|+|m''|$ and $2N_2\geq d+3+|m''|$. After similar estimations for the derivatives of $r_t$ we deduce that $r_t\in S^{m+m'+m''-1}(\mathbb{R}^d)$ uniformly with respect to $t\in[0,1]$ and conclude that $r=\int_0^1r_tdt$ is a symbol in $S^{m+m'+m''-1}(\mathbb{R}^d)$.
\end{itemize}
\end{proof}

\begin{lemma}\label{L6.4}
 If we define
$$
e_0(x,\xi):=\int_{\mathbb{R}^{2d}}e^{i(<\bar{z},\bar{\zeta}>+\bar{d}(x,\bar{z},\xi,0,\bar{\zeta}))}\bar{f}(x,0,\bar{z},\xi,0,\bar{\zeta})\,d\bar{z}\,\dbar\bar{\zeta},
$$
then $e-e_0\in S^{m+m'+m''-1}(\mathbb{R}^d)$.
\end{lemma}
\begin{proof}
 Using the equality:
$$
\bar{f}(x,\bar{y},\bar{z},\xi,\bar{\eta},\bar{\zeta})=\bar{f}(x,0,\bar{z},\xi,\bar{\eta},\bar{\zeta})+\left<\bar{y},\int_0^1\big(\nabla_{\bar{y}}\bar{f}\big)(x,t\bar{y},\bar{z},\xi,\bar{\eta},\bar{\zeta})dt\right>,
$$
we can write
\begin{equation}\label{6.32}
 e=e_0+i\int_0^1\rho_tdt
\end{equation}
with
\begin{equation}\label{6.34}
 \rho_t(x,\xi):=\int_{\mathbb{R}^{4d}}e^{i(<\bar{z},\bar{\zeta}>+<\bar{y},\bar{\eta}>)}\left<\nabla_{\bar{\eta}},e^{i\bar{d}}\big(\nabla_{\bar{y}}\bar{f}\big)(x,t\bar{y},\bar{z},\xi,\bar{\eta},\bar{\zeta})\right>d\bar{y}\,d\bar{z}\,\dbar\bar{\eta}\,\dbar\bar{\zeta}.
\end{equation}
As in the proof of Lemma \ref{L6.3} there is only one special term in \eqref{6.34} (the other having usual estimations), namely the one containing the factor $\nabla_{\bar{\eta}}\big(e^{i\bar{d}}\big)=ie^{i\bar{d}}\nabla_{\bar{\eta}}\bar{d}$. In order to estimate it we notice that
$$
\big(\nabla_{\bar{\eta}}\bar{d}\big)(x,\bar{z},\xi,\bar{\eta},\bar{\zeta})=\big(\nabla_{\xi} d\big)(x-\bar{z},\xi+\bar{\eta})-\big(\nabla_{\xi}d\big)(x-\bar{z},\xi+\bar{\eta}+\bar{\zeta})=-\left[\int_0^1\big(\nabla^2_{\xi,\xi}d\big)(x-\bar{z},\xi+\bar{\eta}+s\bar{\zeta})ds\right]\cdot\bar{\zeta},
$$
and the proof finishes as in the case of Lemma \ref{L6.3}.
\end{proof}

\begin{proposition}\label{P6.5}
 If $f\in S^{m,m',m''}\big(\mathbb{R}^{3d}\times\mathbb{R}^d\times\mathbb{R}^d\times\mathbb{R}^d\big)$ and the symbol $c$ is defined by \eqref{6.8}, then
$$
c(x,\xi)-f\big(x,x,V(x,\xi),\xi,\xi,\xi\big)\left|\det\left(\nabla^2_{x,\xi}U\right)\big(V(x,\xi),\xi\big)\right|^{-1}\,\in\,S^{m+m'+m''-1}(\mathbb{R}^d)
$$
where $y:=V(x,\xi)$ is the unique solution of the equation $\big(\nabla_\xi U\big)(y,\xi)=x$ and we have that $V(x,\xi)-x\in S^0(\mathbb{R}^d)$.
\end{proposition}
\begin{proof}
In the integral defining $e_0$ in Lemma \ref{L6.4} we consider the change of variables:
\begin{equation}\label{6.36}
 \tilde{z}:=\bar{z}+\bar{h},\quad\bar{h}\equiv\bar{h}(x,\bar{z},\xi,\bar{\zeta}):=-\int_0^1\big(\nabla_\xi d\big)(x-\bar{z},\xi+t\bar{\zeta})dt.
\end{equation}
Lemma \ref{L1.6} implies that for any $(x,\xi,\bar{\zeta})\in\mathbb{R}^{3d}$ the map
$$
\Rd\ni\bar{z}\mapsto\bar{z}+\bar{h}\in\Rd
$$
is a $C^\infty$-diffeomorphism having the inverse that we denote by $\tilde{g}\equiv\tilde{g}(x,\tilde{z},\xi,\bar{\zeta})$. We denote by $\tilde{D}\equiv\tilde{D}(x,\tilde{z},\xi,\bar{\zeta})$ the functional determinant of the map $\Rd\ni\tilde{z}\mapsto\tilde{g}(x,\tilde{z},\xi,\bar{\zeta})\in\Rd$. Then $\tilde{g}-\tilde{z}$ and $\tilde{D}$ are functions of class $BC^\infty(\mathbb{R}^{4d})$ and on the domain $|\bar{\zeta}|\leq\epsilon|\xi|$, for $\epsilon>0$ small enough we have satisfied the following estimations:
\begin{equation}\label{6.37}
 \left|\partial^\alpha_{(x,\tilde{z})}\partial^\beta_{(\xi,\bar{\zeta})}(\tilde{g}-\tilde{z})\right|+\left|\partial^\alpha_{(x,\tilde{z})}\partial^\beta_{(\xi,\bar{\zeta})}\tilde{D}\right|\,\leq\,C<\xi>^{-|\beta|},\quad\forall(\alpha,\beta)\in[\mathbb{N}^{2d}]^2.
\end{equation}
Let us denote by
\begin{equation}\label{6.38}
 \tilde{f}\equiv\tilde{f}(x,\tilde{z},\xi,\bar{\zeta}):=\bar{f}\big(x,0,\tilde{g}(x,\tilde{z},\xi,\bar{\zeta}),\xi,0,\bar{\zeta}\big)=f\big(x,x,x-\tilde{g}(x,\tilde{z},\xi,\bar{\zeta}),\xi,\xi,\xi+\bar{\zeta}\big).
\end{equation}
Using the change of variables \eqref{6.36} we obtain that
\begin{equation}\label{6.39}
 e_0(x,\xi)=\int_{\mathbb{R}^{2d}}e^{i<\tilde{z},\bar{\zeta}>}\tilde{f}(x,\tilde{z},\xi,\bar{\zeta})\tilde{D}(x,\tilde{z},\xi,\bar{\zeta})\,d\tilde{z}\,\dbar\bar{\zeta}.
\end{equation}
We notice that the function
\begin{equation}\label{6.40}
 q\equiv q(x,\tilde{z},\xi,\bar{\zeta}):=\tilde{f}(x,\tilde{z},\xi,\bar{\zeta})\tilde{D}(x,\tilde{z},\xi,\bar{\zeta})
\end{equation}
verifies for any $(\alpha,\beta)\in[\mathbb{N}^{2d}]^2$ the estimations:
\begin{equation}\label{6.41}
 \left|\partial^\alpha_{(x,\tilde{z})}\partial^\beta_{(\xi,\bar{\zeta})}q\right|\,\leq\,C_f<\xi>^{m+m'}<\xi+\bar{\zeta}>^{m''}\ \text{ on }\mathbb{R}^{4d},
\end{equation}
\begin{equation}\label{6.42}
 \left|\partial^\alpha_{(x,\tilde{z})}\partial^\beta_{(\xi,\bar{\zeta})}q\right|\,\leq\,C_f<\xi>^{m+m'+m''-|\beta|},\text{ on }\{|\bar{\zeta}|\leq\epsilon|\xi|\}.
\end{equation}
We write the integral in \eqref{6.39} as a sum of two terms defined using the identity
\begin{equation}\label{6.43}
 q(x,\tilde{z},\xi,\bar{\zeta})=q(x,\tilde{z},\xi,0)+\left<\bar{\zeta},\int_0^1\big(\nabla_{\bar{\zeta}}q\big)(x,\tilde{z},\xi,t\bar{\zeta})dt\right>.
\end{equation}
We obtain then the decomposition:
\begin{equation}\label{6.44}
 e_0(x,\xi)=q(x,0,\xi,0)+\rho(x,\xi),\quad\rho(x,\xi)=i\int_0^1\rho_t(x,\xi)dt,
\end{equation}
with
\begin{equation}\label{6.45}
 \rho_t(x,\xi):=\int_{\mathbb{R}^{2d}}e^{i<\tilde{z},\bar{\zeta}>}\left<\nabla_{\tilde{z}},\big(\nabla_{\bar{\zeta}}q\big)(x,\tilde{z},\xi,t\bar{\zeta})\right>\,d\tilde{z}\,\dbar\bar{\zeta}.
\end{equation}
We integrate by parts using the operators $<\tilde{z}>^{-2N_1}(1-\Delta_{\bar{\zeta}})^{N_1}$ and $<\bar{\zeta}>^{-2N_2}(1-\Delta_{\tilde{z}})^{N_2}$ and choosing then $2N_1\geq d+1$ and $2N_2\geq d+2+|m''|$ and using also \eqref{6.41} and \eqref{6.42} we obtain the estimation
$$
\left|\rho_t(x,\xi)\right|\,\leq\,C_f\left(\int\limits_{|\bar{\zeta}|\leq\epsilon|\xi|}<\bar{\zeta}>^{-2N_2}<\xi>^{m+m'+m''-1}\dbar\bar{\zeta}\,+\,\int\limits_{|\bar{\zeta}|>\epsilon|\xi|}<\bar{\zeta}>^{-2N_2}<\xi>^{m+m'}<\xi+t\bar{\zeta}>^{m''}\dbar\bar{\zeta}\right)\,\leq
$$
$$
\leq\,C_f<\xi>^{m+m'+m''-1}.
$$
After similar estimations for the derivatives of $\rho_t$ we deduce that $\rho_t\in S^{m+m'+m''-1}(\mathbb{R}^d)$ uniformly with respect to $t\in[0,1]$ and in conclusion $\rho\in S^{m+m'+m''-1}(\mathbb{R}^d)$.

Let us notice that
\begin{equation}\label{6.47}
 q(x,0,\xi,0)=f\big(x,x,x-\tilde{g}(x,0,\xi,0),\xi,\xi,\xi\big)\tilde{D}(x,0,\xi,0).
\end{equation}
But
$$
\bar{z}=\tilde{g}(x,\tilde{z},\xi,\bar{\zeta})\quad\Longleftrightarrow\quad\tilde{z}=\bar{z}+\bar{h}(x,\bar{z},\xi,\bar{\zeta}),
$$
\begin{equation}\label{6.48}
 \tilde{D}(x,\tilde{z},\xi,\bar{\zeta})=\left|\det\left(\nabla_{\tilde{z}}\tilde{g}(x,\tilde{z},\xi,\bar{\zeta})\right)\right|=\left|\det\left(1_d\,+\,\big(\nabla_{\bar{z}}\bar{h}\big)\big(x,\tilde{g}(x,\tilde{z},\xi,\bar{\zeta}),\xi,\bar{\zeta}\big)\right)\right|^{-1}
\end{equation}
and we obtain that
\begin{equation}\label{6.49}
 \bar{z}=\tilde{g}(x,0,\xi,0)\ \Leftrightarrow\ 0=\bar{z}+\bar{h}(x,\bar{z},\xi,0)\ \Leftrightarrow\ \bar{z}=\big(\nabla_\xi d\big)(x-\bar{z},\xi)\ \Leftrightarrow\ x=\big(\nabla_\xi U\big)(x-\bar{z},\xi)\ \Leftrightarrow\ 
\end{equation}
$$
\ \Leftrightarrow\ \bar{z}=x-V(x,\xi)\ \Leftrightarrow\ \tilde{g}(x,0,\xi,0)=x-V(x,\xi).
$$
Using this last equality in \eqref{6.37} we deduce that $V(x,\xi)-x\in S^0(\mathbb{R}^d)$ and from \eqref{6.48} we obtain that
\begin{equation}\label{6.50}
 \tilde{D}(x,0,\xi,0)=\left|\det\left(\nabla^2_{x,\xi}U\right)\big(V(x,\xi),\xi\big)\right|^{-1}\in S^0(\mathbb{R}^d).
\end{equation}
The conclusion of the Proposition follows now from \eqref{6.44}, \eqref{6.47}, \eqref{6.49} and \eqref{6.50}.
\end{proof}

Considering now the particular case $f(x,y,z,\xi,\eta,\zeta)=\overline{a(z,\zeta)}b(z,\eta)$ and using Propositions \ref{P6.2} and \ref{P6.5} we obtain the main result of this section.
\begin{theorem}\label{T6.6}
  If $a\in S^{m''}(\mathbb{R}^d)$ and $b\in S^{m'}(\mathbb{R}^d)$, then $\left[\mathfrak{Op}^A_\Phi(a)\right]^*\circ\mathfrak{Op}^A_\Phi(b)=c^A(x,D)$ with $c\in S^{m'+m''}(\mathbb{R}^d)$. Moreover
\begin{equation}\label{6.51}
 c(x,\xi)-\overline{a\big(V(x,\xi),\xi\big)}b\big(V(x,\xi),\xi\big)\left|\det\left(\nabla^2_{x,\xi}U\right)\big(V(x,\xi),\xi\big)\right|^{-1}\,\in\,S^{m'+m''-1}(\mathbb{R}^d),
\end{equation}
with $y:=V(x,\xi)$ the unique solution of the equation $\big(\nabla_\xi U\big)(y,\xi)=x$ and we have that $V(x,\xi)-x\in S^0(\mathbb{R}^d)$.
\end{theorem}

\section{The evolution group}\label{Ev-group}

In this section we shall denote by $a$ a symbol on $\mathbb{R}^d$ satisfying the following properties.
\begin{hypothesis}\label{Hyp-a}
 
\begin{enumerate}
 \item $a$ is a real symbol of class $S^1(\mathbb{R}^d)$,
\item $a$ has a homogeneous principal part $a_0$ (in the sense of Definition \ref{D1.7}),
\item the principal part $a_0$ is real and $\underset{(x,\xi)\in\mathbb{T}^*\mathbb{R}^d,|\xi|=1}{\inf}a_0(x,\xi)>0$.
\end{enumerate}
\end{hypothesis}

Under the above assumptions $a$ is a real elliptic symbol of order 1 and due to Lemma \ref{L1.18}.2, $\mathfrak{Op}^A(a)$ defines a self-adjoint operator $P$ in $L^2(\mathbb{R}^d)$ with domain $\mathcal{H}^1_A$. Our aim in this section is to prove that for some $T>0$ suitable choosen the evolution group $G_t:=e^{-itP}$ is of the form $\mathfrak{Op}^A_{\Phi_t}(b_t)$ for any $t\in(-T,T)$, where $\Phi_t:\mathbb{T}^*\mathbb{R}^d\rightarrow\mathbb{T}^*\mathbb{R}^d$ is a symplectomorphism depending only on $a_0$, and $b_t\in S^0(\mathbb{R}^d)$ is an asymptotic sum (in the sense of Lemma \ref{L1.3}) of a series of symbols that are solutions of a specific family of transport equations.

Let us begin by noticing that following Proposition \ref{P1.14}.1 there exists a unique $b\in S^1(\mathbb{R}^d)$ such that
\begin{equation}\label{7.1}
 P\,\equiv\,\mathfrak{Op}^A(a)\,=\,b^A(x,D).
\end{equation}
Moreover, \eqref{1.16} implies that $a_0$ is also the homogeneous principal part (of order 1) of $b$.

Let us now define $\Phi_t$. We fix some real function $\chi\in C^\infty(\mathbb{R}^d)$ such that $\chi(\xi)=0$ for $|\xi|\leq1$ and $\chi(\xi)=1$ for $|\xi|\geq2$. For any $\rho>0$ we denote by $\chi_\rho(\xi):=\chi(\rho^{-1}\xi)$, $\forall\xi\in\mathbb{R}^d$. We define the real symbol $\bar{a}_0\in S^1(\mathbb{R}^d)$ by the following
\begin{equation}\label{7.2}
 \bar{a}_0(x,\xi):=\chi(\xi)a_0(x,\xi),\quad\forall (x,\xi)\in\mathbb{T}^*\mathbb{R}^d.
\end{equation}
We denote by
\begin{equation}\label{7.3}
 b_0:=b-\bar{a}_0\in S^0(\mathbb{R}^d).
\end{equation}

We shall define $\Phi_t$ to be the Hamiltonian flow generated by $\bar{a}_0$, i.e.
\begin{equation}\label{7.3-1}
\Phi_t(Y)=\big(x(t;Y),\xi(t;Y)\big),\text{ for }Y=(y,\eta)\in\mathbb{T}^*\mathbb{R}^d
\end{equation}
is the solution of \eqref{1.24} with $\bar{a}_0$ replacing $a$.

The choice of $T>0$ is subject to two constraints. The first constraint comes from the Remarks \ref{R1.29} and \ref{R1.30} leading to the conclusion that if the following estimation is true:
\begin{equation}\label{7.4}
 \underset{t\in(-T,T)}{\sup}\underset{Y\in\mathbb{T}^*\mathbb{R}^d}{\sup}\left\|1_d-\nabla_yx(t;Y)\right\|\,<\,1,
\end{equation}
then there exists a real function $U_t\in C^\infty\big((-T,T)\times\mathbb{R}^{2d}\big)$, solution of the Hamilton-Jacobi equation \eqref{1.29} with $a$ replaced by $\bar{a}_0$, that is a generating function for $\Phi_t$, $\forall t\in(-T,T)$. In fact Lemma \ref{L1.25}.3 implies that \eqref{7.4} is verified for $T>0$ small enough.

The second restriction on $T$ comes from the requirement that $U_t$ should satisfy Hypothesis \ref{Hyp-U} of section 3 for any $t\in(-T,T)$. Due to Lemma \ref{L1.31} this hapens with a possible reduction of $T>0$. Thus, for $T>0$ small enough we may suppose the existence of a real function $d_t\in C^\infty\big((-T,T);S^1(\mathbb{R}^d)\big)$ such that the following two facts are verified:
\begin{equation}\label{7.5}
 U_t(x,\eta)=<x,\eta>+d_t(x,\eta),\quad\forall(x,\eta)\in\mathbb{R}^{2d},\ \forall t\in(-T,T),
\end{equation}
\begin{equation}\label{7.6}
 \exists\delta\in[0,1),\text{ such that } \left\|\nabla^2_{x,\eta}d_t(x,\eta)\right\|\,\leq\,\delta,\quad\forall(x,\eta)\in\mathbb{R}^{2d},\ \forall t\in(-T,T).
\end{equation}

Thus, the function $U_t$ defined as above can be used in the construction of a magnetic FIO associated to $\Phi_t$ and the magnetic field $B$. The symbol $b_t\in C^1\big((-T,T);S^0(\mathbb{R}^d)\big)$ of this operator is looked upon as an asymptotic sum of the form:
\begin{equation}\label{7.7}
 b_t\,\sim\,\sum_{k=0}^\infty b_{t,k},\quad b_{t,k}\in C^1\big((-T,T);S^{-k}(\mathbb{R}^d)\big),\ \forall k\in\mathbb{N}.
\end{equation}
The symbols $b_{t,k}$ will be defined reccurently as solutions of transport equations associated to the following first order differential operator suggested by the product formula in Theorem \ref{T3.6}:
\begin{equation}\label{7.8}
 \forall\varphi\in\mathcal{S}(\mathbb{R}^{2d}),\qquad\mathfrak{L}^B\varphi:=-i\left<\left(\nabla_\xi\bar{a}_0\right)\left(x,\big(\nabla_x U_t\big)(x,\xi)\right),\nabla_x\varphi\right>\,+\,b_0\big(x,\nabla_xU_t(x,\xi)\big)\varphi\,-
\end{equation}
$$
-\,\left\{\left(\frac{i}{2}\right)\Tr\left[\big(\nabla^2_{\xi,\xi}b\big)\big(x,\nabla_xU_t(x,\xi)\big)\cdot\big(\nabla^2_{x,x}U_t\big)(x,\xi)\big)\right]\,-\,\left<M_t(x,\xi)\cdot\big(\nabla_\xi b\big)\big(x,\nabla_xU_t(x,\xi)\big),\big(\nabla_\xi d_t\big)(x,\xi)\right>\right\}\varphi,
$$
where $M_t$ is defined by \eqref{3.136} with $d$ replaced by $d_t$.

The first transport equation verified by $b_{t,0}$ is:
\begin{equation}\label{7.10}
 \left\{
\begin{array}{rcl}
 i\partial_tb_{t,0}&=&\mathfrak{L}^Bb_{t,0},\quad\forall(x,\xi)\in\mathbb{T}^*\mathbb{R}^d,\ \forall t\in(-T,T),\\
\left.b_{t,0}\right|_{t=0}&=&1.
\end{array}
\right.
\end{equation}

Suppose that for some $k\geq1$ we have already determined the family of symbols 
$$
b_{t,j}\in C^1\big((-T,T);S^{-j}(\mathbb{R}^d)\big),\quad0\leq j\leq k-1.
$$
From the proof of Proposition \ref{P3.1} we notice the existence of the functions $c_{t,j}\in C^0\big((-T,T);S^{1-j}(\mathbb{R}^d)\big)$ such that
\begin{equation}\label{7.11}
 \mathfrak{Op}^A(b)\circ\mathfrak{Op}^A_{\Phi_t}(b_{t,j})\,=\,\mathfrak{Op}^A_{\Phi_t}(c_{t,j}),\quad0\leq j\leq k-1.
\end{equation}

Theorem \ref{T3.6} implies that 
\begin{equation}\label{7.12}
 c_{t,j}-\bar{a}_0b_{t,j}-\mathfrak{L}^Bb_{t,j}\,\in\,C\big((-T,T);S^{-j-1}(\mathbb{R}^d)\big),\quad0\leq j\leq k-1.
\end{equation}
We define then $b_{t,k}\in C^1\big((-T,T);S^{-k}(\mathbb{R}^d)\big)$ as the solution of the following transport problem:
\begin{equation}\label{7.13}
  \left\{
\begin{array}{rcl}
 i\partial_tb_{t,k}&=&\mathfrak{L}^Bb_{t,k}\,+\,f_{t,k},\quad\forall(x,\xi)\in\mathbb{T}^*\mathbb{R}^d,\ \forall t\in(-T,T),\\
\left.b_{t,k}\right|_{t=0}&=&0,
\end{array}
\right.
\end{equation}
where
\begin{equation}\label{7.14}
f_{t,k}\,:=\,  c_{t,k-1}-\bar{a}_0b_{t,k-1}-\mathfrak{L}^Bb_{t,k-1}\,\in\,C\big((-T,T);S^{-k}(\mathbb{R}^d)\big).
\end{equation}

We continue this section with the proof of the existence and uniqueness of the solutions of the transport problems \eqref{7.10} and \eqref{7.13}, using the method of the proof of Lemma IV-29 in \cite{Ro}.

\begin{lemma}\label{L7.1}
 For any $k\in\mathbb{N}$ and $h_t\in C\big((-T,T);S^{-k}(\mathbb{R}^d)\big)$ the problem:
\begin{equation}\label{7.15}
  \left\{
\begin{array}{rcl}
 i\partial_tu_t&=&\mathfrak{L}^Bu_t\,+\,h_{t},\quad\forall(x,\xi)\in\mathbb{T}^*\mathbb{R}^d,\ \forall t\in(-T,T),\\
\left.u_t\right|_{t=0}&=&\delta_{0,k},
\end{array}
\right.
\end{equation}
has a unique solution $u_t\in C^1\big((-T,T);S^{-k}(\mathbb{R}^d)\big)$.
\end{lemma}
\begin{proof}
 Let us consider the Hamiltonian system associated to $\bar{a}_0$:
\begin{equation}\label{7.16}
 \left\{
\begin{array}{rcl}
\dot{x}(t;y,\xi)&=&\big(\nabla_{\xi}\bar{a}_0\big)\big(x(t;y,\xi),\eta(t;y,\xi)\big),\\
\dot{\eta}(t;y,\xi)&=&-\big(\nabla_{x}\bar{a}_0\big)\big(x(t;y,\xi),\eta(t;y,\xi)\big),\\
\big(x(0;y,\xi),\eta(0;y,\xi)\big)&=&(y,\xi),\quad\forall(y,\xi)\in\mathbb{T}^*\mathbb{R}^d,
\end{array}
\right.
\end{equation}
that has a unique global solution as proved in subsection \ref{1.f}. Let us denote this solution
$$
\Phi_t(y,\xi)\equiv\big(x(t;y,\xi),\eta(t;y,\xi)\big).
$$

Suppose now that $u_t(x,\xi)$ is a solution of problem \eqref{7.15} and let us denote by
\begin{equation}\label{7.17}
 z(t)\equiv z(t;y,\xi):=u_t\big(x(t;y,\xi),\xi\big),\qquad g(t)\equiv g(t;y,\xi):=h_t\big(x(t;y,\xi),\xi\big).
\end{equation}
Using the first equation in \eqref{7.16} we deduce that
\begin{equation}\label{7.18}
 \dot{z}(t)\,=\,\big(\partial_tu_t\big)\big(x(t;y,\xi),\xi\big)\,+\,\left<\big(\nabla_{x}u_t\big)\big(x(t;y,\xi),\xi\big),\big(\nabla_{\xi}\bar{a}_0\big)\big(x(t;y,\xi),\eta(t;y,\xi)\big)\right>.
\end{equation}
We notice that \eqref{1.30} implies that
\begin{equation}\label{7.19}
 \big(\nabla_{x}U_t\big)\big(x(t;y,\xi),\xi\big)\,=\,\eta(t;y,\xi).
\end{equation}
Considering the results \eqref{7.15}, \eqref{7.8}, \eqref{7.18} and \eqref{7.19} we deduce that
\begin{equation}\label{7.20}
 \dot{z}(t)\,+\,ib_0\big(x(t;y,\xi),\eta(t;y,\xi)\big)z(t)\,+\,\left(\frac{1}{2}\right)\Tr\left[\big(\nabla^2_{\xi,\xi}b\big)\big(x(t;y,\xi),\eta(t;y,\xi)\big)\cdot\big(\nabla^2_{x,x}U_t\big)\big(x(t;y,\xi),\xi\big)\right]z(t)+
\end{equation}
$$
+\,i\left<M_t\big(x(t;y,\xi),\xi\big)\cdot\big(\nabla_\xi b\big)\big(x(t;y,\xi),\eta(t;y,\xi)\big),\big(\nabla_\xi d_t\big)\big(x(t;y,\xi),\xi\big)\right>z(t)\,=\,-ig(t),
$$
$$
z(0)=\delta_{0,k}.
$$
Let us denote by $K(t;y,\xi)$ the coefficient of $z(t)$ in the equation \eqref{7.20}. Using Lemma \ref{L1.25} and relation \eqref{7.5} we deduce that
\begin{equation}\label{7.21}
 K\in C\big((-T,T);S^0(\mathbb{R}^d)\big),\qquad g\in C\big((-T,T);S^{-k}(\mathbb{R}^d)\big).
\end{equation}
Thus it is evident that the problem \eqref{7.20} has a unique solution on the interval $(-T,T)$ and this is explicitely given by:
\begin{equation}\label{7.22}
 z(t;y,\xi)\,=\,-i\int_0^te^{-\int_s^tK(\sigma;y,\xi)d\sigma}g(s;y,\xi)ds\,+\,\delta_{0,k}e^{-\int_0^tK(s;y,\xi)ds}.
\end{equation}
The relations \eqref{7.21} and the fact that $\delta_{0,k}=0$ for $k\geq1$ imply that
\begin{equation}\label{7.23}
 z\,\in\,C^1\big((-T,T);S^{-k}(\mathbb{R}^d)\big).
\end{equation}

As already noticed in the proof of Proposition \ref{P1.28}, the condition \eqref{7.4} implies that for any $t\in(-T,T)$ and any $\xi\in\mathbb{R}^d$ the map $\mathbb{R}^d\ni y\mapsto x(t;y,\xi)\in\mathbb{R}^d$ is a $C^\infty$-diffeomorphism having an inverse denoted by $f(t;x,\xi)$. Moreover, in the end of the proof of Lemma \ref{L1.31} we have obtained the following estimations:
\begin{equation}\label{7.24}
 \partial^\alpha_x\partial^\beta_\xi f\,\in\,C^\infty\big((-T,T);S^{-|\beta|}(\mathbb{R}^d)\big),\quad\forall(\alpha,\beta)\in[\mathbb{N}^d]^2\text{ with }|\alpha+\beta|\geq1.
\end{equation}
Thus if $u_t$ is a solution of problem \eqref{7.15} and $z(t)$ is defined by \eqref{7.17}, then this $z(t)$ is a solution of \eqref{7.20} and moreover
\begin{equation}\label{7.25}
u_t(x,\xi):=z(t;f(t;x,\xi),\xi).
\end{equation}
Thus uniqueness of the solution of \eqref{7.15} follows from the uniqueness of the solution of \eqref{7.20} and from the fact that the first relation in \eqref{7.17} allows to obtain $u_t$ from $z$.

To prove existence of a solution of problem \eqref{7.15} we start with relation \eqref{7.25} with $z$ the solution of problem \eqref{7.20} and proceed to an explicit verification. The initial condition in \eqref{7.15} is evidently satisfied. Relations \eqref{7.23} and \eqref{7.24} imply that
\begin{equation}\label{7.26}
 u_t\,\in\,C^1\big((-T,T);S^{-k}(\mathbb{R}^d)\big).
\end{equation}
In order to verify also the equation in \eqref{7.15} let us notice that
\begin{equation}\label{7.27}
 \big(\partial_tu_t\big)(x,\xi)\,=\,\dot{z}\big(t;f(t;x,\xi),\xi)\,+\,\left<\big(\nabla_{y}z\big)\big(t;f(t;x,\xi),\xi\big),\dot{f}(t;x,\xi)\right>.
\end{equation}
From \eqref{7.20} and \eqref{7.19} we deduce that
\begin{equation}\label{7.28}
  \dot{z}\big(t;f(t;x,\xi),\xi)\,=\,-K\big(t;f(t;x,\xi),\xi)z\big(t;f(t;x,\xi),\xi)\,-\,ih_t(x,\xi)\,=
\end{equation}
$$
=\,-ib_0\big(x,\nabla_{x}U_t(x,\xi)\big)u_t(x,\xi)\,-\,\left(\frac{1}{2}\right)\Tr\left[\big(\nabla^2_{\xi,\xi}b\big)\big(x,\nabla_{x}U_t(x,\xi)\big)\cdot\big(\nabla^2_{x,x}U_t\big)\big(x,\xi\big)\right]u_t(x,\xi)\,-
$$
$$
-\,i\left<M_t\big(x,\xi\big)\cdot\big(\nabla_\xi b\big)\big(x,\nabla_{x}U_t(x,\xi)\big),\big(\nabla_\xi d_t\big)\big(x,\xi\big)\right>u_t(x,\xi)\,-\,ih_t(x,\xi)\,=
$$
$$
=\,-i\big(\mathfrak{L}^Bu_t\big)(x,\xi)\,-\,ih_t(x,\xi)\,+\,\left<\big(\nabla_{\xi}\bar{a}_0\big)\big(x,\nabla_{x}U_t(x,\xi)\big),\nabla_{x}u_t(x,\xi)\right>.
$$
If we differentiate the first equality in \eqref{7.17} with respect to $y$ we obtain:
\begin{equation}\label{7.29}
\big(\nabla_{y}z\big)\big(t;f(t;x,\xi),\xi)\,=\,\left[\nabla_{x}f\big(t;x,\xi)\right]^{-1}\big(\nabla_{x}u_t\big)(x,\xi).
\end{equation}
Differentiating now the equality $x\big(t;f(t;x,\xi),\xi\big)=x$ with respect to $t$ and respectively with respect to $x$ we obtain
$$
\dot{x}\big(t;f(t;x,\xi),\xi\big)\,+\,\big(\partial_{y}x\big)\big(t;f(t;x,\xi),\xi\big)\cdot\dot{f}(t;x,\xi)\,=\,0,
$$
$$
\big(\partial_{y}x\big)\big(t;f(t;x,\xi),\xi)\cdot\big(\partial_{x}f\big)(t;x,\xi)\,=\,1_d.
$$
Taking into account the first equation in \eqref{7.16} and the equality \eqref{7.19}, the above two equalities imply that
\begin{equation}\label{7.30}
 \dot{f}(t;x,\xi)\,=\,-\left[\partial_{x}f\big(t;x,\xi)\right]\cdot\big(\nabla_{\xi}\bar{a}_0\big)\big(x,\nabla_{x}U_t(x,\xi)\big).
\end{equation}
Let us put together the results \eqref{7.27}, \eqref{7.28}, \eqref{7.29} and \eqref{7.30}, as well as the fact that the transpose of the matrix $\partial_{x}f$ is $\nabla_{x}f$, to conclude that:
$$
i\partial_tu_t\,=\,\mathfrak{L}^Bu_t\,+\,h_t.
$$
\end{proof}

The next statement is the first technical step towards the description of the evolution group $e^{-itP}$ as a magnetic FIO.
\begin{proposition}\label{P7.2}
Suppose given a symbol $a$ verifying Hypothesis \ref{Hyp-a} and let $b\in S^1(\mathbb{R}^d)$, $T>0$ and $\Phi_t$ with $t\in(-T,T)$ be as defined at the begining of this section (i.e. the relations \eqref{7.1}-\eqref{7.6} are verified). Then there exist a family $\{b_t\}_{t\in(-T,T)}$ defining a function of class $C^1\big((-T,T);S^0(\mathbb{R}^d)\big)$ and a family $\{r_t\}_{t\in(-T,T)}$ defining a function of class $C\big((-T,T);S^{-\infty}(\mathbb{R}^d)\big)$ such that the operators $F_t:=\mathfrak{Op}^A_{\Phi_t}(b_t)$ and respectively $R_t:=\mathfrak{Op}^A_{\Phi_t}(r_t)$, define functions of class $C^1\big((-T,T);\mathbb{B}\big(\mathcal{S}(\mathbb{R}^d)\big)\big)$ and respectively $C\big((-T,T);\mathbb{B}\big(\mathcal{S}(\mathbb{R}^d)\big)\big)$ and satisfy:
\begin{equation}\label{7.31}
  \left\{
\begin{array}{l}
 i\partial_tF_t\,-\,b^A(x,D)\circ F_t\,=\,R_t,\\
\left.F_t\right|_{t=0}\,=\,1.
\end{array}
\right.
\end{equation}
\end{proposition}
\begin{proof}
 Let us recall the transport equations \eqref{7.10} and \eqref{7.13} and their solutions $b_{t,k}\in C^1\big((-T,T);S^{-k}(\mathbb{R}^d)\big)$, for any $k\in\mathbb{N}$. We denote by $\tilde{b}_t\in C^1\big((-T,T);S^{0}(\mathbb{R}^d)\big)$ an asymptotic sum of the series $\sum_{k\in\mathbb{N}}b_{t,k}$ and by $s_{t,N}:=\sum_{0\leq k\leq N}b_{t,k}$ the partial sums of the same series. Then $\rho_{t,N}:=\tilde{b}_t-s_{t,N}\in C^1\big((-T,T);S^{-(N+1)}(\mathbb{R}^d)\big)$. We shall also use the notations: $\tilde{F}_t:=\mathfrak{Op}^A_{\Phi_t}(\tilde{b}_t)$ and $\tilde{F}_{t,N}=\mathfrak{Op}^A_{\Phi_t}(s_{t,N})$ for any $N\in\mathbb{N}$.

Let us notice the following general fact: Given any $m\in\mathbb{R}$ and any family $\{e_t\}_{t\in(-T,T)}$ defining a function of class $C^1\big((-T,T);S^m(\mathbb{R}^d)\big)$, Theorem \ref{T3.6} and Proposition \ref{P1.28} imply that
\begin{equation}\label{7.32}
 i\partial_t\left[\mathfrak{Op}^A_{\Phi_t}(e_t)\right]\,-\,b^A(x,D)\circ\mathfrak{Op}^A_{\Phi_t}(e_t)\,=\,\mathfrak{Op}^A_{\Phi_t}(g_t),
\end{equation}
with
\begin{equation}\label{7.33}
g_t:=-\big(\partial_tU_t\big)e_t+i\partial_te_t-h_t\in C\big((-T,T);S^m(\mathbb{R}^d)\big), 
\end{equation}
and $h_t\in C\big((-T,T);S^{m+1}(\mathbb{R}^d)\big)$ being the symbol of the magnetic FIO obtained as the product $b^A(x,D)\circ\mathfrak{Op}^A_{\Phi_t}(e_t)$. From \eqref{1.29} we deduce that
\begin{equation}\label{7.34}
 \partial_tU_t(x,\xi)\,=\,-\bar{a}_0\big(x,\nabla_xU_t(x,\xi)\big),
\end{equation}
and from \eqref{3.149} and \eqref{7.8} we get that
\begin{equation}\label{7.35}
 f_t\,:=\,h_t-\bar{a}_0e_t-\mathfrak{L}^Be_t\,\in\,C\big((-T,T);S^{m-1}(\mathbb{R}^d)\big).
\end{equation}
Finally from \eqref{7.33}-\eqref{7.35} it follows that
\begin{equation}\label{7.36}
 g_t\,=\,i\partial_te_t-\mathfrak{L}^Be_t-f_t.
\end{equation}

We shall use the above remark in order to compute the following expresion:
$$
i\partial_t\tilde{F}_{t,N}\,-\,b^A(x,D)\circ \tilde{F}_{t,N}\,=\,
$$
$$
=\,\sum_{k=0}^N\left\{i\partial_t\left[\mathfrak{Op}^A_{\Phi_t}(b_{t,k})\right]-b^A(x,D)\circ \mathfrak{Op}^A_{\Phi_t}(b_{t,k})\right\}\,=\,\sum_{k=0}^N\mathfrak{Op}^A_{\Phi_t}(g_{t,k}),\quad\forall N\in\mathbb{N},
$$
where using \eqref{7.36}, \eqref{7.35}, \eqref{7.14}, \eqref{7.10} and \eqref{7.13} we have
$$
g_{t,k}\,:=\,i\partial_tb_{t,k}-\mathfrak{L}^Bb_{t,k}-f_{t,k+1}=f_{t,k}-f_{t,k+1}.
$$
From the fact that $f_{t,0}=0$ it follows that
\begin{equation}\label{7.37}
 i\partial_t\tilde{F}_{t,N}\,-\,b^A(x,D)\circ \tilde{F}_{t,N}\,=\,-\mathfrak{Op}^A_{\Phi_t}(f_{t,N+1}).
\end{equation}
For any fixed $N\in\mathbb{N}$ we take into account \eqref{7.32} and \eqref{7.37} and make the following computation:
$$
i\partial_t\tilde{F}_{t}\,-\,b^A(x,D)\circ \tilde{F}_{t}\,=
$$
$$
=\,i\partial_t\tilde{F}_{t,N}\,-\,b^A(x,D)\circ \tilde{F}_{t,N}\,+\,i\partial_t\left[\mathfrak{Op}^A_{\Phi_t}(\rho_{t,N})\right]-b^A(x,D)\circ \mathfrak{Op}^A_{\Phi_t}(\rho_{t,N})\,=\,\mathfrak{Op}^A_{\Phi_t}(r_{t,N}),
$$
where $r_{t,N}\in\,C\big((-T,T);S^{-N-1}(\mathbb{R}^d)\big)$. But from \eqref{7.32} we deduce that $i\partial_t\tilde{F}_{t}-b^A(x,D)\circ \tilde{F}_{t}$ is a magnetic FIO with a symbol that does not depend on $N$. We conclude that $r_{t,0}=r_{t,1}=\ldots=r_{t,N}$, $\forall N\in\mathbb{N}$. If we denote by $\tilde{r}_t$ the common value of these symbols $r_{t,N}$, we conclude that $\tilde{r}_t\in\,C\big((-T,T);S^{-\infty}(\mathbb{R}^d)\big)$. Let us denote by $\tilde{R}_t:=\mathfrak{Op}^A_{\Phi_t}(\tilde{r}_t)$ and notice that we have the equality:
\begin{equation}\label{7.38}
 i\partial_t\tilde{F}_{t}\,-\,b^A(x,D)\circ \tilde{F}_{t}\,=\,\tilde{R}_t.
\end{equation}

For $t=0$, for any $N\in\mathbb{N}$ we have that $\tilde{F}_0=\mathfrak{Op}^A_{\Phi_0}(s_{0,N}+\rho_{0,N})$. Considering our initial conditions in \eqref{7.10} and \eqref{7.13} we deduce that $s_{0,N}=1$ for any $N\in\mathbb{N}$. Moreover, $\Phi_0$ is the identity and thus any corresponding magnetic FIO is a magnetic $\Psi$DO. Also, we notice that for any $N\in\mathbb{N}$ the symbol $\rho_{0,N}$ belongs to $S^{-N-1}(\mathbb{R}^d)$ and we deduce as before that there exists a symbol $\tilde{\rho}\in S^{-\infty}(\mathbb{R}^d)$ that satisfies the equality:
\begin{equation}\label{7.39}
 \tilde{F}_0\,=\,1\,+\,\mathfrak{Op}^A(\tilde{\rho}).
\end{equation}

We finish the proof by defining
$$
F_t\,:=\,\tilde{F}_t\,-\,\mathfrak{Op}^A(\tilde{\rho}),
$$
and noticing that due to \eqref{7.38}, \eqref{7.39} and Lemma \ref{L2.9} it is an operator verifying all the conditions of the Proposition.
\end{proof}

The next step is to eliminate the free term $R_t$ in equation \eqref{7.31} using an idea from \cite{Ku} and our Theorem \ref{T1.19}. Let us notice first that, following Proposition \ref{P1.14} and Lemma \ref{L2.9}, our operator $R_t$ is in fact a magnetic $\Psi$DO with a symbol of class $C\big((-T,T);S^{-\infty}(\mathbb{R}^d)\big)$. Next let us use this remark and prove the following technical step.
\begin{definition}\label{D7.3}
 We define the following sequence of magnetic $\Psi$DO of order $-\infty$:
\begin{equation}\label{7.40}
 W_1(t)\,:=\,iR_t,\qquad W_k(t)\,:=\,\int_0^tW_1(t-\tau)\,W_{k-1}(\tau)\,d\tau,\quad \forall k\geq2,\,\forall t\in(-T,T).
\end{equation}
We denote their symbols by $w_k$ (for $k\geq1$).
\end{definition}

\begin{lemma}\label{L7.3}
 With the above notations and definitions, we have that evidently $w_k\in C\big((-T,T);S^{-\infty}(\mathbb{R}^d)\big)$ for any $k\in\mathbb{N}$ and:
\begin{itemize}
 \item the series $\sum_{k=1}^\infty w_k$ converges in $C\big((-T,T);S^{-\infty}(\mathbb{R}^d)\big)$; let $w$ be its sum;
\item the operator $W(t):=\mathfrak{Op}^A(w(t))$ satisfies the equality:
\begin{equation}\label{7.41}
 W(t)\,=\,W_1(t)\,+\,\int_0^tW_1(t-\tau)\,W(\tau)\,d\tau,\quad\forall t\in(-T,T).
\end{equation}
\end{itemize}
\end{lemma}
\begin{proof}
 From Remark \ref{R1.2} we deduce that $S^{-\infty}(\mathbb{R}^d)=S^{-\infty}_0(\mathbb{R}^d)$ and thus for the first statement it is enough that we prove that $\sum_{k=1}^\infty w_k$ converges in $C\big((-T,T);S^{-\infty}_0(\mathbb{R}^d)\big)$. 

Let us notice that for any $k\geq2$ we have in fact the formula:
\begin{equation}\label{7.42}
 w_k(t)\,=\,\int_0^t\int_0^{t_1}\ldots\int_0^{t_{k-2}}P_k(t,t_1,\ldots,t_{k-1})\,dt_{k-1}\ldots dt_1,\ t_0:=t,
\end{equation}
where
\begin{equation}\label{7.43}
 P_k(t,t_1,\ldots,t_{k-1}):=w_1(t-t_1)\sharp^Bw_1(t_1-t_2)\sharp^B\cdots\sharp^Bw_1(t_{k-2}-t_{k-1})\sharp^Bw_1(t_{k-1}).
\end{equation}

Let us fix some $m\leq0$ and notice that $w_1(t)\in S^{-\infty}(\mathbb{R}^d)\subset S^{m}_0(\mathbb{R}^d)\subset S^{0}_0(\mathbb{R}^d)$ for any $t\in(-T,T)$ and thus using the magnetic composition of symbols theorem we deduce that $P_k(t,t_1,\ldots,t_{k-1})\in S^{m}_0(\mathbb{R}^d)$. Moreover, denoting by $t_0:=t$ and $t_k=0$ and using the notations from Theorem \ref{T1.19} we notice that
$$
p_{-m}\sharp^B\text{\rm ad}^B_{X_1}\cdots\text{\rm ad}^B_{X_N}[P_k]
$$
is a sum of $k^N$ terms of the form
$$
p_{-m}\sharp^B\prod_{l=0}^{k-1}\text{\rm ad}^B_{X_{j^l_1}}\cdots\text{\rm ad}^B_{X_{j^l_q}}\big[w_1(t_l-t_{l+1})\big],
$$
where $\{j^l_1,\ldots,j^l_q\}$ is any subset of $\{1,2,\ldots,N\}$, with $0\leq q \leq N$, and the product is understood as a magnetic Moyal product $\sharp^B$. Now we fix some $T_0\in(0,T)$ and suppose (just for fixing the ideas) that $t\in[0,T_0]$, so that in this case we have $0=t_k\leq t_{k-1}\leq\ldots\leq t_1\leq t_0=t\leq T_0$. We obtain then the inequality
$$
\left\|P_k(t,t_1,\ldots,t_{k-1})\right\|_{m,X_1,\ldots,X_N}\leq k^N\underset{|\tau|\leq T_0}{\sup}\|w_1(\tau)\|_{m,X_1,\ldots,X_N}\left[\underset{|\tau|\leq T_0}{\sup}\|w_1(\tau)\|_{0,X_1,\ldots,X_N}\right]^{k-1}.
$$
Denoting the paranthesis by $M$ and using the equality \eqref{7.42} we deduce that
\begin{equation}\label{7.44}
\underset{0\leq t\leq T_0}{\sup}\|w_k(t)\|_{m,X_1,\ldots,X_N}\leq \underset{|\tau|\leq T_0}{\sup}\|w_1(\tau)\|_{m,X_1,\ldots,X_N}\frac{k^NT_0^{k-1}M^{k-1}}{(k-1)!}.
\end{equation}
This inequality, and the similar one associated to the interval $[-T_0,0]$ imply evidently the convergence of the sum $\sum_{k=1}^\infty w_k(t)$ in $C\big([-T_0,T_0];S^{m}_0(\mathbb{R}^d)\big)$, for any $m\leq0$ and any $T_0\in(0,T)$. In conclusion the series also converges in $C\big((-T,T);S^{-\infty}_0(\mathbb{R}^d)\big)$. By definition we have that for any $q\geq1$:
$$
\sum_{k=1}^{q+1}w_k(t)\,=\,w_1(t)\,+\,\int_0^tw_1(t-\tau)\sharp^B\left[\sum_{k=1}^qw_k(\tau)\right]\,d\tau,\quad\forall t\in(-T,T).
$$
Letting $q\rightarrow\infty$ we obtain that
\begin{equation}\label{7.45}
 w(t)\,=\,w_1(t)\,+\,\int_0^tw_1(t-\tau)\sharp^Bw(\tau)\,d\tau,\quad\forall t\in(-T,T),
\end{equation}
and this implies \eqref{7.41}.
\end{proof}

The next statement is the main result of this section and also the most important result of our paper.
\begin{theorem}\label{T7.4}
 Suppose given a symbol $a\in S^1(\mathbb{R}^d)$ verifying Hypothesis \ref{Hyp-a} and let $P$ be the self-adjoint operator in $L^2(\mathbb{R}^d)$ defined by $\mathfrak{Op}^A(a)$ and $\Phi_t$ for $t\in\mathbb{R}$ be the Hamiltonian flow defined by $\bar{a}_0$, the principal part of $a$ defined by \eqref{7.2}. Then there exist $T>0$ and $g_t\in C^1\big((-T,T);S^0(\mathbb{R}^d)\big)$ such that:
\begin{enumerate}
 \item $G_t\,:=\,e^{-itP}\,=\,\mathfrak{Op}^A_{\Phi_t}(g_t),\qquad\forall t\in(-T,T)$.
\item $g_t\,\sim\,\sum_{k=0}^\infty b_{t,k}$, where $b_{t,k}\in C^1\big((-T,T);S^{-k}(\mathbb{R}^d)\big)$, $\forall k\in\mathbb{N}$, are the solutions of the transport equations \eqref{7.10} and \eqref{7.13}.
\end{enumerate}
\end{theorem}
\begin{proof}
 The positive parameter $T$ is chosen as in the begining of this section, limited by the same restrictions (see \eqref{7.4}-\eqref{7.6}). Let  $F_t$ and $R_t$ be the operators defined in Proposition \ref{P7.2} and $W(t)$ the family of operators from Lemma \ref{L7.3}; we also have then that $F_t\in C^1\big((-T,T);\mathbb{B}\big(\mathcal{S}(\mathbb{R}^d)\big)\big)$ and the operators $R_t$ and $W(t)$ belong to $C\big((-T,T);\mathbb{B}\big(\mathcal{S}(\mathbb{R}^d)\big)\big)$. Moreover, in Lemma 7.40 from \cite{IMP2} it was proved that $G_t\in C^1\big((-T,T);\mathbb{B}\big(\mathcal{S}(\mathbb{R}^d)\big)\big)$. Let us denote by:
\begin{equation}\label{7.46}
 H(t)\,:=\,F_t\,+\,\int_0^tF_{t-\tau}W(\tau)\,d\tau\,\in\,C^1\big((-T,T);\mathbb{B}\big(\mathcal{S}(\mathbb{R}^d)\big)\big).
\end{equation}
We consider \eqref{7.31} and \eqref{7.41} and deduce that
$$
i\partial_tH(t)-b^A(x,D)\circ H(t)\,=\,i\partial_tF_t-b^A(x,D)\circ F_t\,+\,iW(t)\,+\,\int_0^t\left[\left(i\partial_t-b^A(x,D)\right)F_{t-\tau}\right]W(\tau)\,d\tau=
$$
\begin{equation}\label{7.46a}
 =\,R_t\,+\,iW(t)\,+\,\int_0^tR_{t-\tau}W(\tau)\,d\tau\,=\,-iW(t)+iW(t)\,=\,0.
\end{equation}
In conclusion we have
\begin{equation}\label{7.47}
 i\partial_tH(t)-b^A(x,D)\circ H(t)=0,\ \forall t\in(-T,T),\quad H(0)=1.
\end{equation}
But we also have that
\begin{equation}\label{7.48}
 \partial_tG_t\,=\,-ib^A(x,D)G_t,\ \forall t\in(-T,T),\quad G_0=1.
\end{equation}
We denote by
$$
S_t\,:=\,G_{-t}H(t)\,-\,1\,\in\,C^1\big((-T,T);\mathbb{B}\big(\mathcal{S}(\mathbb{R}^d)\big)\big).
$$
Due to the fact that $P$ and $G_t$ commute, we deduce that
$$
\partial_tS_t\,=\,ib^A(x,D)G_{-t}H(t)-iG_{-t}b^A(x,D)H(t)\,=\,0
$$
and $S_0=0$. We conclude that $S_t=0$ and
\begin{equation}\label{7.49}
 G_t\,=\,H(t)\,=\,F_t\,+\,\int_0^tF_{t-\tau}W(\tau)\,d\tau,\quad\forall t\in(-T,T).
\end{equation}
Taking into account that $W(\tau)=\mathfrak{Op}^A(w(\tau))$ with $w\in C\big((-T,T);S^{-\infty}(\mathbb{R}^d)\big)$ and $F_{t-\tau}=\mathfrak{Op}^A_{\Phi_{t-\tau}}(b_{t-\tau})$ with $b_t\in  C^1\big((-T,T);S^{0}(\mathbb{R}^d)\big)$ we deduce that $F_{t-\tau}W(\tau)=\mathfrak{Op}^A_{\Phi_{t-\tau}}(\alpha_{t,\tau})$ where the map
$$
\left\{(t,\tau)\in\mathbb{R}^2;|\tau|\leq|t|\right\}\ni(t,\tau)\mapsto\alpha_{t,\tau}\in S^{-\infty}(\mathbb{R}^d)
$$
is continuous and continuously differentiable with respect to $t$. Thus, using Lemma \ref{L2.9} we deduce that it exists a family of symbols $\{\beta_{t,\tau}\}$ having the same properties as $\alpha_{t,\tau}$ and such that
$$
F_{t-\tau}W(\tau)=\mathfrak{Op}^A_{\Phi_{t-\tau}}(\alpha_{t,\tau})=\mathfrak{Op}^A(\beta_{t,\tau}).
$$
Thus the integral in \eqref{7.49} is of the form $\mathfrak{Op}^A(\gamma_t)$ with $\gamma_t\in C^1\big((-T,T);S^{-\infty}(\mathbb{R}^d)\big)$ and using Lemma \ref{L2.9} we deduce that it is in fact of the form $\mathfrak{Op}^A_{\Phi_t}(\delta_t)$ with $\delta_t\in C^1\big((-T,T);S^{-\infty}(\mathbb{R}^d)\big)$. In conclusion we proved that
$$
G_t=\mathfrak{Op}^A_{\Phi_t}(g_t),\quad\text{with }g_t:=b_t+\delta_t\in C^1\big((-T,T);S^{0}(\mathbb{R}^d)\big)
$$
\end{proof}

\begin{remark}\label{R7.5}
 Let us suppose that for a symbol $a$ that satisfies Hypothesis \ref{Hyp-a}, the homogeneous principal part $a_0$ does only depend on the variable $\xi$. Then $\bar{a}_0=\chi a_0$ also depends only on the variable $\xi$. As noticed in Remarks \ref{R1.26} and \ref{R1.29}, in this case we can take $T=\infty$ and $\Phi_t(y,\eta)=<y+t\nabla\bar{a}_0(\eta),\eta>$, with a generating function $U_t(x,\eta)=<x,\eta>-t\bar{a}_0(\eta)$.

This situation appears for both operators that may be considered as {\it relativistic Schr\"{o}dinger operators}:
\begin{enumerate}
 \item $P:=\mathfrak{Op}^A(<\xi>)$ i.e. $a(\xi)=<\xi>$ and $a_0(\xi)=|\xi|$.
\item $P:=\sqrt{\mathfrak{Op}^A(<\xi>^2)}$. In this case Theorem 6.33 in \cite{IMP2} implies the existence of $a\in S^1(\mathbb{R}^d)$ real, such that $P=\mathfrak{Op}^A(a)$. Then we have $P^2=\mathfrak{Op}^A(<\xi>^2)$ and the Theorem 2.2 in \cite{IMP1} shows that $b:=a^2-<\xi>^2\in S^1(\mathbb{R}^d)$. There exists $R>0$ such that $a^2\geq(1/4)<\xi>^2$ for any $\xi\in\mathbb{R}^d$ with $<\xi>\geq R$. As $d\geq 2$ it follows then that $a$ has a constant sign and thus we have either $a\geq(1/2)<\xi>$ or $-a\geq(1/2)<\xi>$  for any $\xi\in\mathbb{R}^d$ with $<\xi>\geq R$. Due to the G\aa rding inequality for magnetic $\Psi$DO (Theorem 4.3 in \cite{IMP1}; see also Corollary 4.4), in the second case the operator $P\geq0$ would be bounded! As $a$ is bounded for $<\xi>\leq R$, it exists $M>0$ such that $a+M\geq(1/2)<\xi>$. From the equality $a^2-<\xi>^2=b$ we deduce that
$$
a(x,\xi)\,-\,<\xi>\,=\,\frac{b(x,\xi)+M(a(x,\xi)-<\xi>)}{a(x,\xi)+M+<\xi>}\,\in\,S^0(\mathbb{R}^d).
$$
We conclude that as in the first case above the operator $P$ admites as principal symbol $<\xi>$ and thus as homogeneous principal part $|\xi|$.
\end{enumerate}
\end{remark}

\section{Applications}

\subsection{Continuity of a magnetic FIO between magnetic Sobolev spaces}

\begin{theorem}\label{T8.1}
 Suppose given a symplectomorphism $\Phi:\mathbb{T}^*\mathbb{R}^d\rightarrow\mathbb{T}^*\mathbb{R}^d$ defined by a generating function $U$ that satisfies Hypothesis \ref{Hyp-U}. Then for any $a\in S^m(\mathbb{R}^d)$ and for any $s\in\mathbb{R}$ we have that
$$
\mathfrak{Op}^A_{\Phi}(a)\in\mathbb{B}\big(\mathcal{H}^{s+m}_A;\mathcal{H}^s_A\big).
$$
\end{theorem}
\begin{proof}
 Let $Q$ be the self-adjoint operator in $L^2(\mathbb{R}^d)$ associated by Lemma \ref{L1.18}.2 to $\mathfrak{Op}^A(<\xi>)$. We can now either adapt an argument from \cite{Ic} or use a result in \cite{IMP3} stating that $e^{-t(Q-1)}$ is a contraction for any $t\geq0$ so that $Q\geq1$ follows. Due to Theorem 6.33 in \cite{IMP2} it follows that for any $s\in\mathbb{R}$ there exists a symbol $q_s\in S^s(\mathbb{R}^d)$ such that $Q^s=\mathfrak{Op}^A(q_s)$. Using Proposition \ref{P1.14}.1 we deduce the existence of a symbol $r_s\in S^s(\mathbb{R}^d)$ such that $Q^s=r_s^A(x,D)$. Let us introduce the notation $M:=Q^s\mathfrak{Op}^A_{\Phi}(a)Q^{-s-m}$; using theorems \ref{T3.6} and \ref{T4.6} we deduce the existence of a symbol $b\in S^0(\mathbb{R}^d)$ such that $M=\mathfrak{Op}^A_{\Phi}(b)$. From Theorem \ref{T6.6} it follows then the existence of a symbol $c\in S^0(\mathbb{R}^d)$ such that $M^*M=c^A(x,D)$. We use now Proposition \ref{P1.14}.2 to deduce the existence of a symbol $e\in S^0(\mathbb{R}^d)$ such that $M^*M=\mathfrak{Op}^A(e)$. Using the Calderon-Vaillancourt type Theorem (Theorem 3.1 in \cite{IMP1}) we see that $M^*M\in\mathbb{B}\big(L^2(\mathbb{R}^d)\big)$. But that means that $M\in\mathbb{B}\big(L^2(\mathbb{R}^d)\big)$, and Lemma \ref{L1.18}.1 and the equality $\mathfrak{Op}^A_{\Phi}(a)=Q^{-s}MQ^{s+m}$ imply the result of the Theorem.
\end{proof}

\subsection{An Egorov type theorem}\label{Eg-T}

\begin{theorem}\label{T8.2}
 Suppose given a symplectomorphism $\Phi:\mathbb{T}^*\mathbb{R}^d\rightarrow\mathbb{T}^*\mathbb{R}^d$ defined by a generating function $U$ that satisfies Hypothesis \ref{Hyp-U}. Let $s\in\mathbb{R}$ and let $b_\pm\in S^{\pm s}(\mathbb{R}^d)$ be two symbols such that there exists a symbol $e\in S^{-1}(\mathbb{R}^d)$ satisfying
\begin{equation}\label{8.2}
 \mathfrak{Op}^A_{\Phi}(b_+)\circ\left[\mathfrak{Op}^A_{\Phi}(b_-)\right]^*\,-\,1\,=\,e^A(x,D),
\end{equation}
i.e. the formal adjoint of $\mathfrak{Op}^A_{\Phi}(b_-)$ is an approximate right inverse of $\mathfrak{Op}^A_{\Phi}(b_+)$. Then, for any $m\in \mathbb{R}$ and any $a\in S^m(\mathbb{R}^d)$ there exists a symbol $\tilde{a}\in S^m(\mathbb{R}^d)$ such that
\begin{equation}\label{8.3}
 \mathfrak{Op}^A_{\Phi}(b_+)\circ a^A(x,D)\circ\left[\mathfrak{Op}^A_{\Phi}(b_-)\right]^*\,=\,\tilde{a}(x,D),
\end{equation}
\begin{equation}\label{8.4}
 \tilde{a}\,-\,a\circ\Phi^{-1}\,\in\,S^{m-1}(\mathbb{R}^d).
\end{equation}
\end{theorem}
\begin{proof}
 The equality \eqref{8.3} follows from the Theorems \ref{T4.6} and \ref{T5.7} from which we also deduce that if we denote by $\lambda(x,0,\xi)$ the inverse of the $C^\infty$-diffeomorphism
$$
\mathbb{R}^d\ni\xi\mapsto\nabla_xU(x,\xi)\in\mathbb{R}^d,
$$
the following expression
\begin{equation}\label{8.5}
 \tilde{a}(x,\xi)\,-\,b_+\big(x,\lambda(x,0,\xi)\big)a\big((\nabla_\xi U)\big(x,\lambda(x,0,\xi)\big),\lambda(x,0,\xi)\big)\overline{\,b_-\big(x,\lambda(x,0,\xi)\big)}\left|\det\big(\nabla^2_{\xi,x}U\big)\big(x,\lambda(x,0,\xi)\big)\right|^{-1},
\end{equation}
defines a symbol of class $S^{m-1}(\mathbb{R}^d)$.

We use \eqref{8.5} with $a=1$ and \eqref{8.2} to deduce that
\begin{equation}\label{8.6}
 1\,-\,\,b_+\big(x,\lambda(x,0,\xi)\big)\overline{\,b_-\big(x,\lambda(x,0,\xi)\big)}\left|\det\big(\nabla^2_{\xi,x}U\big)\big(x,\lambda(x,0,\xi)\big)\right|^{-1}\,\in\,S^{-1}(\mathbb{R}^d).
\end{equation}
Now from \eqref{8.5} and \eqref{8.6} we conclude that
\begin{equation}\label{8.7}
 \tilde{a}(x,\xi)\,-\,a\big((\nabla_\xi U)\big(x,\lambda(x,0,\xi)\big),\lambda(x,0,\xi)\big)\,\in\,S^{m-1}(\mathbb{R}^d).
\end{equation}
We notice that by definition $\eta=\lambda(x,0,\xi)$ if and only if $\xi=\nabla_x U(x,\eta)$ so that the following equality
$$
\Phi^{-1}\big(x,\nabla_x U(x,\eta)\big)=\big(\nabla_{\eta}U(x,\eta),\eta\big)
$$
and \eqref{8.7} imply the conclusion \eqref{8.4}.
\end{proof}

\subsection{An estimation for the distribution kernel of the evolution group}\label{Ev-Gr}

We consider the assumptions from Section \ref{Ev-group}. We consider a symbol $a\in S^1(\mathbb{R}^d)$ satisfying Hypothesis \ref{Hyp-a}; thus $a$ has a homogeneous principal part $a_0$. Let $\bar{a}_0$ be defined by \eqref{7.2}. We denote by 
$$
\Phi_t(y,\eta):=\big(x(t;y,\eta),\xi(t;y,\eta)\big), \quad\text{respectively}\quad\Phi^0_t(y,\eta):=\big(x_0(t;y,\eta),\xi_0(t;y,\eta)\big)
$$
the Hamiltonian flows defined by $\bar{a}_0$, respectively by $a_0$, whose properties have been described earlier in Lemmas \ref{L1.25} and respectively \ref{L1.27}. We choose $T>0$ in agreement with Remark \ref{R1.30} and Proposition \ref{P1.32}, so that for any $t\in(-T,T)$ the flows $\Phi_t$ and $\Phi^0_t$ are generated by the generating functions $U_t$ and respectively $U^0_t$, that satisfy Hypothesis \ref{Hyp-U}.

We shall use the following notations:
\begin{equation}\label{8.8}
 d(t;x,y)\,:=\,\underset{|\eta|=1}{\inf}\left|x-x_0(t;y,\eta)\right|,
\end{equation}
\begin{equation}\label{8.9}
 D_\epsilon\,:=\,\left\{(t,x,y)\in(-T,T)\times\mathbb{R}^d\times\mathbb{R}^d\ ;\ d(t;x,y)>\epsilon\right\},\quad\epsilon>0,
\end{equation}
\begin{equation}\label{8.10}
 D\,:=\,\underset{\epsilon>0}{\bigcup}D_\epsilon\,=\,\left\{(t,x,y)\in(-T,T)\times\mathbb{R}^d\times\mathbb{R}^d\ ;\ x_0(t;y,\eta)\ne x\ \forall\eta\in\mathbb{R}^d,\,|\eta|=1\right\}.
\end{equation}
Evidently the sets $D_\epsilon$ and $D$ are open subsets of $(-T,T)\times\mathbb{R}^d\times\mathbb{R}^d$.

From Proposition \ref{P1.32}.5 we deduce that
\begin{equation}\label{8.11}
 U^0_t(x,\eta)\,=\,<x,\eta>+d^0_t(x,\eta),\quad\forall t\in(-T,T),\,\forall(x,\eta)\in\mathbb{T}^*\mathbb{R}^d,
\end{equation}
where $d^0_t\in C^\infty\big((-T,T)\times\mathbb{R}^d\times(\mathbb{R}^d\setminus\{0\})\big)$ is positive homogeneous of degree 1 in $\eta$ and for any $T_0\in(0,T)$ and any $k\in\mathbb{N}$ and $(\alpha,\beta)\in[\mathbb{N}^d]^2$ the function $\partial^k_t\partial^\alpha_x\partial^\beta_\eta d^0_t$ is bounded on $[-T_0,T_0]\times\mathbb{R}^d\times\{|\eta|=1\}$.

Let us still denote by
\begin{equation}\label{8.12}
 \varphi_t(x,y,\eta)\,:=\,U^0_t(x,\eta)-<y,\eta>\,=\,<x-y,\eta>+d^0_t(x,\eta),
\end{equation}
\begin{equation}\label{8.13}
 d_{\varphi}(t;x,y)\,:=\,\underset{|\eta|=1}{\inf}\left|\nabla_{\eta}\varphi_t(x,y,\eta)\right|.
\end{equation}

\begin{lemma}\label{L8.3}
 There exists a constant $C\geq1$ such that
$$
C^{-1}d(t;x,y)\,\leq\,d_{\varphi}(t;x,y)\,\leq\,Cd(t;x,y),\quad\forall t\in(-T,T),\,\forall(x,y)\in\big[\mathbb{R}^d\big]^2.
$$
\end{lemma}
\begin{proof}
 The choice of $T>0$ is restricted by a condition of type \eqref{7.4} for $x(t;y,\eta)$ and for $x_0(t;y,\eta)$, i.e. it has to exist $\delta\in[0,1)$ such that
\begin{equation}\label{8.14}
 \left\|1_d\,-\,\nabla_y x(t;y,\eta)\right\|\,\leq\,\delta,\quad\forall t\in(-T,T),\,\forall(y,\eta)\in\mathbb{T}^*\mathbb{R}^d,
\end{equation}
\begin{equation}\label{8.15}
 \left\|1_d\,-\,\nabla_y x_0(t;y,\eta)\right\|\,\leq\,\delta,\quad\forall t\in(-T,T),\,\forall(y,\eta)\in\mathbb{R}^d\times\big(\mathbb{R}^d\setminus\{0\}\big).
\end{equation}
From this last inequality it follows that
\begin{equation}\label{8.16}
 1-\delta\,\leq\,\left\|\nabla_y x_0(t;y,\eta)\right\|\,\leq\,1+\delta,\quad\forall t\in(-T,T),\,\forall(y,\eta)\in\mathbb{R}^d\times\big(\mathbb{R}^d\setminus\{0\}\big).
\end{equation}
Let us denote by $y=f_0(t;x,\eta)$ the unique solution of the equation $x_0(t;y,\eta)=x$, for $\eta\ne0$; thus the map $\mathbb{R}^d\ni x\mapsto f_0(t;x,\eta)\in\mathbb{R}^d$ is the inverse of the $C^\infty$-diffeomorphism $\mathbb{R}^d\ni y\mapsto x_0(t;y,\eta)\in\mathbb{R}^d$, $\forall t\in(-T,T)$ and for $\eta\ne0$. We deduce that $f_0\in C^\infty\big((-T,T)\times\mathbb{R}^d\times(\mathbb{R}^d\setminus\{0\})\big)$ and $f_0$ is positive homogeneous of degree 0 with respect to $\eta$. Noticing that
$$
\partial_xf_0(t;x,\eta)\,=\,\left[\big(\partial_yx_0\big)\big(t;f_0(t;x,\eta),\eta\big)\right]^{-1},
$$
and using \eqref{8.16} we deduce that
\begin{equation}\label{8.17}
 \frac{1}{1+\delta}\,\leq\,\left\|\partial_xf_0(t;x,\eta)\right\|\,\leq\,\frac{1}{1-\delta},\quad\forall t\in(-T,T),\,\forall(x,\eta)\in\mathbb{R}^d\times\big(\mathbb{R}^d\setminus\{0\}\big).
\end{equation}

From Proposition \ref{P1.32}.4 it follows that
\begin{equation}\label{8.18}
 \nabla_\eta\varphi_t(x,y,\eta)\,=\,\nabla_\eta U^0_t(x,\eta)-y\,=\,f_0(t;x,\eta)-y.
\end{equation}
But $y=f_0\big(t;x_0(t;y,\eta),\eta\big)$, so that we have
$$
\nabla_\eta\varphi_t(x,y,\eta)\,=\,\left[\int_0^1\big(\partial_xf_0\big)\big(t;sx+(1-s)x_0(t;y,\eta),\eta\big)ds\right]\cdot\big(x-x_0(t;y,\eta)\big).
$$
Using the second inequality in \eqref{8.17} we conclude that
\begin{equation}\label{8.19}
\left|\nabla_\eta\varphi_t(x,y,\eta)\right|\,\leq\,\frac{1}{1-\delta}\left|x-x_0(t;y,\eta)\right|.
\end{equation}

On the other side, we have the equality $x=x_0\big(t;f_0(t;x,\eta),\eta\big)$, so that
$$
x-x_0(t;y,\eta)\,=\,x_0\big(t;f_0(t;x,\eta),\eta\big)-x_0(t;y,\eta)\,=\,\left[\int_0^1\big(\partial_yx_0\big(t;sf_0(t;x,\eta)+(1-s)y,\eta\big)ds\right]\cdot\big(f_0(t;x,\eta)-y)\big).
$$

We use \eqref{8.18} and the second inequality in \eqref{8.16} to deduce that
\begin{equation}\label{8.20}
 \left|\nabla_\eta\varphi_t(x,y,\eta)\right|\,\geq\,\frac{1}{1+\delta}\left|x-x_0(t;y,\eta)\right|.
\end{equation}
The inequalities \eqref{8.19} and \eqref{8.20} allow us to finish the proof.
\end{proof}
\begin{theorem}\label{T8.4}
 Let $K_t(x,y)$ be the distribution kernel of the operator $G_t:=e^{-itP}$ for some $t\in(-T,T)$, as defined in Section \ref{Ev-group}. The following properties are true:
\begin{enumerate}
 \item For $j\in\{0,1\}$ and for any $(\alpha,\beta)\in[\mathbb{N}^d]^2$, the distributions $\partial^j_t\partial^\alpha_x\partial^\beta_yK_t$ are continuous functions on the domain $D$ defined in \eqref{8.10}.
\item There exists $q\in\mathbb{N}$ such that for any $T_0\in(0,T)$ and any $\epsilon>0$ there exists a constant $C>0$ such that the following estimation is verified:
\begin{equation}\label{8.21}
 \left|K_t(x,y)\right|\,\leq\,Cd(t;x,y)^{-q},\quad\forall(t,x,y)\in D\setminus D_\epsilon,\ |t|\leq T_0.
\end{equation}
\item For any $k\in\mathbb{N}$, any $T_0\in(0,T)$ and any $\epsilon>0$ there exists a constant $C>0$ such that the following estimation is verified:
\begin{equation}\label{8.22}
  \left|K_t(x,y)\right|\,\leq\,Cd(t;x,y)^{-k},\quad\forall(t,x,y)\in D_\epsilon,\ |t|\leq T_0.
\end{equation}
\item All the derivatives of $\omega^B(y,x)K_t(x,y)$ of the type considered in point (1) of this Theorem verify estimations of the type \eqref{8.21} and \eqref{8.22}.
\end{enumerate}
\end{theorem}
\begin{proof}
 From Theorem \ref{T7.4} we know that $G_t=\mathfrak{Op}^A_{\Phi_t}(g_t)$ with a symbol $g_t\in C^1\big((-T,T);S^0(\mathbb{R}^d)\big)$. Thus we have that
\begin{equation}\label{8.23}
 K_t(x,y)\,=\,\omega^B(x,y)L_t(x,y),
\end{equation}
\begin{equation}\label{8.24}
 L_t(x,y)\,:=\,\int_{\mathbb{R}^d}e^{i\big(U_t(x,\eta)-<y,\eta>\big)}g_t(x,\eta)\,\dbar\eta,
\end{equation}
as an oscillating integral defining a distribution of class $\mathcal{S}^\prime(\mathbb{R}^{2d})$. Considering once again the function $\chi_\rho$ introduced at the begining of section \ref{Ev-group} (with $\rho>0$ to be fixed later), we decompose the integral defining $L_t$ as a sum of the two terms $L^0_t$ and $L^\infty_t$ defined by the two symbols $g^0_t:=\chi_\rho g_t$ and $g^\infty_t:=(1-\chi_\rho)g_t$.

Let us consider first 
$$
L^\infty_t(x,y)\,:=\,\int_{\mathbb{R}^d}e^{i<x-y,\eta>}e^{id_t(x,\eta)}(1-\chi_\rho(\eta))g_t(x,\eta)\,\dbar\eta
$$
and notice that the factor $r_t:=e^{id_t}(1-\chi_\rho)g_t$ is a function of class $C^1\big((-T,T);S^{-\infty}(\mathbb{R}^d)\big)$. We conclude that the integral above defines a function $L^\infty_t$ of class $C^1\big((-T,T);C^\infty(\mathbb{R}^{2d}\big)$. Integrating then by parts using the identity $(x-y)e^{i<x-y,\eta>}=-i\nabla_\eta\big(e^{i<x-y,\eta>}\big)$, we obtain that for any $\alpha\in\mathbb{N}^d$ and any $T_0\in(0,T)$ the functions $(x-y)^\alpha L^\infty_t(x,y)$ are of class $L^\infty\big((-T_0,T_0)\times\mathbb{R}^{2d}\big)$.

For the second term let us notice that $g^0_t(x,\eta)=0$ for $|\eta|\leq\rho$ and that Lemma \ref{L1.34} implies that there exists $R>0$ such that $U_t(x,\eta)=U^0_t(x,\eta)$ for any $t\in(-T,T)$ and any $(x,\eta)\in\mathbb{R}^{2d}$ with $|\eta|\geq R$. Thus, choosing $\rho=R$ we can replace $U_t$ by $U^0_t$ in the integral defining $L^0_t$:
\begin{equation}\label{8.26}
 L^0_t(x,y)\,:=\,\int_{\mathbb{R}^d}e^{i\varphi_t(x,y,\eta)}g^0_t(x,\eta)\,\dbar\eta,
\end{equation}
where $\varphi_t$ is defined by \eqref{8.12}, $g^0_t$ is of class $C^1\big((-T,T);S^0(\mathbb{R}^d)\big)$ and $g^0_t(x,\eta)=0$ for $|\eta|\leq R$. The continuity on $D$ of the derivatives $\partial^j_t\partial^\alpha_x\partial^\beta_yL^0_t$, for $j\in\{0,1\}$ and for any $(\alpha,\beta)\in[\mathbb{N}^d]^2$ is now a standard result. In fact, for $(t,x,y)\in D$ we integrate by parts in \eqref{8.26} using the operator
\begin{equation}\label{8.27}
 M\,:=\,-i\frac{<\nabla_\eta\varphi_t,\nabla_\eta>}{|\nabla_\eta\varphi_t|^2},
\end{equation}
that is well-defined due to the result of Lemma \ref{L8.3} and satisfies the identity: $Me^{i\varphi_t}=e^{i\varphi_t}$. The formal adjoint of $M$ is given by:
\begin{equation}\label{8.28}
 ^tM\,:=\,|\nabla_\eta\varphi_t|^{-1}(u+<v,\nabla_\eta>),
\end{equation}
with the functions $u$ and $v$ being continuous functions on $D\times\big(\mathbb{R}^d\setminus\{0\}\big)$ and satsifying the following properties:
\begin{itemize}
 \item For $j\in\{0,1\}$ and for any $(\alpha,\beta,\gamma)\in[\mathbb{N}^d]^3$ the functions $\partial^j_t\partial^\alpha_x\partial^\beta_y\partial^\gamma_\eta u$ and $\partial^j_t\partial^\alpha_x\partial^\beta_y\partial^\gamma_\eta v$ are continuous functions on the domain $D\times\big(\mathbb{R}^d\setminus\{0\}\big)$.
\item For $j\in\{0,1\}$, for any $(\alpha,\beta,\gamma)\in[\mathbb{N}^d]^3$ and for any $T_0\in(0,T)$ there exist a constant $C>0$ and $n\in\mathbb{N}$ such that:
\begin{equation}\label{8.29}
 |\eta|^{1+|\gamma|}\left|\partial^j_t\partial^\alpha_x\partial^\beta_y\partial^\gamma_\eta u(t,x,y,\eta)\right|+|\eta|^{|\gamma|}\left|\partial^j_t\partial^\alpha_x\partial^\beta_y\partial^\gamma_\eta v(t,x,y,\eta)\right|\,\leq\,C\sum_{l=0}^n|\nabla_\eta\varphi_t(x,y,\eta)|^{-l}
\end{equation}
for $(t,x,y)\in D$ with $|t|\leq T_0$, $|\eta|\geq R$. For a suitable choice of $k\in\mathbb{N}$ we can write:
\begin{equation}\label{8.30}
 L^0_t(x,y)\,=\,\int_{\mathbb{R}^d}e^{i\varphi_t}\big(^tM\big)^kg^0_t\,\dbar\eta,\quad\forall(t,x,y)\in D.
\end{equation}
\end{itemize}
It clearly follows that for $j\in\{0,1\}$ and for any $(\alpha,\beta)\in[\mathbb{N}^d]^2$, the distributions $\partial^j_t\partial^\alpha_x\partial^\beta_yL^0_t$ are continuous functions on the domain $D$. Moreover, using Lemma \ref{L8.3} we can also deduce inequalities of the form \eqref{8.21} and \eqref{8.22} for $L^0_t$. In order to verify similar inequalities for $K_t$ we use the estimations 
$$
\underset{|t|\leq T_0}{\sup}\ \underset{(x,y)\in\mathbb{R}^{2d}}{\sup}\left|(x-y)^\alpha L^\infty_t(x,y)\right|<\infty
$$ 
that we obtained above, noticing that for any $T_0\in(0,T)$ we have that
\begin{equation}\label{8.31}
 |x-y|\,-\,C_0\,\leq\,|\nabla_\eta\varphi_t(x,y,\eta)|\,\leq\,|x-y|\,+\,C_0,\quad\forall(x,y)\in\mathbb{R}^{2d},\,\forall\eta\in\mathbb{R}^d\setminus\{0\},\,\forall t\in[-T_0,T_0]
\end{equation}
with 
\begin{equation}\label{8.32}
C_0\,:=\,\underset{x\in\mathbb{R}^d}{\sup}\ \underset{\eta\in\mathbb{R}^d\setminus\{0\}}{\sup}\ \underset{t\in[-T_0,T_0]}{\sup}\left|\nabla_\eta d^0_t(x,\eta)\right|.
\end{equation}
The derivatives of $L_t(x,y)=\omega^B(y,x)K_t(x,y)$ can be estimated in a very similar way.
\end{proof}

\begin{remark}\label{R8.5}
 Let $T_0\in(0,T)$ and $C_0$ defined by \eqref{8.32}. Due to the first inequality in \eqref{8.31} we deduce that for $|x-y|\geq2C_0$ and for any $|t|\leq T_0$ we have that
\begin{equation}\label{8.33a}
 d_\varphi(t;x,y)\,\geq\,(1/4)|x-y|\,+\,(1/2)C_0\,\geq\,C_0,
\end{equation}
and thus using Lemma \ref{L8.3} we conclude that $(t,x,y)\in D_{C^{-1}C_0}\subset D$. From \eqref{8.22} we deduce then that for any $k\in\mathbb{N}$ there exists $C_k>0$ such that for any $(x,y)\in\mathbb{R}^{2d}$ with $|x-y|\geq2C_0$:
\begin{equation}\label{8.33}
 \left|K_t(x,y)\right|\,\leq\,C_k<x-y>^{-k},\quad\forall t\in[-T_0,T_0].
\end{equation}
\end{remark}

\begin{remark}\label{R8.6}
 In the case when $a_0=a_0(\xi)$, the last two conditions in Hypothesis \ref{Hyp-a} imply that $\nabla a_0(\xi)\ne0$ for any $\xi\in\mathbb{R}^d\setminus\{0\}$. In this case we evidently can choose $T=\infty$ and we have $x_0(t;y,\eta)=y+t\big(\nabla a_0\big)(\eta)$ and also
\begin{equation}\label{8.34}
 D\,=\,\left\{(t,x,y)\in\mathbb{R}^{2d+1}\ ;\ x-y\ne t\nabla a_0(\eta),\ \forall|\eta|=1\right\},\quad d(t;x,y)=\underset{|\eta|=1}{\inf}|x-y-t\nabla a_0(\eta)|.
\end{equation}
In particular, for the relativistic Schr\"{o}dinger operator we have $a_0(\xi)=|\xi|$ so that
\begin{equation}\label{8.35}
 D\,=\,\left\{(t,x,y)\in\mathbb{R}^{2d+1}\ ;\ |x-y|^2\ne t^2\right\},\quad d(t;x,y)=\underset{|\eta|=1}{\inf}|x-y-t\eta|.
\end{equation}
\end{remark}

\subsection{Propagation of singularities under the evolution group}

We keep the hypothesis and notations from Subsection \ref{Ev-Gr}. In order to formulate the main result of this subsection let us recall the notion of {\it wave front set} of a distribution, microlocalized form of the {\it singular support}.
\begin{definition}\label{D8.7}
 For a distribution $u\in\mathcal{D}^\prime(\mathbb{R}^d)$, the {\it wave front set} $WFu$ is the closed conic subset of $\mathbb{T}^*\mathbb{R}^d\setminus0$ defined by:\\
{\it a point $(x_0,\xi_0)\in\mathbb{T}^*\mathbb{R}^d\setminus0$ does not belong to $WFu$ if and only if there exists an open neighborhood $V$ of $x_0$ in $\mathbb{R}^d$ and an open cone $\Gamma\subset\mathbb{R}^d\setminus\{0\}$ containing $\xi_0$ such that for any $\varphi\in C^\infty_0(V)$ and any $n\in\mathbb{N}$ the function $<\xi>^n\widehat{\varphi u}(\xi)$ is bounded on $\Gamma$.}
\end{definition}

If $\Pi:\mathbb{T}^*\mathbb{R}^d\setminus0\rightarrow\mathbb{R}^d$ is the canonical projection, we have that $\Pi\big(WFu\big)=\text{\sf sing supp }u$.

We still recall the following things. In Lemma 7.40 of \cite{IMP2} it is proved that for any $t\in\mathbb{R}$ the operator $G_t:=e^{-itP}$ belongs to $\mathbb{B}\big(\mathcal{S}(\mathbb{R}^d)\big)$. By duality we deduce that we also have that $G_t\in\mathbb{B}\big(\mathcal{S}^\prime(\mathbb{R}^d)\big)$. The Hamiltonian flow $\Phi^0_t$ defined by $a_0$ is well defined for any $t\in\mathbb{R}$.
\begin{theorem}\label{T8.8}
 Suppose that $a$ is a symbol satisfying Hypothesis \ref{Hyp-a} and $P$ is the self-adjoint operator in $L^2(\mathbb{R}^d)$ defined by $\mathfrak{Op}^A(a)$. Then for any $t\in\mathbb{R}$ and any $u\in\mathcal{S}^\prime(\mathbb{R}^d)$ the following equality holds:
$$
WF\big(e^{-itP}u\big)\,=\,\Phi^0_t\big(WFu\big).
$$
\end{theorem}
\begin{proof}
 We begin by considering $t$ fixed in the interval $(-T,T)$ for $T>0$ defined in Theorem \ref{T7.4}. Then $e^{-itP}=G_t=\mathfrak{Op}^A_{\Phi_t}(g_t)$ with $g_t\in S^0(\mathbb{R}^d)$. As in the proof of Theorem \ref{T8.4} we write $g_t=g^0_t+g^\infty_t$ with $g^0_t(x,\eta):=\chi_R(\eta)g_t(x,\eta)$ and $g^\infty_t(x,\eta):=(1-\chi_R(\eta))g_t(x,\eta)$, and $R>0$ given by Lemma \ref{L1.34}. We have the associated decomposition $G_t=G^0_t+G^\infty_t$.

Using the evident fact that $g^\infty_t\in S^{-\infty}(\mathbb{R}^d)$ and the Lemmas \ref{L1.16} and \ref{L2.9} we deduce that the distribution kernel $K^\infty_t$ of $G^\infty_t$ is of class $C^\infty\big(\mathbb{R}^d;\mathcal{S}(\mathbb{R}^d)\big)$. It follows that the map
$$
\mathbb{R}^d\ni x\mapsto\big(G^\infty_tu\big)(x)\,=\,<u(\cdot),K^\infty_t(x,\cdot)>\in\mathbb{C}
$$
is of class $C^\infty(\mathbb{R}^d)$ and in consequence $WF\big(G^\infty_tu\big)=\emptyset$. We conclude that
\begin{equation}\label{8.37}
 WF\big(G_tu\big)\,=\,WF\big(G^0_tu\big).
\end{equation}

We fix now an arbitrary $\rho>0$ and choose $r>0$ such that $r>2\rho+2C_0$ with $C_0$ defined by \eqref{8.32}. Due to the support condition for $\chi_\rho$ we have
\begin{equation}\label{8.38}
 \left.WF\big(G^0_tu\big)\right|_{\{|x|<\rho\}}\,=\, \left.WF\big((1-\chi_\rho)G^0_tu\big)\right|_{\{|x|<\rho\}}.
\end{equation}
Moreover, with our choice for $R>0$ we have 
\begin{equation}\label{8.39}
 G^0_t\,=\,\mathfrak{Op}^A_{\Phi_t}(g^0_t)\,=\,\mathfrak{Op}^A_{\Phi^0_t}(g^0_t).
\end{equation}
It follows that the distribution kernel of $(1-\chi_\rho)G^0_t\chi_r$ is given by
\begin{equation}\label{8.40}
 N(x,y)\,:=\,\omega^B(x,y)\big(1-\chi_\rho(x)\big)L^0_t(x,y)\chi_r(y),\quad\forall(x,y)\in\mathbb{R}^{2d},
\end{equation}
with $L^0_t$ defined by \eqref{8.26}. On the support of $N$ we have that $|x|\leq2\rho$ and $|y|\geq r$ so that $|x-y|>2C_0$. The Remark \ref{R8.5} implies that $N$ is a distribution of class $C^\infty_0\big(\mathbb{R}^d;\mathcal{S}(\mathbb{R}^d)\big)$ and thus $(1-\chi_\rho)G^0_t\chi_ru$ belongs to $C^\infty_0(\mathbb{R}^d)$. In conclusion
\begin{equation}\label{8.41}
 WF\big((1-\chi_\rho)G^0_tu\big)\,=\,WF\big((1-\chi_\rho)G^0_t(1-\chi_r)u\big),
\end{equation}
and $(1-\chi_r)u\in\mathcal{E}^\prime(\mathbb{R}^d)$ (the distributions with compact support).

Due to \eqref{8.26}, the distribution $L^0_t$ is the kernel of a standard FIO associated to the homogeneous canonical transformation $\Phi^0_t$. Using Theorem 8.1.9 in \cite{Ho} and the notation \eqref{1.prime}, we deduce that 
$$
WF^\prime\big(L^0_t\big)\,\subset\,\text{\sf graph }\Phi^0_t,\qquad WF^\prime\big((1-\chi_\rho)\omega^BL^0_t\big)\,\subset\,\text{\sf graph }\Phi^0_t.
$$
But $(1-\chi_\rho)\omega^BL^0_t$ is the distribution kernel of the operator $(1-\chi_\rho)G^0_t$, so that due to Theorem 8.2.13 in \cite{Ho} we also have that
\begin{equation}\label{8.42}
 WF\big((1-\chi_\rho)G^0_t(1-\chi_r)u\big)\,\subset\,\Phi^0_t\big[WF\big((1-\chi_r)u\big)\big]\,\subset\,\Phi^0_t\big(WFu\big).
\end{equation}

Putting together \eqref{8.37}, \eqref{8.38}, \eqref{8.41} and \eqref{8.42} we deduce that
$$
\left.WF\big(G_tu\big)\right|_{\{|x|<\rho\}}\,\subset\,\Phi^0_t\big(WFu\big),\quad\forall\rho>0,
$$
and we may conclude that
\begin{equation}\label{8.43}
 WF\big(G_tu\big)\,\subset\,\Phi^0_t\big(WFu\big),\quad\forall t\in(-T,T).
\end{equation}
In order to obtain the inverse inclusion we use the identity $u=G_{-t}\big(G_tu\big)$ and deduce that $WFu\,\subset\,\Phi^0_{-t}\big(WF\big(G_tu\big)\big)$, i.e. $\Phi^0_t\big(WFu\big)\,\subset\,WF\big(G_tu\big)$. We obtain that the equality in the conclusion of the Theorem is true for any $t\in(-T,T)$. To prove it for an arbitrary $t\in\mathbb{R}$ we choose $n\in\mathbb{N}$ such that $t/n\in(-T,T)$ and use the group property of the family $\{\Phi^0_t\}_{t\in\mathbb{R}}$:
$$
WF\big(e^{-itP}u\big)\,=\,WF\big(e^{-i(t/n)P}\cdot\ldots\cdot e^{-i(t/n)P}u\big)\,=\,\Phi^0_{(t/n)}\circ\ldots\circ\Phi^0_{(t/n)}\big(WFu\big)\,=\,\Phi^0_t\big(WFu\big).
$$
\end{proof}

\subsubsection*{Acknowledgements}
R. Purice has been supported by CNCSIS
under the Ideas Programme, PCCE project no. 55/2008 {\it “Sisteme
diferentiale in analiza neliniara si aplicatii”}.

E-mail:Viorel.Iftimie@imar.ro, Radu.Purice@imar.ro


\begin{thebibliography}{00}
%.............................................................................................

\bibitem{Bo} J. M. Bony: {\it Fourier Integral Operators and Weyl-H\"ormander Calculus}, Journ\'{e}es \'{E}quations aux D\'{e}riv\'{e}es Partielles (1994), p. 1--14, Expos\'{e} IX.

\bibitem{DG} J. Derezinski, C. G\'{e}rard: {\it Scattering Theory of Classical and Quantum $N$-Particle Systems}. Texts and Monographs in Physics. Springer-Verlag, Berlin, 1997.

\bibitem{Ho} L. H\"ormander: {\it The Analysis of Linear Partial Differential Operators,I, III}, \;Springer-Verlag, New York, 1983, 1985.

\bibitem{Ic} T. Ichinose: {\it Essential selfadjointness of the Weyl quantized relativistic Hamiltonian}, Ann. Inst. H. Poincar\'{e} Phys. Th\'{e}or. 51 (1989), no. 3, 265--297. 

\bibitem{IMP1} V. Iftimie, M. M\u antoiu and R. Purice: {\it Magnetic
Pseudodifferential Operators}, Publ. RIMS. {\bf 43}, 585--623 (2007).

\bibitem{IMP2} V. Iftimie, M. M\u antoiu and R. Purice: {\it Commutator criteria for magnetic pseudodifferential operators}. Comm. Part. Diff. Eq. {\bf 35}, 1058--1094 (2010).

\bibitem{IMP3} V. Iftimie, M. M\u antoiu and R. Purice: {\it Unicity of the Integrated Density of
States for Relativistic Schr\"{o}dinger Operators with Regular Magnetic Fields and
Singular Electric Potentials}. Integral Eq. Op. Th., DOI: 10.1007/s00020-010-1777-8 (2010).

\bibitem{KO1} M.V. Karasev and T.A. Osborn: {\it Symplectic Areas,
Quantization and Dynamics in Electromagnetic Fields}, J. Math. Phys. {\bf 43} (2002), 756--788.

\bibitem{KO2} M.V. Karasev and T.A. Osborn: {\it Quantum Magnetic Algebra and Magnetic Curvature}, J. Phys.A
{\bf 37} (2004), 2345--2363.

\bibitem{Ku} H. Kumano-go: {\it A Calculus of Fourier Integral Operators on $R^{n}$ and the Fundamental Solution for an Operator of Hyperbolic Type},  Comm. Part. Diff. Eq.  \textbf{1}, 1--44  (1976).

\bibitem{MP} M. M\u antoiu and R. Purice, {\it The Magnetic
Weyl Calculus}, J. Math. Phys. {\bf 45}, No 4 (2004), 1394--1417.

\bibitem{Mu} M. M\"uller: {\it Product Rule for Gauge Invariant Weyl Symbols
and its Application to the Semiclassical Description of Guiding
Center Motion}, J. Phys. A: Math. Gen. {\bf 32}, 1035--1052 (1999).

\bibitem{Ro} D. Robert: {\it Autour de l'approximation semi-classique}, Birkh\"{a}user  1983.

\end{thebibliography}
\end{document}